\documentclass[10pt,journal,compsoc]{IEEEtran}
\ifCLASSOPTIONcompsoc
  \usepackage[nocompress]{cite}
\else
  \usepackage{cite}
\fi

\usepackage[noend]{algorithmic}
\usepackage{algorithm}
\usepackage{amsmath,amssymb,epsfig,latexsym,amssymb,graphics,subfigure}
\usepackage{booktabs} % For formal tables
\usepackage{graphicx}
\usepackage{url}
\usepackage{xcolor}
\usepackage{xspace}
\usepackage{caption}
\usepackage{enumitem}
\usepackage{subfigure}
\usepackage{amsthm}
\usepackage{dsfont}
\usepackage{colortbl}
\usepackage{multirow}
\usepackage{tabularx} % in the preamble
\usepackage{longtable}
\usepackage{epstopdf}
\usepackage{balance}

\usepackage{atbegshi}% http://ctan.org/pkg/atbegshi
\AtBeginDocument{\AtBeginShipoutNext{\AtBeginShipoutDiscard}\addtocounter{page}{-1}}

\usepackage{hyperref}
\usepackage{academicons}
\definecolor{orcidlogocol}{HTML}{A6CE39}

\usepackage{newfloat}
\usepackage{listings}
\floatstyle{ruled}
\newfloat{listing}{tb}{lst}{}
\floatname{listing}{Listing}

\usepackage{balance}  % for  \balance command ON LAST PAGE  (only there!
\definecolor{Gray}{gray}{0.85}

\newcounter{example}[section]
\renewcommand{\theexample}{\nthesection.\arabic{example}}
\newenvironment{example}{
     \refstepcounter{example}
     {\vspace{1ex} \noindent\bf  Example  \theexample:}}{
     \eop\vspace{1ex}} %\hspace*{\fill}\vspace*{1ex}}

\newcounter{definition}[section]
\renewcommand{\thedefinition}{\nthesection.\arabic{definition}}
\newenvironment{definition}{
     \refstepcounter{definition}
     {\vspace{1ex} \noindent\bf  Definition  \thedefinition:}}{
     \eop\vspace{1ex}} %\hspace*{\fill}\vspace*{1ex}}

\newcounter{theorem}[section]
\renewcommand{\thetheorem}{\nthesection.\arabic{theorem}}
\newenvironment{theorem}{\begin{em}
        \refstepcounter{theorem}
        {\vspace{1ex} \noindent\bf  Theorem  \thetheorem:}}{
        \end{em}\vspace{0.5ex}} %\hspace*{\fill}\vspace*{1ex}}

\newcounter{lemma}[section]
\renewcommand{\thelemma}{\nthesection.\arabic{lemma}}
\newenvironment{lemma}{\begin{em}
        \refstepcounter{lemma}
        {\vspace{1ex}\noindent \bf Lemma \thelemma:}}{
        \end{em}\vspace{1ex}} %\hspace*{\fill}\vspace*{1ex}}

\newcommand{\nthesection}{\arabic{section}}
\newcommand{\eop}{\hspace*{\fill}\mbox{$\Box$}}
\newcommand{\stitle}[1]{\vspace{1ex} \noindent{\bf {#1}}}

\newcommand{\sstitle}[1]{\vspace{1ex} \noindent{\textit{ #1}}}

\newcommand{\kw}[1]{{\ensuremath {\mathsf{#1}}}\xspace}

\newcommand{\kwnospace}[1]{{\ensuremath {\mathsf{#1}}}}

\newcommand{\ei}{\end{itemize}}
\newcommand{\ee}{\end{enumerate}}

\newcommand{\beqn}{\begin{eqnarray*}}
\newcommand{\eeqn}{\end{eqnarray*}}

\newcounter{ccc}

\newcommand{\eat}[1]{}

\newfloat{tcm}{thp}{loa}
\floatname{tcm}{Recursive \sql}

\def\subfigcapskip{2pt}
\newcommand{\sql}{{\sc sql}\xspace}

\long\def\comment#1{}

\definecolor{lgray}{gray}{0.85}
\definecolor{llgray}{gray}{0.9}

\newcommand{\reffig}[1]{Figure~\ref{fig:#1}}
\newcommand{\refsec}[1]{Section~\ref{sec:#1}}
\newcommand{\reftable}[1]{Table~\ref{tab:#1}}
\newcommand{\refalg}[1]{Algorithm~\ref{alg:#1}}

\newcommand{\refdef}[1]{Definition~\ref{def:#1}}

\newcommand{\reflem}[1]{Lemma~\ref{lem:#1}}

\newcommand{\topcaption}{%
 \setlength{\abovecaptionskip}{3pt}%
 \setlength{\belowcaptionskip}{1pt}%
 \caption}

\newcommand{\baseline}{\kw{Baseline}}

\newcommand{\mbcenum}{\kw{MBCEnum}}
\newcommand{\mbcenump}{\kw{MBCEnum^*}}
\newcommand{\mbcenumu}{\kw{MBCEnumUtil}}

\newcommand{\mbcs}{\kw{MBCSear}}
\newcommand{\mbcsp}{\kw{MBCSear^*}}
\newcommand{\mbcsu}{\kw{MBCSearUtil}}
\newcommand{\mbcsup}{\kw{MBCSearUtil^*}}
\newcommand{\mbcss}{\kwnospace{MBCSear}\textrm{-}\kw{SSP}}
\newcommand{\mbcssp}{\kwnospace{MBCSear}\textrm{-}\kw{SSP^*}}

\newcommand{\vertexreductionv}{\kw{VertexReduction^+}}
\newcommand{\edgereductionv}{\kw{EdgeReduction^+}}

\renewcommand{\baselinestretch}{0.98}
\begin{document}

\title {Balanced Clique Computation in Signed Networks: Concepts and Algorithms}


\author{Zi Chen, Long Yuan$^{*}$, Xuemin Lin,  Lu Qin, Wenjie Zhang}

\thanks{* Zi Chen and Long Yuan are the joint first authors. Long Yuan is the corresponding author.}

\IEEEcompsocitemizethanks{

\IEEEcompsocthanksitem Z. Chen is with the Software Engineering Institute, East China Normal University, Shanghai, China.\protect\\
E-mail: zchen@sei.ecnu.edu.cn.

\IEEEcompsocthanksitem L. Yuan is with the School of Computer Science and Engineering, Nanjing University of Science and Technology, Nanjing, China.\protect\\
E-mail: longyuan@njust.edu.cn.

\IEEEcompsocthanksitem X. Lin and W. Zhang are with University of New South Wales, Sydney, Australia.\protect\\
E-mail: \{lxue,zhangw\}@cse.unsw.edu.au

\IEEEcompsocthanksitem L. Qin is with Centre for QCIS, University of Technology, Sydney, Australia.\protect\\
E-mail: lu.qin@uts.edu.au.

%\IEEEcompsocthanksitem X. Zhao is with National University of Defense Technology, Changsha, China.\protect\\
%E-mail: xiangzhao@nudt.edu.com.

}

%\thanks{Manuscript received xxx; revised xxx.}

%\markboth{IEEE TRANSACTIONS ON KNOWLEDGE AND DATA ENGINEERING, xxx, xxx, xxx}%
{CHEN \MakeLowercase{\textit{et al.}}:  BALANCED CLIQUE COMPUTATION IN SIGNED NETWORKS: CONCEPTS AND ALGORITHMS}

\IEEEtitleabstractindextext{%
\begin{abstract}

Clique is one of the most fundamental models for cohesive subgraph mining in network analysis.  Existing clique model mainly focuses on  unsigned networks. However, in real world, many applications are modeled as signed networks with positive and negative edges. As the signed networks hold their own properties different from the unsigned networks, the existing clique model is inapplicable for the signed networks. Motivated by this, we propose the balanced clique model that considers the most fundamental and dominant theory, structural balance theory, for signed networks. Following the balanced clique model, we study the \underline{m}aximal \underline{b}alanced \underline{c}lique \underline{e}numeration problem (\kw{MBCE})  which computes all the maximal balanced cliques in a given signed network. Moreover, in some applications, users prefer a unique and representative balanced clique with maximum size rather than all balanced cliques. Thus, we also study the \underline{m}aximum \underline{b}alanced \underline{c}lique \underline{s}earch problem (\kw{MBCS}) which computes the balanced clique with maximum size. We show that \kw{MBCE} problem and \kw{MBCS} problem are both NP-Hard. For the \kw{MBCE} problem, a straightforward solution is to treat the signed network as two unsigned networks and leverage the off-the-shelf techniques for unsigned networks. However, such a solution is inefficient for large signed networks. To address this problem, in this paper, we first propose a new maximal balanced clique enumeration algorithm by exploiting the unique properties of signed networks. Based on the new proposed algorithm, we devise two optimization strategies to further improve the  efficiency of the enumeration. For the \kw{MBCS} problem, we first propose a baseline solution. To overcome the huge search space problem of the baseline solution, we propose a new search framework based on search space partition. To further improve the efficiency of the new framework, we propose multiple optimization strategies regarding to redundant search branches and invalid candidates. We conduct extensive experiments on large real datasets. The experimental results demonstrate the efficiency, effectiveness and scalability of our proposed algorithms for \kw{MBCE} problem and \kw{MBCS} problem.

\end{abstract}

\begin{IEEEkeywords}
Balanced Clique, Structural Balance Theory, Signed Network, Graph Algorithm
\end{IEEEkeywords}}

\maketitle
\IEEEdisplaynontitleabstractindextext
\IEEEpeerreviewmaketitle
%\iffalse
%\bibliography{./reference.bib}
%\fi
\IEEEraisesectionheading{\section{Introduction}}

\IEEEPARstart{W}{ith} the proliferation of graph applications, research efforts have been devoted to many fundamental problems in analyzing graph data \cite{DBLP:conf/dasfaa/OuyangYZQL18,DBLP:journals/tkde/YuanQZCY18,DBLP:journals/pvldb/YuanQLCZ17,DBLP:journals/dpd/FengCLQZY18,DBLP:conf/www/LiuYLQZZ19,DBLP:conf/dasfaa/WuYLYZ19,DBLP:conf/dasfaa/QingYZQLZ18,DBLP:journals/pvldb/ouyang2020}. Clique is one of the most fundamental cohesive subgraph models in graph analysis, which requires each pair of vertices has an edge. Due to the completeness requirement,  clique model owns many interesting cohesiveness properties, such as the distance of any two vertices in a clique is one, every one vertex in a clique forms a dominate set of the clique and the diameter of a clique is one  \cite{DBLP:journals/eor/PattilloYB13}. As a result, clique model has wide application scenarios in social network mining, financial analysis and computational biology and has been extensively investigated for decades. Existing studies on clique mainly focus on the unsigned networks, i.e., all the edges in the graph share the same property \cite{bron1973finding,eppstein2010listing,eppstein2011listing,DBLP:journals/vldb/YuanQLCZ16}. Unfortunately, relationships between two entities in many real-world applications have completely opposite properties, such as friend-foe relationships between users in social networks \cite{easley2010networks,DBLP:conf/icdm/KumarSSF16}, support-dissent opinions in opinion networks \cite{kunegis2009slashdot}, trust-distrust relationships in trust networks \cite{leskovec2010signed} and partnership-antagonism in protein-protein interaction networks \cite{ou2015detecting}.  Modelling these applications as signed networks with positive and negative edges allows them to capture more sophisticated semantics than unsigned networks \cite{cartwright1956structural,marvel2009energy,abell2009structural,leskovec2010signed,marvel2011continuous,derr2018signed}. Consequently, existing studies on clique ignoring the sign associated with each edge may be inappropriate to characterize the cohesive subgraphs in a signed network and there is an urgent need to define an exclusive clique model tailored for the signed networks.

\begin{figure}[t]
\begin{center}
\subfigure[$G$]{
\includegraphics[width=0.245\columnwidth]{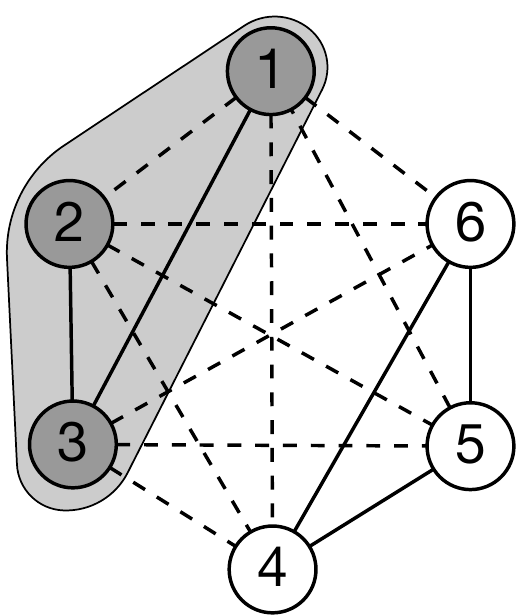}
}\hspace{1.2cm}\vspace{-0.1cm} \subfigure[$G'$]{
\includegraphics[width=0.245\columnwidth]{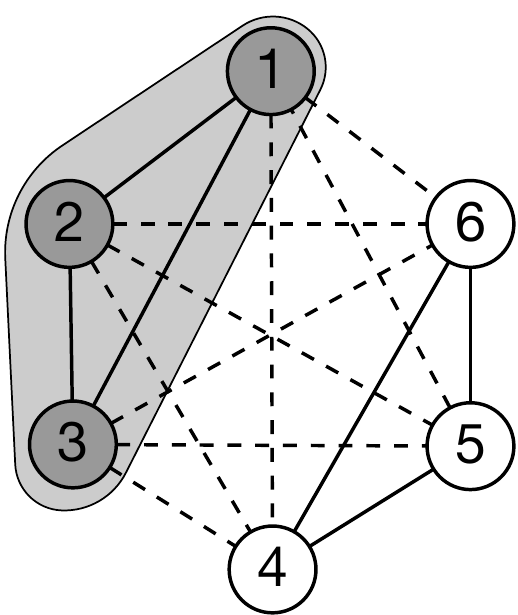}
}
\end{center}
\topcaption{Imbalanced Graph and  Balanced Graph}
\label{fig:balanced}
\vspace{-0.4cm}
\end{figure}

For the signed networks, the most fundamental and dominant theory revealing the dynamics and construction of the signed networks is the \emph{structural balance theory} \cite{heider1946attitudes,harary1953notion,cartwright1956structural,marvel2009energy,abell2009structural,easley2010networks,leskovec2010signed,marvel2011continuous,derr2018signed}. The intuition underlying the structural balance theory can be described as the aphorisms: ``The friend (resp. enemy) of my friend (resp. enemy) is my friend, the friend (resp. enemy) of my enemy (resp. friend) is my enemy''. Specifically, a signed network $G$ is structural balanced if $G$ can be split into two subgraphs such that the edges in the same subgraph are positive and the edges between subgraphs are negative \cite{harary1953notion}.  In a signed network, an imbalanced sub-structure is unstable and tends to evolve into a balanced state. Consider the graph $G$ shown in \reffig{balanced} (a). The negative edge between $v_1$ and $v_2$ makes $G$ imbalanced. $v_1$ and $v_2$ have a mutual ``friend'' $v_3$ and mutual ``enemies'' $v_4$, $v_5$ and $v_6$. It means $v_1$ and $v_2$ share more common grounds than differences. According to \emph{structural balance theory}, $v_1$ and $v_2$ tend to be allies as time goes by. $G'$  shown in  \reffig{balanced} (b) is the evolved balanced counterpart of $G$. In $G'$, the sign of the edge between $v_1$ and $v_2$ becomes positive. $\{v_1, v_2, v_3\}$ and $\{v_4, v_5, v_6\}$ form two alliances and the edges in the same alliance are positive and the edges connecting different alliances are negative. As illustrated in this example, structural balance reflects the key characteristics of the signed networks.

According to the above analysis, clique model is a fundamental cohesive subgraph model in graph analysis, but there is no appropriate counterpart in the signed networks. Meanwhile, the structure of the signed networks is expected to be balanced based on the structure balance theory. Motivated by this, we propose a maximal balanced clique model in this paper. Formally, given a signed network $G$, a maximal balanced clique $C$ is a maximal subgraph of $G$ such that (1) $C$ is complete, i.e., every pair of vertices in $C$ has an edge. (2) $C$ is balanced, i.e.,  $C$ can be divided into two parts such that the edges in the same part are positive and the edges connecting two parts are negative. This definition not only catches the essence of the clique model in the unsigned networks but also guarantees that a detected clique is stable in the signed networks. In this paper, we aim to devise efficient algorithms to enumerate all maximal balanced cliques in a given signed network.

% \begin{figure}
% \begin{center}
% \includegraphics[width=0.9\columnwidth]{fig/voter.pdf}
% \end{center}
% \topcaption{A Partial Opinion Network from Twitter}
% \label{fig:graph}
% \vspace{-0.5cm}
% \end{figure}

Moreover, in real signed networks, the number of maximal balanced cliques could be extremely large. For instance, in  "Douban" network which is a Chinese score service website, there are more than a million balanced cliques in it. However, in some applications, users prefer a unique and representative balanced clique with maximum size rather than all balanced cliques. Maximum clique search problem is a fundamental and hot research topic in graph analysis. In the literature, numerous studies have been conducted, such as maximum clique search \cite{DBLP:journals/vldb/Chang20,DBLP:journals/pvldb/LuYWZ17}, maximum quasi-clique search \cite{DBLP:conf/aaai/ChenCPWLZY21}, maximum bi-clique search\cite{DBLP:journals/pvldb/LyuQLZQZ20}, k*-partite clique with maximum edges\cite{DBLP:journals/pvldb/ZhouW020}, clique with maximum edge/vertex weight on weighted graph\cite{DBLP:journals/soco/SevincD20,DBLP:journals/access/LiWLWY18}. Motivated by this, we aim to devise a maximum balanced clique search algorithm to find out the balanced clique with maximum vertex size, which can scale to  large-scale  real signed networks (with more than 100 million edges).  
%}

\stitle{Applications.}  Balanced clique computation can be used in many applications, for example:

\sstitle{(1) Opinion leaders detection in opinion networks.} Opinion leaders are people who are active in a community  capturing the most representative opinions in the social networks \cite{song2007identifying}. In an opinion network, each vertex represents a user and there is a positive/negative edge between two vertices if one user support/dissent another user. A maximal balanced clique in an opinion network represents a group of users, such that these users actively involve in the opinion networks and have their clear standpoints. Hence, the users in the maximal balanced cliques are good candidates of opinion leaders in the opinion network.     	

\sstitle{(2) Finding international alliances-rivalries  groups.} The international relationships between nations can be modeled as a signed network, where each vertex represents a nation, positive and negative edges indicate alliances and rivalries, respectively. Computing the maximal balanced cliques in such networks reveals hostile groups of allied forces\cite{easley2010networks,Axelrod1993}.
%, such as the Allied and Axis power during World War II or the North Atlantic Treaty Organization and the Warsaw Pact during the Cold War.
  We can extend it to find the alliances-rivalries commercial groups among business organizations similarly, such as \{Pepsi, KFC\} vs \{Coke, McDonald\}\cite{Margarita2009}.

%\sstitle{(3) Semantic expansion.} In information retrieval, when a user submits a keyword query, say, "image", he, may also be interested in results related to other keywords such as "pictures," "photo," and so on. If a semantic link network [6] over keywords is available, then we can expand the query by including "positive" and "negative" keywords in the same "semantic maximal balanced clique".

\sstitle{(3) Synonym and antonym groups discovery.} In a word network, each vertex represents a word and there is a positive edge between two synonyms and a negative edge between two antonyms\cite{Miller1995}. In such signed networks, our model can discover synonym groups that are antonymous with each other, such as, \{interior, internal, intimate\} and \{away, foreign, outer, outside, remote\}. These discovered  groups may be further used in applications such as automatic question generation \cite{Vishwajeet19} and semantic expansion \cite{Adit2018}.

\stitle{Contributions.} In this paper,  we make the following contributions:

\sstitle{(1) The first work to study the maximal balanced clique model.} We formalize the balanced clique model in signed networks based on the structural balance theory. To the best of our knowledge, this is the first work considering the structural balance of the cliques in signed networks. We also prove the NP-Hardness of the problem.
	
\sstitle{(2) A new framework tailored for maximal balanced clique enumeration in signed networks.} After investigating the drawbacks of the straightforward approach, we propose a new framework for the maximal balanced clique enumeration. Our new framework enumerates the maximal balanced cliques based on the signed network directly and its memory consumption is linear to the size of the input signed network.
	
	%We firstly propose an intuitive baseline algorithm. To improve the efficiency, we develop a more efficient algorithm following grow-peel framework, which progressively grows one side of community, meanwhile, peels another side of community. We also design an overlap strategy and a greedy strategy during candidate selection. A new paradigm. after investigating the baseline approach.
	
\sstitle{(3) Two effective optimization strategies to further improve the enumeration  performance.} We explore two optimization strategies, in-enumeration optimization and pre-enumeration optimization, to further improve the enumeration performance. The in-enumeration optimization can avoid the exploration for unpromising vertices during the enumeration while the pre-enumeration techniques can prune unpromising vertices and edges before enumeration.  %Remarkably,  the time complexity of the pre-enumeration optimization techniques can be well bounded theoretically and they are efficient and effective in practice as shown in our experiments.
	
	%To reduce the graph size, We propose two useful global pruning techniques, 1-Hop reduction and 2-Hop reduction which consider 1-hop degree of vertices and 2-hop degree of edges separately, to prune unpromise vertices and edges. In addition, We also present several optimization techniques, in detail, we remove useless candidates and prune branches using early terminate condition and special pivot.

%	\textcolor{blue}{
	\sstitle{(4) An efficient maximum balanced clique search algorithm.} To address the maximum balanced clique search problem, we first propose a baseline algorithm. In order to reduce the search space during the search process of baseline, we propose a search space partition-based algorithm \mbcss by partitioning the whole search space into multiple search regions. In each search region, two size thresholds $\overline{\kappa}$ and $\underline{\kappa}$ are used to search the result matching the size requirement specific to this search region, such that the search space is limited into a small area. To further improve the efficiency of \mbcss algorithm, we also explore three optimization strategies to prune invalid search branches and candidates during the search process.
%	they are coloring-based branch pruning and vertex domination-based candidate pruning. Besides, we extend the pre-enumeration techniques of \kw{MBCE} to further prune useless vertices and edges.
%	}
	
%	
%	\sstitle{(5)A new graph structure.} After investigating the urgent need for local structure information and memory problem in our algorithms, we design a new graph structure that can not only support our \kw{MBCE} and \kw{MBCS} algorithms but also directly provide the local structure information required by our algorithms.
	
	\sstitle{(5) Extensive performance studies on  real datasets.} We first evaluate the performance of \kw{MBCE} algorithms by conducting extensive experimental studies on  real datasets. As shown in our experiments, the baseline approach only works on  small datasets while our approach can complete the enumeration efficiently on both small and large datasets. %for the small signed networks, our approach can achieve 1-2 orders of speedup compared with the baseline; for the large singed networks, the baseline fails to output results while our approach can complete the enumeration efficiently.  %The results indicate our Grow-Peel algorithm consistently outperforms the baseline algorithm. And 2-Hop reduction shows powerful effectiveness during global pruning. Besides, the extensive performance evaluation shows our several improvement techniques further enhance the performance of our algorithms. Finally, case study in word networks indicates our model can identify reasonable, cohesive and balance communities in real-world applications.
	 Then, we evaluate the performance of our proposed \kw{MBCS} algorithm. The baseline algorithm can not get the result within a reasonable time on large datasets, while our optimized algorithm shows high efficiency, effectiveness and scalability.

% \footnote{For reproducibility, our code is anonymously released: \url{https://drive.google.com/file/d/18L538NMN7xDO6Xu0CMa2oSYmx07K5vZ5/view?usp=sharing}}

%\stitle{Outline.} \refsec{related} reviews the related work. \refsec{preliminaries} provides preliminaries including the definition of balanced clique model and problem statement. \refsec{baseline} introduces the baseline algorithm.  \refsec{improveapproch} presents our new enumeration framework.  \refsec{optimization} shows several optimization techniques.
%\refsec{mbcs} studies the maximum balanced clique search problem.\refsec{data} proposes a new graph structure.
% \refsec{performance} reports the results of experimental studies.     \refsec{conclusion} concludes our paper.

\vspace{-0.2cm}

\section{Related Work}
\label{sec:related}

\stitle{Signed network analysis.} 
%Signed network analysis has attracted much attention in the literature. In these works, the theories explaining the potential social dynamics process in signed networks have been extensively studied. Among these theories, \emph{structural balance theory} is the most fundamental and dominant one \cite{zheng2015social}.
 Structural balance theory is originally introduced in \cite{heider1946attitudes} and generalized in the graph  formation in \cite{harary1953notion,cartwright1956structural}. After that, structural balance theory is developed extensively \cite{marvel2009energy,abell2009structural,leskovec2010signed,marvel2011continuous,derr2018signed}. 
In these works,  it is interesting to mention that the authors in \cite{marvel2011continuous} model the evolving procedure of a signed network and theoretically prove that the network would evolve into a balanced clique when the  mean value of the initial friendliness among the vertices $\mu \leq 0$. \cite{zheng2015social} provides a comprehensive survey on structural balanced theory.

Besides, a large body of literature on mining signed networks has been emerged.  Among them, the most closely related work to ours  is \cite{li2018efficient} in which an $(\alpha, k)$-clique model is proposed.
% Given a signed network $G$, an $(\alpha, k)$-clique is defined as a maximal clique $C$ such that the negative degree for each vertex in $C$ is not greater than $k$ and the positive degree for each vertex in $C$ is not less than $\alpha k$. 
 Compared with our model,  $(\alpha, k)$-clique model only considers the amount of positive and negative edges in the clique and the structural balance of the clique is totally ignored, which makes $(\alpha, k)$-clique model essentially different from our model. In \cite{hao2014detecting}, a $k$-balanced trusted clique model is proposed.
%  A $k$-balanced trusted clique is defined as a clique with $k$ vertices consisting with positive edges only. 
  Although the $k$-balanced trusted clique model has a similar name with our model, it ignores the negative edges in the clique, which means the information of the negative edges are totally missed.

%Community detection in signed networks is also related to our work. For example, \cite{gao2016detecting,lo2011mining,chu2016finding,nature2017,DBLP:journals/jmlr/LiuTM17,DBLP:journals/pami/LiuXTZ19} 
%aim to find the antagonistic communities in a signed network. These works mainly focus on exploring several groups of dense subgraphs and most of them don't have a clear structural definition of their community model, while our work aims to enumerate the clique structure in a signed network. Moreover, these solutions generally involve a complicated optimization procedure, thereby, they are hard to handle large signed networks, while our proposed algorithm is scalable to enumerate all the maximal balanced cliques in large signed networks with hundreds of millions of edges as verified in our experiments. A  survey on signed network mining can be found in \cite{tang2016survey}.

\stitle{Clique on unsigned networks.} Clique model is one of the most fundamental cohesive subgraph models. \cite{bron1973finding} proposes an efficient algorithm for maximal clique enumeration based on backtracking search.\cite{R2}  first considers the memory consumption during the maximal clique enumeration. Based on \cite{bron1973finding}, more efficient algorithms are investigated \cite{tomita2006worst,eppstein2010listing,eppstein2011listing}.\cite{eppstein2010listing} proposes a novel branch pruning strategy, which can efficiently reduce the search space by ignoring the search process from the neighbors of the pivot. 
%\cite{zhang2014finding} studies the maximal biclique enumeration problem on bipartite graphs.\cite{zhang2014finding} keeps growing the vertex set in one side and peeling the vertex set in another side to enumerate the maximal biciques.  It also utilizes some techniques to further improve the enumeration performance, such as choosing vertex with small degree from candidate set to reduce the search tree depth and pruning vertices which may produce non-maximal bicliques. These techniques for biclique enumeration inspire our techniques presented in \refsec{localpruning}.
\cite{fakhfakh2017algorithms} reviews recently advances in maximal clique enumeration.  Based on clique, other cohesive subgraph models are also studied, such as $k$-core \cite{seidman1983network}, $k$-truss\cite{A6,DBLP:conf/sigmod/HuangCQTY14}, $k$-edge connected component\cite{A7,DBLP:journals/pvldb/YuanQLCZ16,DBLP:journals/vldb/YuanQLCZ17}, and $(r, s)$-nuclei \cite{sariyuce2015finding,DBLP:journals/tweb/SariyuceSPC17}. 
%Note that our balanced clique model is different from the existing cohesive subgraph models on unsigned networks and it cannot be well solved by the existing works. If we just consider the positive edge in the signed network and use the traditional methods on unsigned networks for community detection, the found results would  ignore the negative edges and half meaningful information in the signed network is lost.

%\iffalse
%\bibliography{./reference.bib}
%\fi

\vspace{-0.2cm}

\section{Problem Statement}
\label{sec:preliminaries}
In this paper, we consider an undirected and unweighted signed network $G=(V,E^{+},E^-)$, where $V$ denotes the set of vertices, $E^+$ denotes the positive edges and $E^-$ denotes the negative edges connecting the vertices in $G$. We denote the number of vertices and number of edges by $n$ and $m$, respectively.  For each vertex $v \in G$, let $N^+_G(v)$ represents the positive neighbors of $v$, and let $N^-_G(v)$ represents the negative neighbors of $v$. We use $d^+_G(v)$ and $d^-_G(v)$ to denote the positive and negative degree of $v$, respectively. We also use $N_G(v)$ and $d_G(v)$ to denote the neighbors and degree of $v$, i.e., $N_G(v) = N^-_G(v) \cup N^+_G(v)$ and $d_G(v) = d^+_G(v) + d^-_G(v)$. For simplicity, we omit G in the above notations if the context is self-evident.
 % We use $n$, $m_p$, $m_n$ to denote the number of nodes and edges in $G$ respectively, i.e, $n = |V(G)|$, $m_p = |E^+(G)$|, $m_n = |E^-(G)|$.

\begin{definition}
\label{the:balancegraph}
\textbf{(Balanced Network \cite{harary1953notion})} Given a  signed network $G=(V,E^{+},E^-)$, it's balanced iff it can be split into two subgraphs $G_L$ and $G_R$, s.t. $\forall (u,v) \in E^+ \rightarrow u,v\in G_L$ or $u,v \in G_R$, and $\forall (u,v) \in E^- \rightarrow u\in G_L, v \in G_R$ or $u\in G_R, v \in G_L$.
\end{definition}

\begin{definition}
\label{def:balancecommunity}
\textbf{(Maximal Balanced Clique)} Given a signed network $G=(V,E^{+},E^-)$, a maximal balanced clique $C$ is a maximal subgraph of $G$ that satisfies the following constraints:
\begin{itemize}
	\item Complete: $C$ is  complete, i.e, $\forall u,v \in C \rightarrow (u,v)\in E^+ \cup E^-$.
	\item Balanced: $C$ is  balanced, i.e, it can be split into two  sub-cliques $C_L$ and $C_R$, s.t. $\forall u,v\in C_L$ or $u,v \in C_R   \rightarrow (u,v) \in E^+$, and $\forall u\in C_L, v \in C_R$ or $u\in C_R, v \in C_L  \rightarrow (u,v) \in E^-$.%\item Maximal: There doesn't exist any sub-network $C'$ such that $C'$ is a balanced clique and $C \subseteq C'$.
\end{itemize}
\end{definition}
\vspace{-0.4cm}

\begin{definition}
\label{def:maximumclique}
\textbf{(Maximum Balanced Clique)} Given a signed network $G=(V,E^{+},E^-)$, a maximum balanced clique $C^*$ in $G$ is a balanced clique with the maximum vertex size. 
\end{definition}

Since many real applications require that the number of vertices in $C_L$ and $C_R$ is not less than a fixed threshold, we add a size constraint on $|C_L|$ and $|C_R|$ s.t. $|C_L| \geq k$ and $|C_R| \geq k$. With the size constraint, users can control the size of the returned maximal balanced cliques based on their specific requirements. We formalize the studied problems in the paper as follows:

\stitle{Problem Statement.} Given a signed network $G$ and an integer $k$, 

\begin{itemize}
  \item[$\bullet$] the maximal balanced clique enumeration (\kw{MBCE}) problem aims to compute all the maximal balanced cliques $C$ in $G$ s.t. $|C_L| \geq k$ and $|C_R| \geq k$ for $C$.
  \item[$\bullet$] the maximum balanced clique search (\kw{MBCS}) problem aims to compute the balanced clique $C^*$ in $G$ s.t. $|C^*_L| \geq k$, $|C^*_R| \geq k$ and $|C^*_L|+|C^*_R|$ is maximum.
\end{itemize}

%In this paper, we use a integer parameter $k$ as a size threshold. Given a signed graph $G(V,E^{+},E^-)$ and integer $k$, the problem of \underline{f}inding \underline{m}aximal \underline{b}alance \underline{c}ommunity ($FMBC$) is to find all $MBCs$ such that $|C_L| \geq k$ and $|C_R| \geq k$ for each of $MBCs$. With big $k$, users can discover communities with large size of sub-communities on both sides.

\begin{figure}
\begin{center}
\begin{tabular}[t]{c}
  \includegraphics[width=0.65\columnwidth]{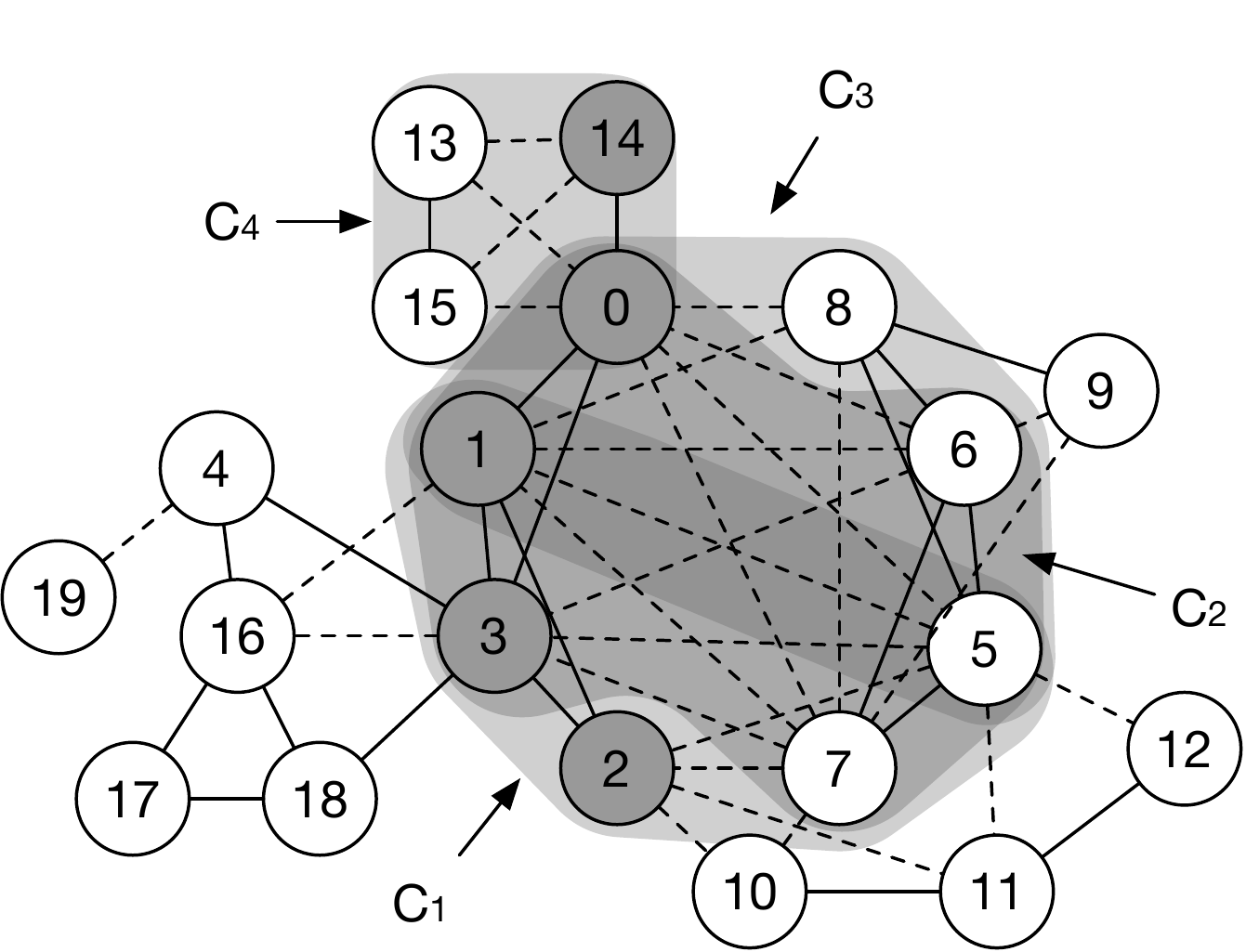}
\end{tabular}
\end{center}
\vspace{-0.2cm}
\topcaption{Maximal Balanced Clique in $G$ ($k=2$)}
\label{fig:community}
\vspace{-0.2cm}
\end{figure}

\begin{example}
\label{exp:preliminary}
Consider the signed network $G$ in \reffig{community} in which positive/negative edges are denoted by solid/dashed lines. Assume  $k$ = 2, there are 4 maximal balanced cliques in $G$, namely, $C_1=\{\{v_1,v_2,v_3\},\{v_5,v_7\}\}$, $C_2=\{\{v_0,v_1,v_3\},\{v_5,v_6,v_7\}\}$, $C_3$ $=\{\{v_0,v_1\},\{v_5,v_6,v_8\}\}$, $C_4=\{\{v_0,v_{14}\},\{v_{13},v_{15}\}\}$, where vertices in $C_L$ and $C_R$ are marked with different colors. Among them, $C_2$ is the maximum balanced clique.
% Take $C_4$ as an example,  it is complete as any two vertices in $C_4$ have an edge. It is balanced as it can be split into $\{v_0, v_{14}\}$ and $\{v_{13}, v_{15}\}$, and two positive edges $(v_0, v_{14})$ and $(v_{13}, v_{15})$ exist in $E^+$. $v_0$ has negative edges to $v_{13}$, $v_{15}$ and similar negative edges exist for $v_{13}$, $v_{14}$ and $v_{15}$. $C_4$ is maximal because no more vertices can be added into it to make it complete and balanced. The maximum balanced clique $C^*$ in $G$ is $C_2$ which has the most vertices.
\end{example}

\stitle{Problem Hardness.}The \kw{MBCE} problem is NP-Hard, which can be proved following the NP-Hardness of maximal clique enumeration problem \cite{cheng2012fast,Schmidt:2009:SPA:1514449.1514768}. Given an unsigned network $G=(V,E)$, we can transfer $G$ to a signed network $G'$ as follows: we first keep  all the vertices of $G$ in $G'$ and all the edges of $G$ as positive edges in $G'$; then, we add a new vertex $v$ to $G'$ and connect  $v$ to all vertices in $G'$ with negative edges. It's clear that each maximal clique $C$ in $G$ corresponds a maximal balanced clique $\{\{v\},C\}$ in $G'$ (assume $k=1$), which means the maximal clique enumeration problem in $G$ can be reduced to the \kw{MBCE} problem in $G'$. As the maximal clique enumeration problem is NP-Hard \cite{cheng2012fast,Schmidt:2009:SPA:1514449.1514768}, our \kw{MBCE} problem is also NP-Hard. 

Similarly, reconsidering $G$ and $G'$, the maximum clique $C^*$ in $G$ corresponds the  maximum balanced clique $\{\{v\},C^*\}$ in $G'$, and vice versa. As the maximum clique search problem is NP-Hard\cite{DBLP:journals/pvldb/LuYWZ17,DBLP:journals/vldb/Chang20}, our \kw{MBCS} problem is also NP-Hard.

\vspace{-0.2cm}
\section{A Baseline Algorithm for MBCE Problem}
\label{sec:baseline}

% \begin{figure}
% \begin{center}
% \begin{tabular}[t]{c}
% \includegraphics[width=1\columnwidth]{fig/baseline.pdf}
% \end{tabular}
% \end{center}
% \topcaption{An Example for Baseline Algorithms, $k$=2}
% \label{fig:baseline}
% \end{figure}

We first propose a baseline algorithm to address \kw{MBCE} problem based on existing methods for maximal clique enumeration \cite{eppstein2011listing} and maximal biclique enumeration \cite{zhang2014finding} in unsigned networks. For a signed network $G = (V, E^+, E^-)$, we can treat it as the combination of two unsigned networks $G^+=(V, E^+)$ and $G^- = (V, E^-)$. For any maximal balanced clique $C = \{C_L,C_R\}$ in $G$, it is clear that $C_L$ (resp. $C_R$) is a clique in $G^+$ and the subgraph induced by vertices in $C_L$ and $C_R$ in $G^-$ is a biclique. Therefore, we can enumerate the maximal balanced cliques in $G$ in two steps: 1) compute all the maximal cliques in $G^+$ with \cite{eppstein2011listing}; 2) for each pair of the computed maximal cliques $C_i$ and $C_j$ in $G^+$, compute the maximal bicliques in the bipartite subgraph induced by the vertices in $C_i$ and $C_j$ in $G^-$ with \cite{zhang2014finding}. The returned maximal bicliques in $G^-$ are the maximal balanced cliques in $G$. 

\stitle{Drawbacks of \kw{baseline}.} Since \baseline does not consider the uniqueness of the signed networks and processes \kw{MBCE} with the techniques for the unsigned networks, it  has two drawbacks:

\begin{itemize}

\item Memory consumption. \baseline has to store all the maximal cliques in $G^+$ in memory. The number of maximal cliques could be  exponential to the number of vertices \cite{eppstein2010listing}, which makes \baseline unable to handle large networks.
\item Efficiency.  In \kw{baseline}, all the maximal cliques in $G^+$ are enumerated and every pair of maximal cliques  are explored. The time complexity of \baseline is $O(T_{\kw{Cli}}+ \eta^2\cdot T_{\kw{BiCli}})$, where $T_{\kw{BiCli}}$/$T_{\kw{Cli}}$ represent the time complexity of maximal (bi)clique enumeration, and  $\eta$ is the number of enumerated maximal cliques in $G^+$. Considering the maximal (bi)clique enumeration is time-consuming and the number of maximal cliques could be very large, it is inefficient for \kw{MBCE} problem.

% since the information of negative edges is missed during maximal clique enumeration, we cannot judge whether current positive clique exactly be the component clique in the related maximal balanced clique, thus, it's unavoidable for the redundant computation of maximal clique enumeration.

%In the enumeration,  two types of unfruitful computations regarding \kw{MBCE} are involved: 1) the enumerated maximal cliques that are not a part of a maximal balanced clique,  2) the explored pair of maximal cliques that are not a part of a maximal balanced clique. Considering the maximal (bi)clique enumeration is time-consuming and the number of maximal cliques could be very large, these two types of unfruitful computations together make \baseline inefficient for \kw{MBCE} problem.
\end{itemize}

\vspace{-0.2cm}

\section{A New Enumeration Framework}
\label{sec:improveapproch}

%In \baseline, the signed network is treated as a specific combination of two unsigned networks and the enumeration utilizes the existing techniques designed for unsigned networks, which leads the drawbacks discussed in \refsec{baseline}.

%The reason that \baseline owns above drawbacks regarding \kw{MBCE} problem is that it treats the signed network as a specific combination of two unsigned networks and  adopts the two-phases enumeration strategy, namely maximal clique enumeration followed by maximal biclique enumeration, by utilizing the existing techniques designed for unsigned networks.  To overcome these drawbacks, we propose a new framework for the \kw{MBCE} problem by considering the uniqueness of signed network. Since our new framework searches the maximal balanced clique on the signed network directly and the two-phases enumeration is discarded,  the drawbacks caused by the two-phases enumeration in \baseline can be totally avoided.  Moreover, since our new framework is tailored for signed networks, more optimization strategies based on the properties of signed networks can be applied.  In this section, we focus on the enumeration framework. In the next section, we present our optimization strategies.

Revisiting \kw{baseline}, the root leading to its drawbacks discussed above is that it treats the signed network as a specific combination of two unsigned networks and utilizes the existing techniques designed for the unsigned networks. Therefore, we have to explore new techniques by considering the uniqueness of signed networks to overcome the drawbacks of \baseline and improve the efficiency of the enumeration. In this section, we  present a new enumeration framework which aims to address the memory consumption problem. In next section, we further optimize the enumeration framework to improve the efficiency.

%To overcome these drawbacks, we propose a new framework for the \kw{MBCE} problem by considering the uniqueness of signed network. Since our new framework searches the maximal balanced clique on the signed network directly and the two-phases enumeration is discarded,  the drawbacks caused by the two-phases enumeration in \baseline can be totally avoided.  Moreover, since our new framework is tailored for signed networks, more optimization strategies based on the properties of signed networks can be applied.  In this section, we focus on the enumeration framework. In the next section, we present our optimization strategies.

%Based on \refdef{balancecommunity}, we have the following lemma:

\begin{lemma}
\label{lem:partial}
Given a signed network $G$, for a balanced clique $C = \{C_L, C_R\}$ in $G$, if there is a vertex $v$ in $G$ such that $\forall u \in C_L \rightarrow (v, u) \in E^+$  and $\forall w \in C_R \rightarrow (v, w) \in E^-$, then $C' = \{C_L \cup \{v\}, C_R\}$ is also a balanced clique in $G$.
\end{lemma}
%
%\begin{proof}
%It can be proved following \refdef{balancecommunity} directly.
%\end{proof}

According to \reflem{partial}, if we maintain a  balanced clique $C=\{C_L, C_R\}$, let $P_L$ be the set of vertices that are positive neighbors of all the vertices in $C_L$ and negative neighbors of all the vertices in $C_R$, let $P_R$ be the set of vertices that are positive neighbors of all the vertices in $C_R$ and negative neighbors of all the vertices in $C_L$, we can enlarge  $C$ by  adding vertices from $P_L$ and $P_R$ into $C_L$ and $C_R$, respectively. Furthermore, if we update the $P_L$ and $P_R$ based on the new   $C_L$ and $C_R$ accordingly and repeat the above enlargement procedure, we can obtain a maximal balanced clique when no more vertices can be added into $C_L$ or $C_R$.

%Before the search procedure, to bound the size of candidates, we first sort vertices by the value of core (regardless of the sign of edges) in increasing order. For vertices with same core value, they're sorted by removed order during core decomposition procedure. For each vertex, only their neighbors with higher order will be considered as candidates.

\begin{algorithm}[t]
\caption{\small \mbcenum($G=(V, E^+, E^-), k$)}
{
\small
\begin{algorithmic}[1]

%\STATE $C_L,C_R$: current clique on left and right side;
%\STATE $P_L,P_R$: candidant set;
%\STATE $Q_L,Q_R$: excluded set;
\STATE $\kw{Flag} \leftarrow$ \textbf{true};%: control the side to grow;
\FOR {\textbf{each} $v_i \in \{v_0,v_1,\cdots,v_{n-1}\} \in V$}
\STATE $C_L \leftarrow \{v_i\}$, $C_R \leftarrow \emptyset$
\STATE $P_L \leftarrow N^+_G(v_i) \cap \{v_{i+1},\cdots,v_{n-1}\}$;
\STATE $P_R \leftarrow N^-_G(v_i) \cap \{v_{i+1},\cdots,v_{n-1}\}$;
\STATE $Q_L \leftarrow N^+_G(v_i) \cap \{v_0,\cdots,v_{i-1}\}$;
\STATE $Q_R \leftarrow N^-_G(v_i) \cap \{v_0,\cdots,v_{i-1}\}$;
\STATE \mbcenumu$(C_L, C_R, P_L, P_R, Q_L, Q_R)$;
\ENDFOR

\vspace{0.1cm}
\STATE \textbf{Procedure} \mbcenumu($C_L, C_R, P_L, P_R, Q_L, Q_R$)
\IF{$P_L=\emptyset$ and $P_R=\emptyset$ and $Q_L=\emptyset$ and $Q_R=\emptyset$}
\IF{$|C_L| \geq k$ and  $|C_R| \geq k$}
\STATE \textbf{output} $C = \{C_L, C_R\}$;
\ENDIF
\STATE \textbf{return}
\ENDIF
\STATE $\kw{Flag} \leftarrow !\kw{Flag}$;
% \STATE $p \leftarrow \kw{argmax}_{v \in P_L \cup P_R \cup Q_L \cup Q_R}\{d_l(v)\}$;
% \STATE \kw{newP_L} $\leftarrow P_L \setminus N^+_G(p)$; \hfill //$~$assume $p$ from $P_L \cup Q_L$
% \STATE \kw{newP_R}$ \leftarrow P_R \setminus N^-_G(p)$;
\IF{$\kw{Flag}$}
\FOR{\textbf{each} $v \in$ \kw{P_L}}
\STATE $\mbcenumu(C_L \cup \{v\}, C_R, N^+_G(v) \cap P_L, N^-_G(v) \cap P_R, N^+_G(v) \cap Q_L, N^-_G(v) \cap Q_R)$;
\STATE $P_L \leftarrow P_L \setminus \{v\}$; $Q_L \leftarrow Q_L \cup \{v\}$;
\ENDFOR
\FOR{\textbf{each} $v \in$ \kw{P_R}}
\STATE $\mbcenumu(C_L, C_R \cup \{v\}, N^-_G(v) \cap P_L, N^+_G(v) \cap P_R, N^-_G(v) \cap Q_L, N^+_G(v) \cap Q_R)$;
\STATE $P_R \leftarrow P_R \setminus \{v\}$; $Q_R \leftarrow Q_R \cup \{v\}$;
\ENDFOR
\ELSE
\STATE line 19-21; line 16-18;
\ENDIF
\end{algorithmic}
}
\label{alg:fmbc}
\end{algorithm}

\stitle{Algorithm of \mbcenum.} Following the above idea, our algorithm for \kw{MBCE} is shown in \refalg{fmbc}. For each vertex $v_i$ in $G$ (line 2), we enumerate all the maximal balanced cliques containing $v_i$ (line 3-8). Note that $v_0, v_1, \dots, v_n$ are in the degeneracy order \cite{wendong} of $G$. We use $C_L$ and $C_R$ to maintain the  balanced clique, which are initialized with $v_i$ and $\emptyset$, respectively (line 3). Similarly, we also initialize $P_L$ and $P_R$ as discussed above (line 4-5).  Moreover, we use $Q_L$ and $Q_R$ to record the vertices that have been processed to avoid outputting duplicate maximal balanced cliques (line 6-7). After initializing these six sets, we invoke procedure \mbcenumu $~$to enumerate all the maximal balanced cliques containing $v_i$ (line 8).

Procedure \kw{MBCEnumUtil} performs the maximal balanced clique enumeration based on the given six sets. If $P_L$, $P_R$, $Q_L$ and $Q_R$ are empty, which means current  balanced clique $C = \{C_L, C_R\}$ cannot be enlarged and it is a maximal balanced clique, \kw{MBCEnumUtil} checks whether $C_L$ and $C_R$ satisfy the size constraint. If the size constraint is satisfied, it outputs the maximal balanced clique $C$ (line 11-12). Otherwise, \kw{MBCEnumUtil} adds a vertex from $P_L$ to $C_L$, updates the corresponding $P_L$, $P_R$, $Q_L$ and $Q_R$, and recursively invokes itself to further enlarge the  balanced clique (line 17). When $v \in P_L$ is processed, $v$ is removed from $P_L$ and added in $Q_L$ (line 18). Similar processing steps are applied on vertices in $P_R$ (line 19-21). Variable \kw{Flag} (line 1) is used to control the order of adding new vertex into $C_L$ or $C_R$. With the switch operation in line 14, we can guarantee that we add vertex into $C_L$, then into $C_R$, recursively.

\begin{figure}[t]
\begin{center}
\includegraphics[width=0.95\columnwidth]{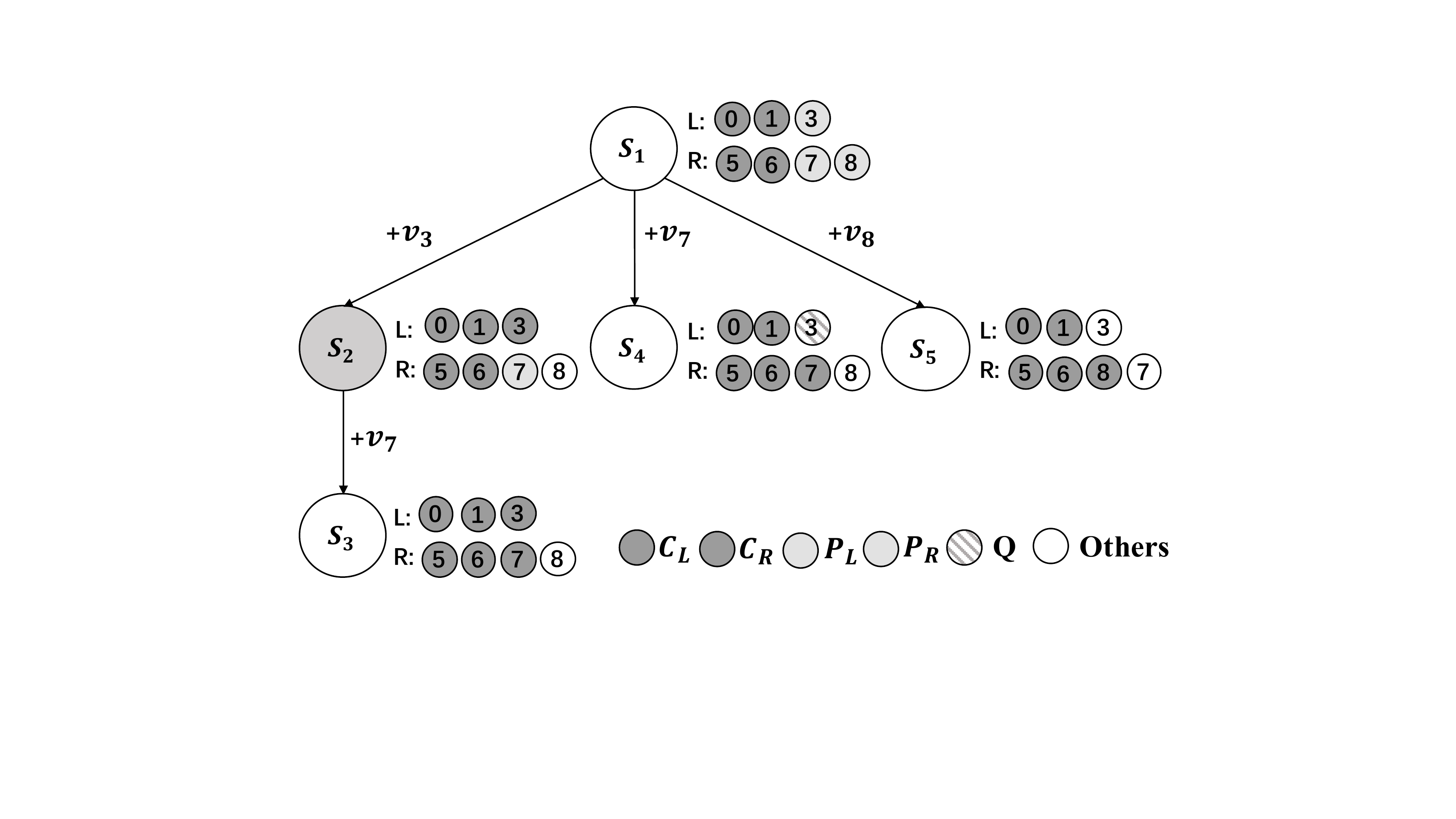}
\end{center}
\vspace{-0.2cm}
\topcaption{Search Tree for \mbcenum}
\label{fig:tree}
\vspace{-0.6cm}
\end{figure}

%\begin{example}
%The enumeration procedure of \kw{MBCEnum} can be illustrated as a search tree.  \reffig{tree} shows part of the search tree when we conduct the \kw{MBCE} on $G$ in \reffig{community} through \mbcenum. $S_1, S_2, \dots$ represent different search states during the enumeration.  At $S_1$, we assume that we have a balanced clique $C = \{C_L=\{v_0, v_1\},$ $ C_R=\{v_5, v_6\}\}$, $P_L$=$\{v_2, v_3\}$, $P_R$=$\{v_7, v_8\}$ at this state. We first grow search branch by adding  $v_2$ from $P_L$ into $C_L$. Since $v_7$ and $v_8$ are not $v_2$'s negative neighbors, they are removed from $P_R$ at $S_2$. Because $P_R$ is empty at $S_2$, we keep expending $C_L$ by adding  $v_3$ from $P_L$. At $S_3$, $P_L$, $R_L$, $Q_L$ and $Q_R$ are empty, we obtain a maximal balanced clique $C_2 = \{\{v_0,v_1,v_2,v_3\},\{v_5,v_6\}\}$ and this search branch starting from $v_2$ finishes. We return back to $S_1$ and $v_2$ is moved to $Q_L$. Then, we add  $v_3$ into $C_L$ at $S_4$ and add $v_7$ into $C_R$ at $S_5$ and obtain $C_1 = \{\{0,1,3\},\{5,6,7\}\}$. The search continues in a similar way until all the vertices in $P_L$ and $P_R$ at $S_1$ are explored.
%\end{example}

\begin{example}
The enumeration procedure of \kw{MBCEnum} can be illustrated as a search tree.  \reffig{tree} shows part of the search tree when we conduct the \kw{MBCE} on $G$ in \reffig{community} through \mbcenum. $S_1, S_2, \dots$ represent different search states during the enumeration.  At $S_1$, we assume that we have a balanced clique with $C_L$=$\{v_0, v_1\}$, $C_R$=$\{v_5, v_6\}\}$, $P_L$=$\{v_3\}$ and $P_R$=$\{v_7, v_8\}$ at this state. We first grow search branch by adding $v_3$ into $C_L$. At $S_2$, $P_L$ is empty, $P_R=\{v_7\}$, hence, we add $v_7$ into $C_R$ at $S_3$ and obtain $C_2 = \{\{v_0,v_1,v_3\},\{v_5,v_6,v_7\}\}$. Now, the search branch from $v_3$ is finished, we return to $S_1$ state. Since $v_3$ has been explored, it is removed to $Q_L$, $P_L$ is empty now.  Then, we add $v_7$ from $P_R$ to $C_R$ at $S_4$. Due to $C_2$ has been found at  $S_3$, current result at $S_4$ is not maximal and can not be output. Next, we return to $S_1$ and find $C_3 = \{\{v_0,v_1\},\{v_5,v_6,v_8\}\}$ by adding $v_8$ into $C_R$. Here, the search procedure at this search tree is finished. Other maximal balanced cliques can be found in a similar way.
\end{example}

%Now we analyse the time and space complexity of \refalg{fmbc}.

%\begin{theorem}
%\label{thm:memory}
%Given a signed network $G$, the memory consumption of \refalg{fmbc} to enumerate the maximal balanced cliques is $O(n+m)$.
%\end{theorem}
%\vspace{-0.2cm}
%\begin{proof}
%In  \refalg{fmbc}, only the input signed network is stored in the main memory and the maximal balanced cliques are enumerated in a depth-first search manner. Thus, the total space consumption  of \refalg{fmbc} is $O(n+m)$.	
%\end{proof}

Based on \refalg{fmbc}, it is clear that the memory consumption of our enumeration framework is linear to the size of the input signed network. Therefore, the drawback of large memory consumption in \baseline is avoided.

\section{Enumeration Optimization Strategies}
\label{sec:optimization}
Although \refalg{fmbc} addresses the memory consumption problem in \kw{MBCE}, the efficiency of \refalg{fmbc} is disappointing.  In this section, we present two optimization strategies, namely in-enumeration optimization and pre-enumeration optimization, to further improve the efficiency of the enumeration. %The in-enumeration optimization is applied in the enumeration procedure of \refalg{fmbc} while the pre-enumeration optimization works before  \refalg{fmbc}.  %firstly propose two efficient global pruning algorithms, 1-hop reduction and 2-hop reduction, to significantly reduce graph size. Then, we introduce several local pruning techniques during grow-peel algorithm.

\vspace{-0.2cm}

\subsection{In-Enumeration Optimization}
\label{sec:localpruning}

%With above pre-pruning techniques, we can efficiently reduce the graph size at first. Then we run our MBC enumeration algorithm to explore communities on the reduced graph. To further reduce the search space during exploration process, we propose several candidate pruning techniques and branches pruning techniques, denoted by CP and BP.
%We consider two in-enumeration optimization strategies, namely branch pruning and early termination, to further improve the performance. %Some optimization techniques are inspired by the works for  clique enumeration on unsigned networks introduced in \refsec{related}.

\stitle{\underline{Branch Pruning.}} Branch pruning aims to prune the unfruitful branches in the search tree of \refalg{fmbc} to improve the performance.

\stitle{Pivot Choosing.} Consider the maximal balanced clique search procedure of \refalg{fmbc}, assume that we currently have $C_L$, $C_R$, $P_L$ and $P_R$, and we add a vertex $v$ from $P_L$ to $C_L$ in line 17. After finishing the search starting from $v$, we do not need to further explore the positive neighbors of $v$ in the for loop of line 16 and the negative neighbors of $v$ in the for loop of line 19. The reasons are as follows: w.o.l.g, let $v'$ be a positive neighbor of $v$, although we skip the maximal balanced clique search starting from $v'$, these maximal balanced cliques containing $v'$ must be explored by the searching branches starting $v$ or neighbors of $v'$. Therefore skipping the search starting from $v$'s neighbors does not affect the correctness of \refalg{fmbc}.

In this paper, to maximum the benefits of pivot technology, we define the local degree for a vertex $v \in P_L\cup Q_L (P_R \cup Q_R)$ as $d_l(v) = |N^{+(-)}(v) \cap P_L| + |N^{-(+)}(v) \cap P_R|$, and we choose the vertex $v$ that satisfies $\max_{v \in V'}\{d_l(v)\}$ as the pivot, where $V'=P_L \cup P_R \cup Q_L \cup Q_R$.

\stitle{Candidate Selection.} In the search procedure of \refalg{fmbc}, heuristically, search starting from a vertex with small local degree will have a short and narrow search branch, which means the search starting from the vertex will be finished very fast. Moreover, due to the search finish of the vertex, the vertex will be added into the excluded set and it can be used to further prune other search branches. Therefore,  instead of adding vertices from $P_L$ and $P_R$ into $C_L$ and $C_R$ randomly in line 16 and 19 of \refalg{fmbc}, we add vertices in the increasing order of their local degrees.

\stitle{\underline{Early Termination.}} We consider different conditions that we can terminate the search early in \refalg{fmbc}. For a balanced clique $C = \{C_L, C_R\}$, the maximal possible size of $C_L$ ($C_R$) for the final maximal balanced clique is $|C_L| + |P_L|$ ($|C_R| + |P_R|$). Based on the size constraint of $k$, we have the following rule:  %Inspired by early termination technique on clique enumeration problem, in this paper, we early terminate search branch if current result set and candidate set is an exact maximal balance community. We propose following theorem:

\begin{itemize}
\item \textbf{ET Rule 1:} If $|C_L|+|P_L|<k$ or $|C_R|+|P_R|<k$, we can terminate current search directly.
\end{itemize}

In \refalg{fmbc}, we use $Q_L$ and $Q_R$ to store such vertices  that the maximal balanced cliques containing them have been enumerated. Therefore, during the enumeration, if there exists a vertex $v \in Q_L (Q_R)$ such that $P_L (P_R) \subseteq N^+_G(v)$ and $P_R (P_L) \subseteq N^-_G(v)$, then we can conclude that the maximal balanced cliques have been enumerated.  Following this, we have our second rule:

\begin{itemize}
\item \textbf{ET Rule 2:} If $\exists v \in Q_L$, s.t., $P_L \subseteq N^+_G(v)$ and $P_R \subseteq N^-_G(v)$
or $\exists v \in Q_R$, s.t., $P_R \subseteq N^+_G(v)$ and $P_L \subseteq N^-_G(v)$, then we can terminate current search directly.
\end{itemize}

In a certain search of \refalg{fmbc}, if all the vertices in $P_L$ ($P_R$) consist a clique formed by positive edges and every vertex in $P_L$ ($P_R$) has negative edges to all the vertices in $P_R$ ($P_L$), then $P_L$ and $P_R$ consist a balanced clique. Then, based on \refdef{balancecommunity},  $C_L \cup P_L$ and $C_R \cup P_R$ consist a maximal balanced clique. Therefore, we have our third early termination rule:

\begin{itemize}
\item \textbf{ET Rule 3:} If
   $\forall p_l \in P_L$, s.t., $P_L \subseteq \{\{p_l\} \cup N^+_G(p_l)\}$ and $P_R \subseteq N^-_G(p_l)$ and
     $\forall p_r \in P_R$, s.t., $P_R \subseteq \{\{p_r\} \cup N^+_G(p_r)\}$ and $P_L \subseteq N^-_G(p_r)$, we can output $C = (C_L \cup P_L, C_R \cup P_R)$ and terminate current search directly.
\end{itemize}

Note that, in order to avoid outputting duplicate maximal balanced cliques, ET Rule 3 must be applied after  ET Rule 2.

\stitle{Algorithm of \mbcenump.} Utilizing the in-enumeration optimization strategies, we propose the optimized algorithm \mbcenump. The pseudocode is omitted here due to space constraints. 

%The maximal balanced clique enumeration algorithm with in-enumeration optimization strategies is shown in \refalg{fmbcp}. %It follows a similar framework to \refalg{fmbc} and the added optimization strategies are highlighted with bold font.
%Since the pseudocode is self-explained, we omit the detailed description here.

\begin{theorem}
Given a signed network $G$, the time complexity of \mbcenump  is $O(\sigma n \cdot 3^{\sigma/3})$, where $\sigma$ is the degeneracy number of $G$.
\end{theorem}

\vspace{-0.2cm}

\subsection{Pre-Enumeration Optimization}
\label{sec:preenumeration}

In pre-enumeration optimization, we aim to  remove the unpromising vertices and edges that not contained in any maximal balanced clique. We explore two optimization strategies based on the neighbors of a vertex and the common neighbors of an edge.

\label{sec:globalpruning}
\stitle{\underline{Vertex Reduction.}} To reduce the size of a signed network, we first consider the neighbors of each vertex $v$, i.e., $N^+_G(v)$ and $N^-_G(v)$ to remove the unpromising vertices. We first define:

\begin{definition}
\label{def:balancecore}
\textbf{(($l, r$)-signed core)} Given a signed network $G = (V, E^+, E^-)$, two integers $l$ and $r$, a ($l, r$)-signed core is a maximal subgraph $\mathcal{C}$ of $G$, s.t., $\min_{v\in \mathcal{C}}\{d^+_{\mathcal C}(v)\} = l$, $ \min_{v\in \mathcal{C}}\{d^-_{\mathcal C}(v)\} = r$.
\end{definition}

%And we have the following lemma:

\begin{lemma}
\label{lem:balancecocre}
Given a signed network $G$ and threshold $k$, a maximal balanced clique satisfying the size constraint with $k$ is contained in a $(k-1,k)$-signed core.
\end{lemma}

%\begin{proof}
%We can prove it by contradiction. Assume there is a vertex $v$ in a maximal balanced $C$ satisfying the size constraint with $k$ but not in a $(k-1,k)$-signed core. Based on \refdef{balancecommunity}, the positive degree of $v$ in $C$ is not less than $k-1$ and the negative degree of $v$ in $C$ is not less than $k$. This contradicts with our assumption. Thus, the lemma holds.
%\end{proof}

Therefore, in order to compute the maximal balanced cliques in a given signed network $G$ with integer $k$, we only need to compute the maximal balanced cliques in the corresponding $(k-1, k)$-signed core of $G$. The remaining problem is how to efficiently compute the $(k-1, k)$-signed core. We propose a linear algorithm \kw{VertexReduction} to address this problem.
%, which is shown in \refalg{onehop}.
%\begin{algorithm}[t]
%\caption{\small $\onehop$($G=(V, E^+, E^-), k$)}
%
%{
%\small
%\begin{algorithmic}[1]
%\WHILE{$\exists v \in V$, s.t. $d^+_{G}(v) < k-1$ or $d^-_{G}(v) < k$}
%\FOR{\textbf{each} $u \in N^+_{G}(v)$}
%\STATE $d^+_{G}(u) \leftarrow d^+_{G}(u) -1$;
%\ENDFOR
%\FOR{\textbf{each} $u \in N^-_{G}(v)$}
%\STATE $d^-_{G}(u) \leftarrow d^-_{G}(u) -1$;
%\ENDFOR
%\STATE $G \leftarrow G \setminus v$;
%\ENDWHILE
%\end{algorithmic}
%}
%\label{alg:onehop}
%\end{algorithm}

\stitle{Algorithm of \kw{VertexReduction}.} Based on \refdef{balancecore}, to compute the $(k-1, k)$-signed core in the signed network $G$, we only need to identify the vertices with $d^+_G(v) < k-1$ or $d^-_G(v) < k$ and remove them from $G$. Due to the removal of such vertices, more vertices will violate the degree constraints, we can further remove these vertices until no such kind of vertices exist in $G$. 
%Following this idea, in \refalg{onehop}, we first identify a vertex $v$ with $d^+_G(v) < k-1$ or $d^-_G(v) < k$ (line 1). Since $v$ will be removed from $G$, we decrease the positive degree by 1 for each positive neighbor of $v$ (line 2-3) and decrease the negative degree by 1  for each negative neighbor of $v$ (line 4-5). Then, we remove $v$ from $G$ (line 6). The algorithm terminates when no vertex with $d^+_{G}(v) < k-1$ or $d^-_{G}(v) < k$ exists in $G$ (line 1). It is clear that \refalg{onehop} correctly computes the $(k-1, k)$-signed core of $G$. And we have the  following theorem regarding its efficiency.

\begin{theorem}
\label{thm:balancecocre}
Given a signed network $G$ and an integer $k$, the time complexity of \kw{VertexReduction} is $O(n+m)$.
\end{theorem}

%\vspace{-0.3cm}
%\begin{proof}
%In \refalg{onehop}, we use a queue to store vertices that should be removed  in line 6. Since every vertex is pushed in and popped from the queue at most once, the total processing time for this part is $O(n)$. Moreover, when a vertex is removed, we have to update the degrees for their neighbors once, the total time cost is $O(m)$. Therefore,  the time complexity of \refalg{onehop} is $O(n+m)$.
%\end{proof}
%\vspace{-0.3cm}

%\begin{figure}[t]
%\begin{center}
%\begin{tabular}[t]{c}
%%\includegraphics[width=0.75\columnwidth]{fig/core.pdf}
%\includegraphics[width=0.75\columnwidth]{fig/reduction.pdf}
%\end{tabular}
%\end{center}
%\topcaption{Vertex Reduction and Edge Reduction}
%\label{fig:core}
%\vspace{-0.4cm}
%\end{figure}

%
%\begin{example}
%\label{exp:1hopdegree}
%Let $k=2$, \reffig{core} shows an example of vertex reduction by \refalg{onehop} on the signed network $G$ in \reffig{community}. $v_4, v_{12}, v_{16}, v_{17}, v_{18}, v_{19}$ are pruned, because they are not contained in the $(1,2)$-signed core.
%\end{example}

\begin{figure}[t]
\begin{center}
\begin{tabular}[H]{c}
\subfigure[]{
\label{fig:triangle1}
\centering
\includegraphics[width=0.2\columnwidth]{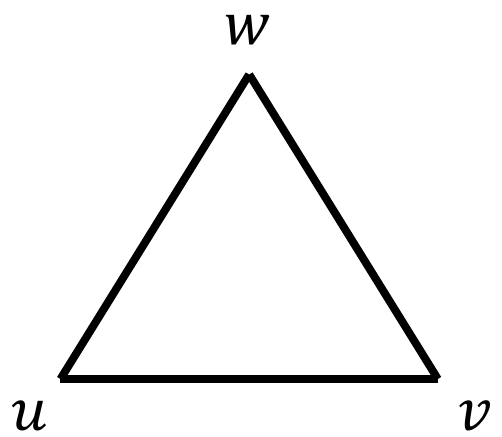}
}
% \vspace{-0.2cm}
\subfigure[]{
\label{fig:triangle2}
\centering
\includegraphics[width=0.2\columnwidth]{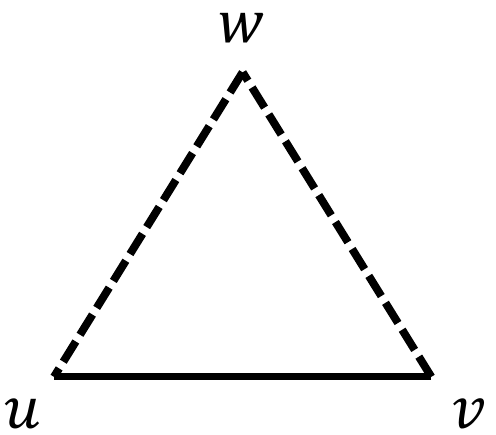}
}
\subfigure[]{
\label{fig:triangle3}
\centering
\includegraphics[width=0.2\columnwidth]{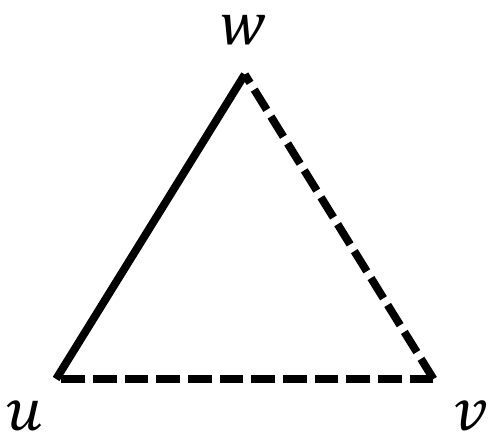}
}
\subfigure[]{
\label{fig:triangle4}
\centering
\includegraphics[width=0.2\columnwidth]{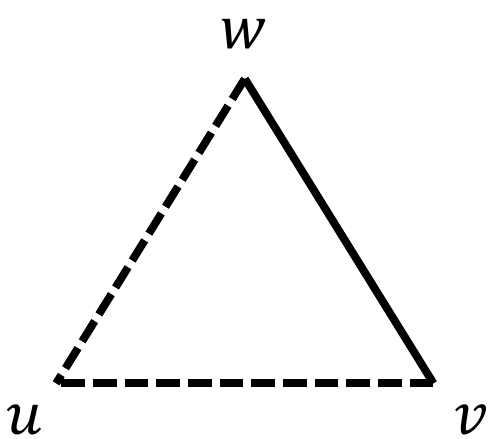}
}
\end{tabular}
\end{center}
%\vspace{-0.4cm}

\topcaption{Different types of common neighbors for  $(u, v)$}
\label{fig:triangle}
\vspace{-0.4cm}
\end{figure}

\stitle{\underline{Edge Reduction.}} In this part, we explore the opportunities  to remove unpromising edges with respect to \kw{MBCE} by considering the common neighbors of an edge formed by different types of edges. Specifically, for a positive/negative edge $(u, v)$, we define the edge common neighbor number:

\begin{definition}
\label{def:2hopdegree}
\textbf{(Edge Common Neighbor Number)} Given a signed network $G =(V, E^+, E^-)$, for a positive edge $(u, v)$, we define:
\begin{itemize}
	\item $\delta^{++}_G(u, v) = |\{w|(u, w) \in E^+ \wedge (v, w) \in E^+ \}|$
	\item $\delta^{--}_G(u, v) = |\{w|(u, w) \in E^- \wedge (v, w) \in E^- \}|$
\end{itemize}
for a negative edge $(u, v)$, we define:
	\begin{itemize}
	\item $\delta^{+-}_G(u, v) = |\{w|(u, w) \in E^+ \wedge (v, w) \in E^- \}|$
	\item $\delta^{-+}_G(u, v) = |\{w|(u, w) \in E^- \wedge (v, w) \in E^+ \}|$
\end{itemize}
\end{definition}
\vspace{-0.3cm}

\reffig{triangle} shows the different types of common neighbors used in \refdef {2hopdegree}. For a positive edge $(u, v)$, \reffig{triangle} (a) and (b) show the common neighbor $w$ used in $\delta^{++}_G(u, v)$ and $\delta^{--}_G(u, v)$, respectively. For a negative edge $(u, v)$, \reffig{triangle} (c) and (d) show the common neighbor $w$ used in $\delta^{+-}_G(u, v)$ and $\delta^{-+}_G(u, v)$, respectively. Note that $G$ is undirected and every edge is stored once in $G$. Based on \refdef{2hopdegree}, we have the following lemma:

\begin{lemma}
\label{lem:2hopdegree}
Given a signed network $G$ and an integer $k$, let $G'$ be the maximal sub-network of $G$ s.t.,
\begin{enumerate}
		\item $\forall (u, v) \in E^+_{G'} \rightarrow \delta^{++}_{G'}(u,v) \ge k-2 \wedge \delta^{--}_{G'}(u,v) \ge k$;
		\item $\forall (u, v)\in E^-_{G'} \rightarrow \delta^{+-}_{G'}(u, v) \ge k-1 \wedge \delta^{-+}_{G'}(u, v) \ge k-1$;
\end{enumerate}
then, every maximal balanced clique $C=\{C_L,C_R\}$ in $G$ satisfying the size constraint with $k$ is contained in $G'$.
\end{lemma}

\stitle{Algorithm of \kw{EdgeReduction}.} With \reflem{2hopdegree}, in order to enumerate the maximal balanced cliques in a given signed network $G$ with respect $k$, we only need to keep the edges in $G'$ shown in \reflem{2hopdegree} and the positive/negative edges not in $G'$  can be safely pruned. We first compute $\delta^{++}_G(u, v)$ and $\delta^{--}_G(u, v)$ for each positive edge of $G$ and $\delta^{+-}_G(u, v)$ and $\delta^{-+}_G(u, v)$ for each negative edge of $G$. Following \reflem{2hopdegree}, for each positive edge $(u, v)$ such that $\delta^{++}_G(u,v) < k-2$ or $\delta^{--}_G(u,v) < k$, we remove $(u, v)$. After that, we decrease the corresponding edge common neighbor numbers that have been changed due to the removal of $(u, v)$ for the edge incident to $(u, v)$ based on \refdef{2hopdegree}. It's similar to negative edges. The algorithm terminates when all the edges satisfy conditions in \reflem{2hopdegree}.
%Now we focus on efficiently computing $G'$  and our algorithm is shown in \refalg{twohop}. We first compute $\delta^{++}_G(u, v)$ and $\delta^{--}_G(u, v)$ for each positive edge of $G$ (line 1-2) and $\delta^{+-}_G(u, v)$ and $\delta^{-+}_G(u, v)$ for each negative edge of $G$ (line 3-4). Following \reflem{2hopdegree}, for each positive edge $(u, v)$ such that $\delta^{++}_G(u,v) < k-2$ or $\delta^{--}_G(u,v) < k$, we remove $(u, v)$ (line 9). After that, we decrease the corresponding edge common neighbor numbers that have been changed due to the removal of $(u, v)$ for the edge incident to $(u, v)$ (line 10-15) based on \refdef{2hopdegree}. Similarly, for each negative edge not satisfying the conditions in \reflem{2hopdegree}, we remove it and decrease the corresponding edge common neighbor numbers (line 17-24). The algorithm terminates when all the edges satisfy conditions in \reflem{2hopdegree}. %Based on the steps of \refalg{twohop}, it is clear that we obtain $G'$ shown in \reflem{2hopdegree}. And we have the following theorem regarding the time complexity of \refalg{twohop}.

\begin{theorem}
\label{thm:2hopdegree}
Given a signed network $G$, an integer $k$, the time complexity of \kw{EdgeReduction} is $O(m^{1.5})$.
	
\end{theorem}
\section{Maximum Balanced Clique Search}
\label{sec:mbcs}

Maximum clique search problem is a fundamental and hot research topic in graph analysis. In this section, we study the maximum balanced clique search problem.
\subsection{A Baseline Approach}
\label{sec:baseline2}
We first propose a baseline approach, namely \mbcs, to compute the maximum balanced clique in the input graph. We continuously enumerate the maximal balanced cliques in the input graph and maintain the maximum balanced clique $C^*$ found so far. For each search branch ($C_L,C_R,P_L,P_R$),  if $|C_L| + |C_R| + |P_L| + |P_R| \le \epsilon$, where $\epsilon=|C^*|$, we can terminate the branch. When the enumeration finishes, it is easy to verify that $C^*$ is the maximum balanced clique.

\stitle{Drawbacks of \mbcs.}  Although the straightforward approach can find the maximum balanced clique, the complexity of \mbcs is the same as that of \kw{MBCEnum^*} in the worst case. The search space of \mbcs is huge. In details, the drawbacks of \mbcs are twofold.

\begin{itemize}
%tight bounds for $|C_L|$ and $|C_R|$.
\item \textbf{Lack of rigorous size constraints for $|C_L|$ and $|C_R|$.} Given a signed graph $G$, during the search process, \mbcs only holds the size constraint $|C_L|+|C_R|+|P_L|+|P_R|> \epsilon$ for each search branch. However, when $\epsilon$ is small, most of search branches have $|C_L|+|C_R|+|P_L|+|P_R|$ larger than $\epsilon$ which causes the fail of size constraint for most  search branches. Unfortunately, as our algorithm constantly searches larger result than at present, the value of $\epsilon$ is gradually increasing from a small value, which makes \mbcs has to search the result with large search space.

\item \textbf{Massive invalid search branches.} Although the search branches meet the size constraint with $\epsilon$, the structure between $P_L$ and $P_R$ maybe sparse which will generates invalid search branches. Hence,   during the search process, more pruning techniques is needed urgently. Moreover, the optimization strategies based on $k$ in \kw{MBCEnum^*}, like vertex reduction and edge reduction, are limited here, as $C^*$ usually has size much larger than $k$. Therefore, the remaining graph after reduction is still huge on large-scale signed network. 
\end{itemize}

\stitle{Main idea.} In the further work, we aim to improve the efficiency of our algorithm. 
\begin{itemize}
\item To address the first drawback of lacking of rigorous size constraints for $|C_L|$ and $|C_R|$, we can propose $\underline{\kappa}$ and $\overline{\kappa}$ as the lower bounds for $min\{|C_L|,|C_R|\}$ and $max\{|C_L|,|C_R|\}$, respectively. Then, we can get the balanced cliques with $|C_L|\ge \underline{\kappa}$ and $|C_R|\ge \overline{\kappa}$ within narrow search space (assume $|C_L| \le |C_R|$). Under different value of $\underline{\kappa}$ and $\overline{\kappa}$, the search space is split into multiple partitions. Moreover, with initializing $\underline{\kappa}$ or $\overline{\kappa}$ as large value, we can  search balanced cliques with large size as priority. Under large $\epsilon$ and rigorous bounds $\underline{\kappa}$ and $\overline{\kappa}$, the search space can be significantly reduced.
\item To address the second drawback of massive invalid search branches, regarding to bounds $\epsilon$, $\underline{\kappa}$ and $\overline{\kappa}$, we can propose optimizations to forecast the size of balanced clique found in the current search branch to avoid invalid search branches and remove redundant vertices from candidates. Moreover, we can extend the vertex reduction and edge reduction of \kw{MBCEnum^*} with new bounds to prune more useless vertices and edges.
\end{itemize}

%\iffalse
%\bibliography{./reference.bib}
%\fi
\subsection{Search Space Partition-based Framework}
\label{sec:spacesplit}
%We have analyzed that the huge search space of \mbcs algorithm leads to its low efficiency.
 To improve the efficiency of our approach, regarding to the first drawback of \mbcs, in this subsection, we propose a new maximum balance clique search framework \mbcss with two lower bounds $\underline{\kappa}$ and $\overline{\kappa}$ for $|C_L|$ and $|C_R|$. Given certain value $\underline{\kappa}_i$ and $\overline{\kappa}_i$, a search region is  denoted as $(\underline{\kappa}_i, \overline{\kappa}_i)$, the maximum balanced clique found in it should satisfies $|C_L| \geq \underline{\kappa}_i$ and $|C_R| \geq \overline{\kappa}_i$ if $|C_L|\le |C_R|$ (otherwise, swap $L$ and $R$). Under different value of $\underline{\kappa}$ and $\overline{\kappa}$, the whole search space can be divided into several search regions. In each search region, we keep searching larger result than at present. When all search regions are explored, the final result $C^*$ can be found.
 
  As our main idea, to search the result with large size as priority, for the first search region $(\underline{\kappa}_0 , \overline{\kappa}_0)$,  $\overline{\kappa}_0$ is initialized as a large integer value. Obviously, as $\overline{\kappa}_0$ value is large,  benefited from the strict size constraint, most of search branches of \mbcs are ineligible now. Hence the result can be found  quickly in this search region. Besides, to obey the size threshold $k$, we make $\underline{\kappa}_0=k$. Then, to cover the whole search space,  for the later search regions,  we keep increasing $\underline{\kappa}$ and decreasing  $\overline{\kappa}$  until $\underline{\kappa}=\overline{\kappa}$. In another word, for $i<j$,  $\underline{\kappa}_i\le \underline{\kappa}_j$ and  $\overline{\kappa}_i \ge \overline{\kappa}_j$. 
% Besides, because $\underline{\kappa}_0=k$ and $\underline{\kappa}$ keeps increasing, all balanced cliques found in every search region will obey the size threshold $k$. 
 
 Here,  we first assign the possible maximum value to $\overline{\kappa}_0$,  we have the following lemma:

%\footnote{The degeneracy of a graph $G$ is the least $k$ such that every induced subgraph of $G$ contains a vertex with $k$ or fewer neighbors.}
\begin{lemma}
\label{lem:degeneracy}
Given a signed network $G=(V, E^+, E^-)$, for every balanced clique, we have $max\{|C_L|,|C_R|\} \le \sigma+1$, where $\sigma$ is the degeneracy number of $G^+$.
\end{lemma}
\begin{proof}
In unsigned networks, the degeneracy number plus 1 is an upper bound for the maximum size of  cliques\cite{DBLP:journals/pvldb/LuYWZ17}. Based on \refdef{balancecommunity},  in a signed network $G=(V, E^+, E^-)$,  for every balanced clique $C=\{C_L, C_R\}$,  $C_L$ and $C_R$ are traditional cliques in  $G^+$. Therefore, $|C_L|$ and $|C_R|$ are must not greater than  $\sigma+1$, respectively.
\end{proof}

Based on \reflem{degeneracy},  we assign  $(k, \sigma+1)$ to the first search region $(\underline{\kappa}_0, \overline{\kappa}_0)$. Then,  in the later search region,  we continue to seek larger balanced clique than the current one. However,  not every search region can find a valid result. To skip invalid search regions, we have the following lemma:

\begin{lemma}
\label{lem:ssmbcs}
Given a signed network $G$, the maximum balanced clique found in the $i$-th search region $(\underline{\kappa}_i , \overline{\kappa}_i)$ is denoted by $C^*_i=\{C_L^i, C_R^i\}$. Then,  for the next search region $(\underline{\kappa}_{i+1} , \overline{\kappa}_{i+1})$ , we have $\underline{\kappa}_{i+1}=|C^*_i|-\overline{\kappa}_i$.
\end{lemma}

\begin{proof}
We prove it by contradiction.  Following the $i$-th search region,   in the next search region $(\underline{\kappa}_{i+1} , \overline{\kappa}_{i+1})$, if we get a larger balanced clique $C^*_{i+1}=\{C_L^{i+1}, C_R^{i+1}\}$ than $C^*_{i}$. Based on our search framework, we have $max\{|C_L^{i+1}|, |C_R^{i+1}|\}<\overline{\kappa}_i$, otherwise, $C^*_{i+1}$ will be found in the $i$-th search region rather than the $(i+1)$-th search region. Now,  we assume $min\{|C_L^{i+1}|, |C_R^{i+1}|\}<|C^*_i|-\overline{\kappa}_i$. Combining with $max\{|C_L^{i+1}|, |C_R^{i+1}|\}<\overline{\kappa}_i$, we have $|C^*_{i+1}|=|C_L^{i+1}|+|C_R^{i+1}|<|C^*_i|$. Obviously,  it is against with our premise that $C^*_{i+1}$ is larger than $C^*_{i}$. Therefore, the assumption for $min\{|C_L^{i+1}|, |C_R^{i+1}|\}<|C^*_i|-\overline{\kappa}_i$ does not hold. We get $min\{|C_L^{i+1}|, C_R^{i+1}\}\geq|C^*_i|-\overline{\kappa}_i$, i.e.,  $\underline{\kappa}_{i+1}=|C^*_i|-\overline{\kappa}_i$.
\end{proof}

Based on \reflem{ssmbcs}, after the $i$-th search region, the search regions with $\underline{\kappa} < |C^*_i|-\overline{\kappa}_i$ can be skipped directly.

\begin{figure}[t]
\begin{center}
\includegraphics[width=0.95\columnwidth]{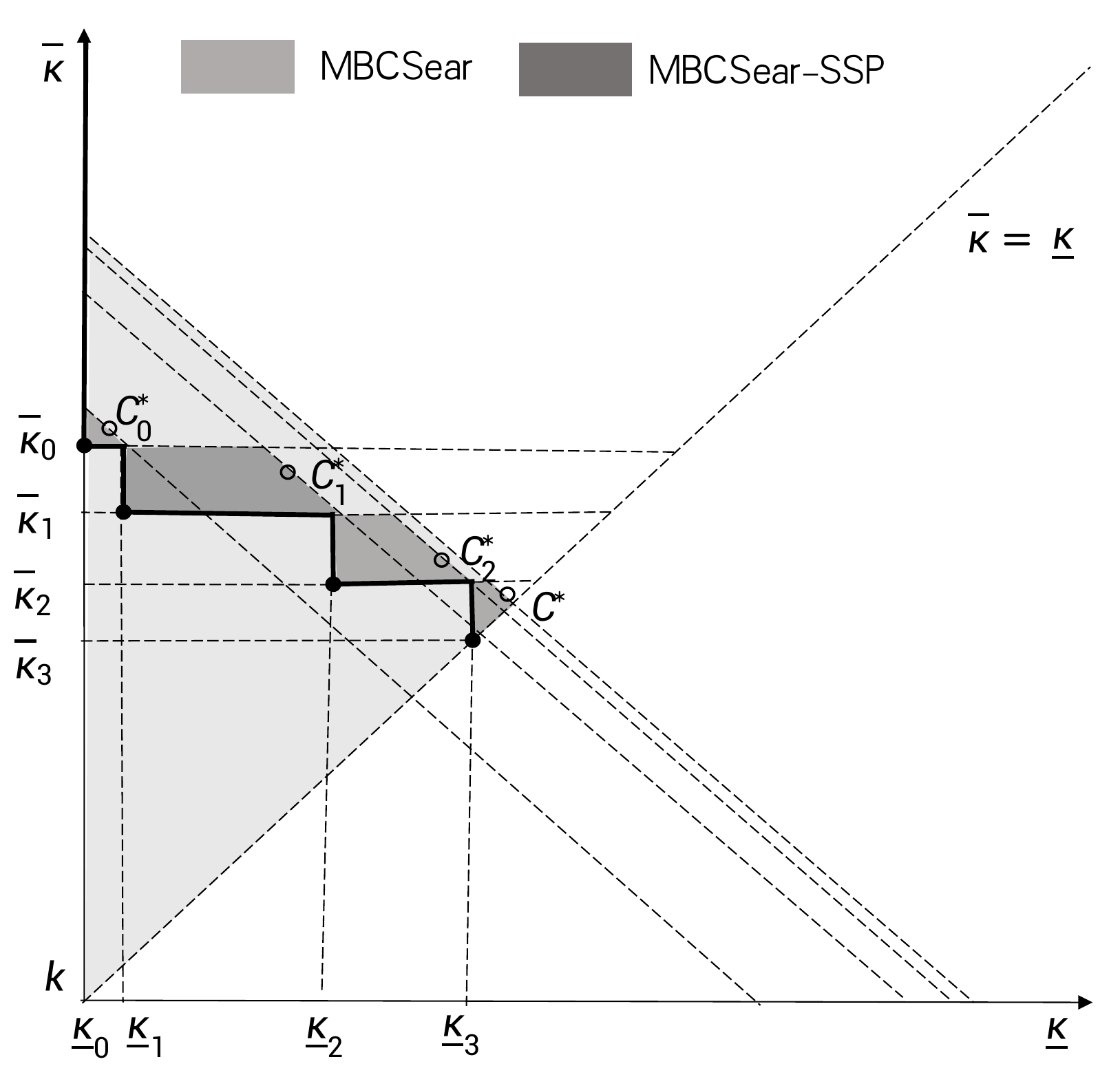}
\end{center}
\vspace{-0.2cm}
\topcaption{ The search space of \mbcs and \mbcss}
\label{fig:ss}
%\vspace{-0.4cm}
\end{figure}

\begin{algorithm}[t]
\caption{\small \mbcss($G=(V,  E^+,  E^-), k$)}
{
\small
\begin{algorithmic}[1]
\STATE compute degeneracy $\sigma$ of $G^+=(V,  E^+)$;
\STATE $\epsilon \leftarrow 2k$; $\underline{\kappa} \leftarrow k$;  $\overline{\kappa} \leftarrow \sigma + 1$; $\overline{\kappa}' \leftarrow -1$;
\WHILE{$\overline{\kappa}\geq \underline{\kappa}$ 
and $\overline{\kappa}<\overline{\kappa}'$ }

\STATE \mbcs$(G=(V,  E^+,  E^-), \epsilon)$ with adding size constraints: $min\{\overline{L} ,\overline{R}\} < \underline{\kappa}$ or   $max\{\overline{L} ,\overline{R}\} < \overline{\kappa}$;
%\STATE \textbf{if} $\underline{\kappa}=\overline{\kappa}$ \textbf{then}  \textbf{break};
\STATE $\overline{\kappa}' \leftarrow \overline{\kappa}$;
 $\underline{\kappa} \leftarrow max\{\epsilon - \overline{\kappa},  k\}$;
  $\overline{\kappa} \leftarrow max\{Dec(\overline{\kappa}),  \underline{\kappa}\}$;

%\STATE $\overline{\kappa} \leftarrow max\{\lceil\frac{\overline{\kappa}}{2}\rceil,  \underline{\kappa}\}$;
%\STATE $\overline{\kappa} \leftarrow max\{\overline{\kappa}-2,  \underline{\kappa}\}$;
%\STATE
\ENDWHILE
\end{algorithmic}
}
\label{alg:mbcss}
\end{algorithm}

\stitle{Algorithm of \mbcss.} Following the above idea, the new maximum balanced clique search algorithm \mbcss is shown at \refalg{mbcss}. Given a signed network $G=(V,  E^+,  E^-)$ and size threshold $k$, we first compute the degeneracy number $\sigma$ of $G^+=(V,  E^+)$ (line 1). We initialize $\epsilon=2k$, $\underline{\kappa} = k$,  $\overline{\kappa} = \sigma+ 1$ (line 2). Then, in each search region, \refalg{mbcss} invokes \mbcs to find the maximum balanced clique in current search region with adding size constraints: $min\{\overline{L} ,\overline{R}\} < \underline{\kappa}$ or   $max\{\overline{L} ,\overline{R}\} < \overline{\kappa}$ on each search branch.  For the next search region, we let $\underline{\kappa}=max\{\epsilon - \overline{\kappa},  k\}$ (line 5), since the size threshold $k$ is held for $C_L$ and $C_R$ as well. For the next value of $\overline{\kappa}$, $\overline{\kappa}=max\{Dec(\overline{\kappa}), \underline{\kappa}\}$ , where $Dec(\overline{\kappa})$ is used to decrease $\overline{\kappa}$. 
When $\overline{\kappa}<\underline{\kappa}$, let $\overline{\kappa}$=$\underline{\kappa}$. The search process finishes when $\overline{\kappa}=\underline{\kappa}$ and $\overline{\kappa}$ can not be reduced anymore, since if $\overline{\kappa}$ is unchanged, this search region is covered by the previous region already(line 3). When the search process of the final search region is finished, \refalg{mbcss} terminates and returns the maximum balanced clique $C^*$ in $G$.

%Regarding to $Dec(\overline{\kappa})$, there is a trade-off between $\omega$ and $SS_i$. If we decrease $\overline{\kappa}_i$ by a small number, $SS_i$ will be small as well. But, due to \refalg{mbcss} will keep running until  $\overline{\kappa}$=$\underline{\kappa}$,  $\omega$ will be large and vice versa.

%Here, we compare the search space between \mbcs and \mbcss. As shown at \reffig{ss}, since $C^*$ is the maximum result without any larger one, all balanced cliques searched by \mbcs below the line where $C^*$ is at \reffig{ss}. However, benefited from search space partition, \mbcss only search the local maximum balanced clique in each search region. In a result, the search space of \mbcss is much smaller than \mbcs .

% \begin{proof}
% \end{proof}

Then, we discuss how $Dec(\overline{\kappa})$ decreases the value of $\overline{\kappa}$. The total running time of \mbcss can be formulated as $T=\Sigma^{\omega}_{i=0}{t_i}$, where $\omega$ is the number of search regions and $t_i$ is the partial running time of the $i$-th search region.  In this paper, to keep the efficiency of \mbcss, we give a heuristic way to decrease $\overline{\kappa}$. In detail, if the last search region is time-consuming (set time threshold like 20 seconds), to make the total running time $T$ as small as possible, we reduce the amount of search regions $\omega$ by decreasing $\overline{\kappa}$ by a large value, i.e., $Dec(\overline{\kappa})=\lceil\frac{\overline{\kappa}}{2}\rceil$, otherwise, $Dec(\overline{\kappa})=\overline{\kappa}-2$. 

Now, we compare the search space between \mbcs and \mbcss. As shown at \reffig{ss}, since $C^*$ is the maximum balanced clique, \mbcs searches all balanced cliques with size less than $C^*$ until  $C^*$ is found. The search space of \mbcs is shown at \reffig{ss}. For the search space of \mbcss, benefited from the tight bounds for $|C_L|$ and $|C_R|$ at each search region, \mbcss searches local maximum balanced clique within small search space. For instance, as shown at \reffig{ss}, it finds the current maximum balanced clique $C^*_0$ in the first search region. Then, in the second search region, it searches balanced cliques with $\overline{\kappa}_1\le|C_R|<\overline{\kappa}_0$ and $|C_L|\ge \overline{\kappa}_1$ (assume $|C_L|\le |C_R|$) until $C^*_1$ is found. As shown at \reffig{ss}, the search space of \mbcss is much smaller than \mbcs.

\begin{example}
%unfortunately, there does not exist a valid balanced clique in the first search region.
Reconsidering the signed network $G$ in \reffig{community}, $k=2$, \refalg{mbcss} first computes degeneracy number $\sigma$ of $G^+$, get $\sigma=2$. So, the first search region is $(2, 3)$. \refalg{mbcss} find result $C^*_0=\{\{v_0, v_1, v_3\}, \{v_5, v_6,v_7\}\}$ at first, $\epsilon=6$. Then, based on \reflem{ssmbcs}, $\underline\kappa_1$ is 3. As we should keep $\overline\kappa \ge \underline\kappa$, $\overline\kappa_1$ is 3 as well. However, as the value of $\overline{\kappa}$ is unchanged, this search region is covered by the first search region. Hence, \refalg{mbcss} is terminated and returns $C^*_0$ as $C^*$. Comparing with \mbcs, the search space of \mbcss is reduced from $(2,2)$ to $(2,3)$.
\end{example}

\subsection{Optimization Strategies}
\label{sec:optimization2}
Regarding to the second drawback of \mbcs on massive invalid search branches in each search region, in this subsection, we explore the chance to further improve the efficiency of our approach. Under our rigorous lower bounds $\underline\kappa$ and $\overline\kappa$, we first propose two optimization strategies, coloring-based branch pruning and vertex domination-based  candidate pruning, to prune invalid search branches and remove meaningless vertices from candidates. Then, we extend the vertex$\&$edge reduction technologies of \kw{MBCE} to prune more unnecessary vertices and edges in advance.

\subsubsection{Coloring-based Branch Pruning}
\label{sec:color}
Given a search region $(\underline{\kappa},\overline{\kappa})$ and a search branch, if the upper bound of the balanced clique size in current search branch is less than the lower bounds $\underline{\kappa}$, $\overline{\kappa}$ and $\epsilon$, current search branch can be pruned directly. Looking back to \mbcs, it uses the candidates size to form the upper bound. However, this upper bound is too loose, because although the number of candidates is large, the connectivity between candidates maybe sparse, which will lead to many invalid search branches. Hence, now, we aim to  propose a tighter upper bound based on vertex coloring. 
 
 \begin{definition}
 \label{def:color}
 \textbf{(Vertex Coloring\cite{LOVASZ1989319})} Given a graph $G$, vertex coloring in $G$ aims to assign colors to each vertex such that vertices are different in color from their neighbors. The amount of colors needed in $G$ is named chromatic number, denoted by $\gamma(G)$. 
\end{definition}

\begin{lemma}
\label{lem:color}
Given a search branch ($C_L,C_R,P_L,P_R$), the maximum balanced clique from this branch is denoted as $C'=\{C'_L,C'_R\}$, we have that $|C'_L| \le \gamma(\mathbb{G}_L) +|C_L|$, $|C'_R| \le \gamma(\mathbb{G}_R) +|C_R|$, where $\mathbb{G}_{L(R)}$ is the positive subgraph produced by $P_{L(R)}$. 
	\end{lemma}
\begin{proof}
Given a graph $G$, the chromatic number $\gamma(G)$ is an upper bound of the maximum size of cliques in $G$\cite{DBLP:journals/pvldb/LuYWZ17}. Based on it, the lemma can be proved.
\end{proof}

Based on \reflem{color}, if the upper bound does not meet the size requirements of current search region, i.e., $min\{\gamma(\mathbb{G}_{L})+|C_L|,\gamma(\mathbb{G}_{R})+|C_R|\}<\underline{\kappa}$ or $max\{\gamma(\mathbb{G}_{L})+|C_L|,\gamma(\mathbb{G}_{R})+|C_R|\}<\overline{\kappa}$ or $\gamma(\mathbb{G}_{L})+|C_L|+\gamma(\mathbb{G}_{R})+|C_R|\leq \epsilon$, the search branch can be early terminated directly. 

\stitle{Algorithm of \kw{ColoringPrune}}. We propose \kw{ColoringPrune} algorithm to prune search branches. The pseudocode is shown at \refalg{color}. It first computes $\gamma(\mathbb{G}_L)$ and $\gamma(\mathbb{G}_R)$ (line 1-8). Then it returns \kw{true} if the upper bound does not meet the size requirements, which means this search branch can be pruned directly, otherwise, returns \kw{false} (line 9-11).

\begin{algorithm}[t]
\caption{\small $\kw{ColoringPrune}(C_L,C_R,P_L, P_R , \epsilon, \underline{\kappa},\overline{\kappa}$)}
{
\small
\begin{algorithmic}[1]
% \STATE $\mathbb{G}_L$ is reduced by $C_L,P_L$ with positive edges;
% \STATE $\mathbb{G}_R$ is reduced by $C_R,P_R$ with positive edges;
%\STATE $\mathbb{G}_L$ is produced by $P_L$ with positive edges;
%\STATE $\mathbb{G}_R$ is produced by $P_R$ with positive edges;
\STATE $\mathbb{G}_L \leftarrow G^+(P_L)$; $\mathbb{G}_R \leftarrow G^+(P_R)$;
\STATE $\gamma(\mathbb{G}_{L(R)})$ $\leftarrow 0$; $col(v)\leftarrow 0$ \textbf{for} \textbf{each} $v\in P_{L(R)}$;
\FOR{\textbf{each} $v \in P_{L(R)}$} 
\STATE $col(v) \leftarrow 1$;
\WHILE{$\exists u\in N_{\mathbb{G}_{L(R)}}(v)$, s.t., $col(u)=col(v)$}
\STATE $col(v) \leftarrow col(v)+1$; 
\ENDWHILE
 
\IF{$col(v)>\gamma(\mathbb{G}_{L})$}
\STATE $\gamma(\mathbb{G}_{L(R)})\leftarrow \gamma(\mathbb{G}_{L(R)})+1$;
\ENDIF 
\ENDFOR 

\IF{$min\{\gamma(\mathbb{G}_{L})+|C_L|,\gamma(\mathbb{G}_{R})+|C_R|\}<\underline{\kappa}$ or $max\{\gamma(\mathbb{G}_{L})+|C_L|,\gamma(\mathbb{G}_{R})+|C_R|\}<\overline{\kappa}$ or $\gamma(\mathbb{G}_{L})+|C_L|+\gamma(\mathbb{G}_{R})+|C_R|\leq \epsilon$ }
\STATE \textbf{return} \textbf{true};
\ENDIF

\STATE \textbf{return} \textbf{false};

\end{algorithmic}
}
\label{alg:color}
\end{algorithm}

\begin{theorem}
The space complexity of \refalg{color} is $O(|P_L|+|P_R|+|E_{\mathbb{G}_L}|+|E_{\mathbb{G}_R}|)$, the time complexity is $O(|P_L|+|P_R|+|E_{\mathbb{G}_L}|+|E_{\mathbb{G}_R}|)$.

\end{theorem}
%\begin{proof}
%For the space complexity, \refalg{color} constructs two subgraphs $\mathbb{G}_L$ and $\mathbb{G}_R$, which consume $O(|P_L|+|E^+_{\mathbb{G}_L}|)$ space and $O(|P_R|+|E^-_{\mathbb{G}_R}|)$ space respectively.
%For the time complexity, constructing $\mathbb{G}_L$ and $\mathbb{G}_R$ consumes $O(|P_L|+|P_R|+|E_{\mathbb{G}_L}|+|E_{\mathbb{G}_R}|)$ time. Finding color for  vertex $v$ needs $O(N_{\mathbb{G}_{L/R}}(v))$ time and each vertex is computed once. As a result, the total time complexity  of \refalg{color} is $O(|P_L|+|P_R|+|E_{\mathbb{G}_L}|+|E_{\mathbb{G}_R}|)$.
%\end{proof}

Note that \refalg{color} can be directly applied to $\kw{MBCE}$ problem with $\underline{\kappa}=k$, $\overline{\kappa}=k$, $\epsilon=2k$. However, since $k$ is usually small, the effectiveness of \refalg{color} is limited in $\kw{MBCE}$ problem. 
%As $\overline{\kappa}$ and $\epsilon$ are much larger than $k$, \refalg{color} is more suitable for $\kw{MBCS}$ problem. 

 \subsubsection{Vertex Domination-based  Candidate Pruning}

 To further improve the efficiency, we reduce the number of candidates in $P_L$ and $P_R$ at each search branch by pruning invalid vertices from candidates. Our key thought is that if we have the prior knowledge to  know that vertex $v$ forms balanced clique with size larger than that of vertex $u$, then, $u$ is dominated by $v$, denoted by $u \in \Phi_v$,  the search relevant to $u$ can be skipped. Simplely, for each search branch, we can use the local neighborhood between candidates to figure out the domination relationship, i.e., if $N_l(u)\subseteq N_l(v)$, then $u \in \Phi_v$. Then, we have:

\begin{lemma}
\label{lem:baseopt}
Given a signed network $G$ and a search branch with candidate sets $P_L$ and $P_R$, for each vertex $v  \in P_L \cup P_R$, if $v$ is dominated, the sub-search branch from  $v$ can be skipped.
\end{lemma}

\begin{proof}
Given a search branch, for vertices $u,v$ in candidates such that $u \in \Phi_v$, we use $C_v$ and $C_u$ to indicate the balanced cliques maintaining $v$ and $u$, respectively. Due to $N_l(u)\subseteq N_l(v)$, then $C_u\setminus\{u\}\subseteq C_v\setminus\{v\}$, i.e., $|C_u|\leq |C_v|$. Hence $u$ won't belong to any balanced clique larger than $C_v$. The search from $u$ can be skipped.
\end{proof}

\begin{figure}[t]
%\setlength{\abovecaptionskip}{20pt}
%\captionsetup[subfigure]{aboveskip=-1pt,belowskip=-1pt}
\begin{center}
%\subfigure[\small Case 1]{
%\label{fig:case1}
%\centering
%\includegraphics[width=0.49\columnwidth]{MaximumBC/fig/special1.pdf}
%}\vspace{-0cm}\subfigure[\small Case 2]{
%\label{fig:case2}
%\centering
%\includegraphics[width=0.49\columnwidth]{MaximumBC/fig/special2.pdf}
%}\\

\subfigure[\small Case 1$\&$3]{
\label{fig:case3}
\centering
\includegraphics[width=0.49\columnwidth]{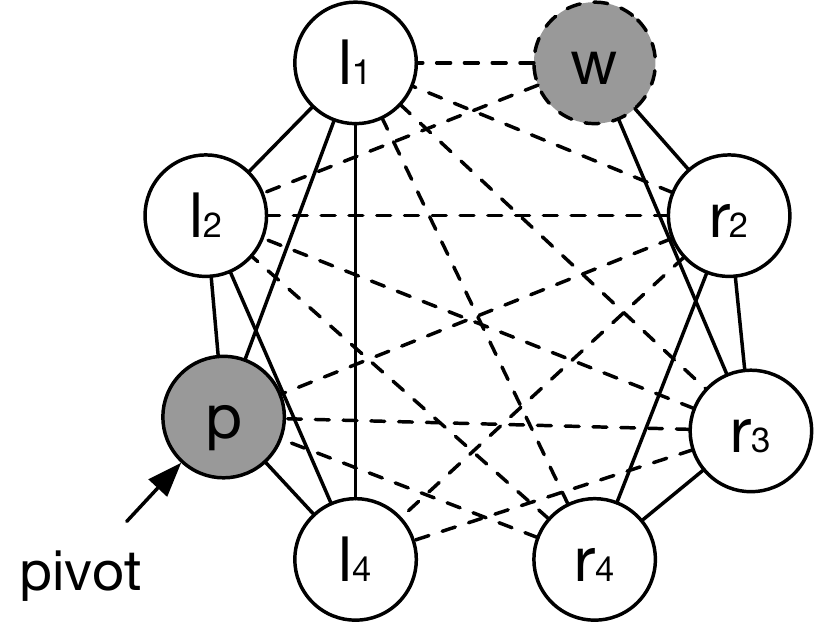}
}\vspace{-0cm}\subfigure[\small Case 2$\&$4]{
\label{fig:case4}
\centering
\includegraphics[width=0.49\columnwidth]{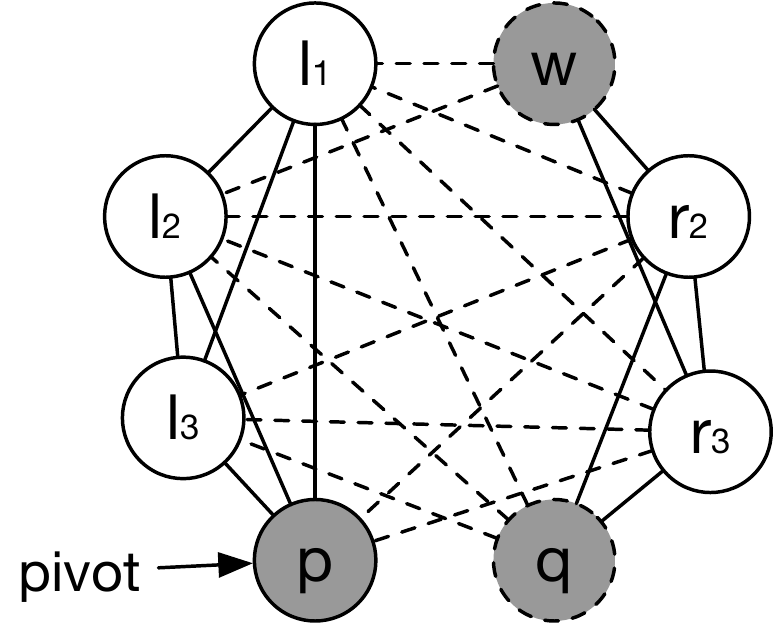}
}

\end{center}
\topcaption{Vertex Domination Cases}
\label{fig:special}
\vspace{-0.4cm}
\end{figure}

The vertex domination can be computed with adjacency list join operation. The time complexity of computing vertex domination at single search branch is $O(\mathrm{C_\mathbb{n}^2}\cdot(\mathbb{n}-1))=O(\mathbb{n}^3)$, where $\mathbb{n} =|P_L|+|P_R|$, which is time-consuming. Moreover, the amount of search branches is huge. To obtain the vertex domination effectively, in this paper, we only consider four special cases based on pivot technique, which can be computed within const time. The four special cases are introduced as follows, $p$ is the selected pivot.

\begin{itemize}
\item \textbf{Case 1:} If $p\in Q_L \cup Q_R$ and $newP_L\cup newP_R=\{w\}$, then $\Phi_p=\{w\}$.
 %\item  $p$ is $w$'s dominating vertex.

\item \textbf{Case 2:} If $p\in Q_L \cup Q_R$,  $newP_L\cup newP_R=\{w,q\}$ and $(w,q) \notin E$, then, $\Phi_p=\{w,q\}$.
%\item $p$ is $w$'s dominating vertex and $q$.

\item \textbf{Case 3:} If $p\in P_L \cup P_R$,  $newP_L\cup newP_R=\{p,w\}$, then, $\Phi_p=\{w\}$.

%\item \textbf{Case 4:} If $p\in P_L \cup P_R$ and $|newP_L|+|newP_R|=3$, then, for $w,q \in newP_L \cup newP_R \setminus \{p\}$, if $w,q$ from the same side and $(w,q) \notin E^+$ or $(w,q)$from different sides but $(w,q) \notin E^-$, $\Phi_p=\{w,q\}$.
\item \textbf{Case 4:} If $p\in P_L \cup P_R$,  $newP_L\cup newP_R=\{p,w,q\}$ and $(w,q) \notin E$, then, $\Phi_p=\{w,q\}$.

\end{itemize}

%\begin{example} 
\reffig{special} shows the four special cases of  vertex domination, respectively. For case 1$\&$3(\reffig{case3}), $p$ is selected as pivot, then, the surviving candidate is $w$ as other vertices are $p$'s neighbors. Since $N_l(w)\subseteq N_l(p)$, $w$ is dominated by $p$. For case 2$\&$4(\reffig{case4}), $w$ and $q$ are two surviving candidates with pivot $p$, as $(w,q)\notin E$, they will not appear at a common balanced clique, hence, $N_l(w)\subseteq N_l(p)$, $N_l(q)\subseteq N_l(p)$, $\Phi_p=\{w,q\}$. %\end{example}

Based on \reflem{baseopt}, when we meet the above four special cases of vertex domination, we skip the searches from vertices in $\Phi_p$ by deleting them from candidate sets, this process only consumes const time. In this way, the invalid candidates can be further pruned effectively.

\subsubsection{Vertex$\&$Edge Reduction Variants}
\label{sec:preenumeration2}

%Given a search region, \mbcs algorithm search balanced clique under two size thresholds $\underline\kappa$ and  $\overline\kappa$.
Although the optimization strategies proposed in \kw{MBCE} algorithm, like vertex reduction and edge reduction, are still applicable for \mbcs, they only ensure that $C_L$ and $C_R$ are not less than $k$ but lack of binding force of our new size bounds $\underline\kappa$, $\overline\kappa$ and $\epsilon$. Moreover, the value of $\overline\kappa$, $\underline\kappa$ and $\epsilon$ are much larger than $k$, especially for $\overline\kappa$ and $\epsilon$, which makes the effectiveness of the reduction optimizations is limited in \kw{MBCS} problem. Hence, we extend the vertex reduction and edge reduction such that they can support tighter bounds to prune more vertices and edges. 

 \stitle{Vertex reduction variant.}  
% Meanwhile, for each search branch, the size of candidates directly determines the search space produced by current branch. Therefore, 
% We first prune the invalid vertices in candidates by  considering the local degree of vertices in candidates. We have the following lemma:
We first propose the vertex reduction variant with considering the  degree of vertices. We have the following lemma:

\begin{lemma}
\label{lem:candidateprune}
Given a signed network $G=(V,E^+,E^-)$, a search region $(\underline{\kappa}_i,\overline{\kappa}_i)$ and $\epsilon$, $\mathbb{G}=(\mathbb{V},\mathbb{E^+},\mathbb{E^-})$ is a subgraph of $G$, s.t., $(a)\forall v \in \mathbb{V}, min\{d^+_{\mathbb{G}}(v)+1,d^-_{\mathbb{G}}(v)\} \geq \underline{\kappa}_i,  max\{d^+_{\mathbb{G}}(v)+1,d^-_{\mathbb{G}}(v)\} \geq \overline{\kappa}_i$; $(b)\forall v \in \mathbb{V},d^+_{\mathbb{G}}(v)+d^-_{\mathbb{G}}(v)+1 > \epsilon $.
%\end{itemize}
We have $C^*_i \subseteq  \mathbb{G}$.
\end{lemma}

\begin{proof}
Given a search region $(\underline{\kappa}_i,\overline{\kappa}_i)$, let $C^*_i=\{C_L,C_R\}$ is the maximum balanced clique in current search region. Based on \refalg{mbcss}, $C^*_i=\{C_L,C_R\}$ should satisfy $min\{|C_L|,|C_R|\} \geq \underline{\kappa}_i$ and $max\{|C_L|,|C_R|\} \geq \overline{\kappa}_i$.  Meanwhile, based on \refdef{balancecommunity}, $\forall v\in C_L$, $d_{C^*_i}^+(v) = |C_L|-1$, $d_{C^*_i}^-(v) = |C_R|$. $\forall v\in C_R$, $ d_{C^*_i}^+(v) = |C_R|-1$, $d_{C^*_i}^-(v) = |C_L|$. Combining them, we get $min\{d^+_{C^*_i}(v)+1$, $d^-_{C^*_i}(v)\} \geq \underline{\kappa}$, $max\{d^+_{C^*_i}(v)+1$, $d^-_{C^*_i}(v)\} \geq \overline{\kappa}$.
% Here, \reflem{candidateprune}(a) can be proved. 
 Moreover, as $C^*_i$ is the current maximum balanced clique, the degree of vertices in  $C^*_i$ should not less than $\epsilon$.
%, \reflem{candidateprune}(b) can be proved as well.
\end{proof}

\stitle{Algorithm of \vertexreductionv.} Based on \reflem{candidateprune}, we can reduce the size of the candidate sets by continuously deleting the vertices that do not meet the degree constraints in \reflem{candidateprune}. We propose \vertexreductionv algorithm, the pseudocode is shown at \refalg{prunev}. It continuously removes vertices until all vertices left in candidate sets meet the degree constraints.  

\begin{algorithm}[t]
\caption{\small $\vertexreductionv(C_L,C_R,P_L, P_R, \epsilon, \underline{\kappa},\overline{\kappa}$)}
{
\small
\begin{algorithmic}[1]
% \STATE $\mathbb{G}$ is reduced by $C_L,C_R,P_L, P_R$;
\STATE $V' \leftarrow P_L\cup P_R$; $\mathbb{G} \leftarrow G(V')$ ;
\WHILE{$\exists v \in P_{L(R)}$, s.t. $min\{d^+_{\mathbb{G}}(v)+1+|C_{L(R)}|,d^-_{\mathbb{G}}(v)+|C_{R(L)}| \}<\underline{\kappa} $  or  $max\{d^+_{\mathbb{G}}(v)+1+|C_{L(R)}|,d^-_{\mathbb{G}}(v)+|C_{R(L)}|\} <\overline{\kappa} $ or $d^+_{\mathbb{G}}(v)+1+d^-_{\mathbb{G}}(v)+|C_L|+|C_R|\leq \epsilon$ }
\FOR{\textbf{each} $u \in N^+_{\mathbb{G}}(v)$}
\STATE $d^+_{\mathbb{G}}(u) \leftarrow d^+_{\mathbb{G}}(u) -1$;
\ENDFOR
\FOR{\textbf{each} $u \in N^-_{\mathbb{G}}(v)$}
\STATE $d^-_{\mathbb{G}}(u) \leftarrow d^-_{\mathbb{G}}(u) -1$;
\ENDFOR
\STATE $\mathbb{G} \leftarrow \mathbb{G} \setminus v$;
\STATE $P_{L(R)} \leftarrow P_{L(R)} \setminus v$;
\ENDWHILE
\end{algorithmic}
}
\label{alg:prunev}
\end{algorithm}

\begin{theorem}
The space complexity of \refalg{prunev} is $O(|P_L|+|P_R|+|E^+_{\mathbb{G}}|+|E^-_{\mathbb{G}}|)$.
The time complexity of \refalg{prunev} is $O(|P_L|+|P_R|+|E^+_{\mathbb{G}}|+|E^-_{\mathbb{G}}|)$, where $\mathbb{G}=G(P_L\cup P_R)$.
\end{theorem}
%\begin{proof}
%For the space complexity, \refalg{prunev} needs maintain a subgraph $\mathbb{G}$ which  consumes  $O(|P_L|+|P_R|+|E^+_{\mathbb{G}}|+|E^-_{\mathbb{G}}|)$ space.
% For the time complexity, in the worst case, each vertex in candidate sets is removed. When a vertex is removed, its neighbors' degree is updated.  Hence, it only needs linear time of $O(|P_L|+|P_R|+|E^+_{\mathbb{G}}|+|E^-_{\mathbb{G}}|)$
%
%
%\end{proof}

 \stitle{Edge reduction variant.} After the vertex reduction variant, we extend the edge reduction now. Inspired by the edge reduction technology utilized in \kw{MBCE}, we continue to explore the edge reduction under a certain search region. 
Reconsidering the edge common neighbor number introduced at \refdef{2hopdegree}, we have the following lemma:

\begin{lemma}
\label{lem:2hopdegree2}
Given a signed network $G$, a search region $(\underline{\kappa}_i,\overline{\kappa}_i)$ and $\epsilon$,  $\mathbb{G}=(\mathbb{V},\mathbb{E^+},\mathbb{E^-})$ is a subgraph of $G$, s.t., 
\begin{itemize}
      \item $\forall (u, v) \in \mathbb{E^+} \rightarrow min\{\delta^{++}_{\mathbb{G}}(u,v)+2,\delta^{--}_{\mathbb{G}}(u,v)\} \geq \underline{\kappa}_i \wedge max\{\delta^{++}_{\mathbb{G}}(u,v)+2,\delta^{--}_{\mathbb{G}}(u,v)\} \geq \overline{\kappa}_i$;
      \item $\forall (u, v)\in \mathbb{E^-} \rightarrow min\{\delta^{+-}_{\mathbb{G}}(u,v)+1,\delta^{-+}_{\mathbb{G}}(u,v)+1\}\geq \underline{\kappa}_i \wedge max\{\delta^{+-}_{\mathbb{G}}(u,v)+1,\delta^{-+}_{\mathbb{G}}(u,v)+1\}\geq \overline{\kappa}_i$.
\end{itemize}
Then, we have $C^*_i \subseteq  \mathbb{G}$
\end{lemma}

% \begin{proof}
%The lemma can be proved like \reflem{2hopdegree}.
% \end{proof}

\stitle{Algorithm of \edgereductionv.} Based on \reflem{2hopdegree2}, we propose \edgereductionv algorithm. Given a signed network $G=(V, E^+, E^-)$, a search region $(\underline{\kappa}_i,\overline{\kappa}_i)$ and $\epsilon$, before the search starts, it removes the invalid edges that do not meet the requirements for the edge common neighbor number in \reflem{2hopdegree2} until no more edges can be pruned. The pseudocode is omitted here. The time complexity of \edgereductionv is $O(m^{1.5})$.

\subsubsection{The Optimized Algorithm }
Utilizing the above optimization strategies, i.e., coloring-based branch pruning, vertex domination-based candidate pruning and vertex\&edge reduction variants, we propose our optimized algorithm \mbcsp  to search maximum balanced clique in a given search region. The pseudocode is shown at \refalg{mbcsp}. Given a search region $(\underline{\kappa}$,$\overline{\kappa})$, for each search branch in this region, it first prunes candidates in $P_L$ and $P_R$ by invoking \vertexreductionv algorithm (line 4). Then, the surviving candidate sets are judged to see whether it meets the size requirements (line 5-7). If the candidate sets and explored sets are both empty, we get a larger result and return it (line 8-11). Otherwise, \refalg{mbcsp}   tries to prune invalid branch by invoking \kw{ColoringPrune} (line 12-13). After that, it chooses the pivot and  distinguishes the four special cases for vertex domination to further prune candidates (line 18-27). Then, it continuously search larger balanced clique in the remaining candidates by recursively calling itself. When all search branches are finished, \refalg{mbcsp} can get the maximum balanced clique in the given search region.

\begin{algorithm}[t]
\caption{\small \mbcsp($G=(V, E^+, E^-), \epsilon$)}
{
\small
\begin{algorithmic}[1]

% \STATE $\epsilon \leftarrow 2k$;
 % $\overline{\epsilon} \leftarrow k$; $\underline{\epsilon} \leftarrow k$;

 \FOR {\textbf{each} $v \in V$}
% \STATE $C_L \leftarrow \{v_i\}$, $C_R \leftarrow \emptyset$
% \STATE $P_L \leftarrow N^+_G(v_i) \cap \{v_{i+1},\cdots,v_{n-1}\}$;
% \STATE $P_R \leftarrow N^-_G(v_i) \cap \{v_{i+1},\cdots,v_{n-1}\}$;
% \STATE $Q_L \leftarrow N^+_G(v_i) \cap \{v_0,\cdots,v_{i-1}\}$;
% \STATE $Q_R \leftarrow N^-_G(v_i) \cap \{v_0,\cdots,v_{i-1}\}$;
\STATE initialize $C_L, C_R, P_L, P_R, Q_L, Q_R$;

\STATE \mbcsup$(C_L, C_R, P_L, P_R, Q_L, Q_R,\epsilon)$;
 \ENDFOR

\vspace{0.1cm}
\hspace{-0.6cm} \textbf{Procedure} \mbcsup($C_L$, $C_R$, $P_L$, $P_R$, $Q_L$, $Q_R$, $\epsilon$)

\hspace{-0.2cm}//$~$\textbf{Vertex Reduction Varint$\quad$}
\STATE $\vertexreductionv(C_L, C_R, P_L, P_R, \epsilon, \underline{\kappa},\overline{\kappa})$; 

\STATE  $\overline{L} \leftarrow |C_L|+|P_L|$; $\overline{R}  \leftarrow |C_R|+|P_R|$;

\IF{$\overline{L} +\overline{R} \leq \epsilon$  or $min\{\overline{L} ,\overline{R}\} < \underline{\kappa}$ or  $max\{\overline{L} ,\overline{R}\} < \overline{\kappa}$}
% or $\overline{L} <k$ or $\overline{R} <k$ }

\STATE \textbf{return};
\ENDIF
\IF{$P_L=\emptyset$ and $P_R=\emptyset$ and $Q_L=\emptyset$ and $Q_R=\emptyset$}
% \IF{$|C_L|+|C_R|>\epsilon$}
\STATE $C^* \leftarrow \{C_L, C_R\}$;
\STATE $\epsilon \leftarrow |C_L|+|C_R|$; 
% $\underline{\epsilon} \leftarrow min\{|C_L|, |C_R|\}$; $\overline{\epsilon} \leftarrow max\{|C_L|, |C_R|\}$;
% \ENDIF
\STATE \textbf{return};
\ENDIF

%\STATE early termination of $\mbcenump$ (line 7-12 of \refalg{fmbcp});

\hspace{-0.2cm}//$~$\textbf{Coloring-based Branch Pruning$\quad$}
\IF{$\kw{ColoringPrune}(C_L,C_R,P_L, P_R , \epsilon, \underline{\kappa},\overline{\kappa})$}
\STATE \textbf{return};
\ENDIF

\STATE $p \leftarrow \kw{argmax}_{v \in P_L \cup P_R \cup Q_L \cup Q_R}\{d_l(v)\}$;
\STATE // assume $p$ from $P_L \cup Q_L$ ($P_R \cup Q_R$ ) 
\STATE $newP_L \leftarrow P_L \setminus N^{+(-)}_G(p)$;
\STATE $newP_R \leftarrow P_R \setminus N^{-(+)}_G(p)$;

\hspace{-0.2cm}//$~$\textbf{Vertex Domination-based Candidate Pruning$\quad$}

\hspace{-0.2cm}//$~$\textbf{Case 1$\quad$}
\IF{$p\in  Q_{L(R)}$ and $|newP_L|+|newP_R|=1$} 
\STATE \textbf{return};
\ENDIF

\hspace{-0.2cm}//$~$\textbf{Case 2$\quad$}
\IF{$p\in  Q_{L(R)}$ and $|newP_L|+|newP_R|=2$}
\IF{$w,q \in newP_L\cup newP_R$ and $(w,q) \notin E$}
%\IF{$w \in newP_L, q \in newP_R$ and $(w,q) \notin E^-$}
\STATE \textbf{return};
%\ElsIf{$w,q \in newP_{L(R)}$ and $(w,q) \notin E^+$}
%\STATE \textbf{return};
\ENDIF
\ENDIF

\hspace{-0.2cm}//$~$\textbf{Case 3$\quad$}
\IF{$p\in  P_{L(R)}$ and $|newP_L|+|newP_R|=2$}
\STATE $newP_{L(R)} \leftarrow \{p\} $; $newP_{R(L)}  \leftarrow \emptyset $;
\ENDIF

\hspace{-0.2cm}//$~$\textbf{Case 4$\quad$}
\IF{$p\in P_{L(R)}$ and $|newP_L|+|newP_R|=3$}
\IF{$w,q \in newP_{L}\cup newP_{R}$ and $(w,q) \notin E$}
\STATE $newP_{L(R)} \leftarrow  \{p\}$; $newP_{R(L)} \leftarrow \emptyset$;
%\IF{$w \in newP_L, q \in newP_R$ and $(w,q) \notin E^-$}
%\STATE $newP_L \leftarrow newP_L \setminus \{w\}$; $newP_R \leftarrow newP_R \setminus \{q\}$;
%\ElsIf{$w,q \in newP_{L(R)}$ and $(w,q) \notin E^+$}
%\STATE $newP_{L(R)} \leftarrow newP_{L(R)} \setminus \{w,q\}$;
\ENDIF
\ENDIF

\STATE search result within $newP_L$ and $newP_R$ as \mbcenum;
%line 18 of \mbcs with replacing \mbcsu with \mbcsup ;

% \FOR{\textbf{each} $v \in$ \kw{newP_L}}
% \STATE $\mbcsu(C_L \cup \{v\}, C_R, N^+_G(v) \cap P_L, N^-_G(v) \cap P_R, N^+_G(v) \cap Q_L, N^-_G(v) \cap Q_R)$;
% \STATE $P_L \leftarrow P_L \setminus \{v\}$; $Q_L \leftarrow Q_L \cup \{v\}$;
% \ENDFOR

% \FOR{\textbf{each} $v \in$ \kw{newP_R}}
% \STATE $\mbcsu(C_L, C_R \cup \{v\}, N^-_G(v) \cap P_L, N^+_G(v) \cap P_R, N^-_G(v) \cap Q_L, N^+_G(v) \cap Q_R)$;
% \STATE $P_R \leftarrow P_R \setminus \{v\}$; $Q_R \leftarrow Q_R \cup \{v\}$;
% \ENDFOR

\end{algorithmic}
}
\label{alg:mbcsp}
\end{algorithm}

% \begin{example}
% \end{example}

%\subsubsection{The optimized \mbcssp Algorithm}
Based on \mbcsp algorithm, we are ready to propose the formal optimized algorithm \mbcssp, the pseudocode is shown at \refalg{mbcssp}. Given a signed network $G=(V, E^+, E^-)$ and size threshold $k$, for each search region $(\overline{\kappa},\underline{\kappa})$, the algorithm first invokes \edgereductionv to reduce the graph size by removing invalid edges  before search start (line 4). Then,  it invokes \refalg{mbcsp} to search the maximum balanced clique in current search region (line 5). In the end, when finish the search in all search regions, it gets the maximum balanced clique $C^*$ in $G$ and terminates.

\begin{algorithm}[t]
\caption{\small \mbcssp($G=(V, E^+, E^-),k$)}
{
\small
\begin{algorithmic}[1]
\STATE compute degeneracy $\sigma$ of $G^+=(V, E^+)$;
\STATE $\epsilon \leftarrow 2k$; $\underline{\kappa} \leftarrow k$;  $\overline{\kappa} \leftarrow \sigma + 1$; $\overline{\kappa'} \leftarrow -1$;
%\While{$\overline{\kappa}\geq \underline{\kappa}$ and $\overline{\kappa}!=\overline{\kappa'}$ }
\WHILE{$\overline{\kappa}\geq \underline{\kappa}$ 
and $\overline{\kappa}<\overline{\kappa}'$ }

\STATE $G' \leftarrow \edgereductionv(G, \epsilon, \underline{\kappa},\overline{\kappa}$);

\STATE \mbcsp$(G',\epsilon)$;

%\STATE \textbf{if} $\underline{\kappa}=\overline{\kappa}$ \textbf{then}  \textbf{break};

\STATE  $\overline{\kappa'} \leftarrow \overline{\kappa}$;
 $\underline{\kappa} \leftarrow max\{\epsilon - \overline{\kappa}, k\}$;  $\overline{\kappa} \leftarrow max\{\lceil\frac{\overline{\kappa}}{2}\rceil, \underline{\kappa}\}$;

\ENDWHILE

\end{algorithmic}
}
\label{alg:mbcssp}
\end{algorithm}

\section{Performance Studies}
\label{sec:performance}
\newskip\subfigcapskip \subfigcapskip = -4pt
In this section, we present our experimental results.All the experiments are performed  on a machine with two Intel Xeon 2.2GHz CPUs and 64GB RAM running CentOS 7.

\stitle{Algorithms.}
 We evaluate \kw{MBCE} algorithms and \kw{MBCS} algorithms.

For \kw{MBCE} algorithms, they are 
 \baseline, \kw{MBCEnum} and \kw{MBCEnum^*}.  \baseline is the baseline solution shown in  \refsec{baseline}. \kw{MBCEnum} is our algorithm shown in  \refsec{improveapproch}. \kw{MBCEnum^*} is the algorithm with the in-enumeration optimization shown in  \refsec{localpruning}. Note that the pre-enumeration optimization strategies can be also used in  \baseline and \kw{MBCEnum}, thus, we apply them for all three algorithms for fairness. 

For \kw{MBCS} algorithms, they are \mbcs, \mbcss and \mbcssp. \mbcs is the baseline approach introduced at \refsec{baseline2}. \mbcss is proposed at \refsec{spacesplit}. \mbcssp  is the improved algorithm shown at \refsec{optimization2}. Note that, for fair,  we apply the \edgereductionv proposed at \refsec{preenumeration2} to both  \mbcss and \mbcssp.

All algorithms are implemented in C++, using  g++ complier with -O3. The time cost is measured as the amount of wall-clock time elapsed during the program's execution.  If an algorithm cannot finish in 12 hours, we denote the processing time as \kw{INF}. 
% For reproducibility, the code is anonymously released and can be downloaded \footnote{\url{https://drive.google.com/file/d/18L538NMN7xDO6Xu0CMa2oSYmx07K5vZ5/view?usp=sharing}}.

%\subsection{Experiments on Real Datasets}
\stitle{Real datasets.} We evaluate our algorithms on nine real  datasets. \kw{Slashdot} and \kw{Epinions} are signed networks in real world. AdjWordNet, \kw{DBLP} and \kw{Douban} are signed networks used in \cite{chu2016finding}, \cite{li2018efficient} and \cite{xu2012towards}, respectively. For other datasets, we transfer them from unsigned network to signed network. In order to simulate the balanced clique as much as possible, the vertices in the graph are divided into two groups with a ratio of 4:1, the edges connecting vertices from the same group are positive edges, otherwise, they are negative edges. In this way, all the cliques in the original graph correspond to  balanced cliques in the signed graph. For data source,  \kw{AdjWordNet} is downloaded from WordNet (\url{https://wordnet.princeton.edu/}). \kw{DBLP} and \kw{Dbpedia} are downloaded from KONECT (\url{http://konect.cc/}).  \kw{Douban} is from authors in \cite{xu2012towards}. Other datasets are downloaded from SNAP (\url{http://snap.stanford.edu}). The details of each dataset are shown in Table \ref{tab:dataset}.
%, and combined with the characteristics of large positive edges and less negative edges in the real sighed graph,
%Then, we evaluate our MBCS algorithms on both small datasets and large datasets. 

\begin{table}
\topcaption{Statistic for real datasets}
\label{tab:dataset}
\centering
\def\arraystretch{1.1}
\setlength{\tabcolsep}{0.30em}
{\small
\begin{tabular}{c|c|c|c|c}
\hline
 Dataset&  $n$ &  $m$&  $|E^+|$ &  $|E^-|$ \\%&$d^+_{max}$&$d^-_{max}$\\
\hline
 \kw{AdjWordNet} & 21,247&426,896&378,993&47,903\\
 \kw{Slashdot} &77,357&516,575&396,378&120,197\\
 \kw{Epinions} &131,828&841,372&717,667&123,705\\
%  \kw{Youtube}&1,134,890&2,987,624&2,041,101&946,523\\
 \kw{DBLP} &1,314,050&5,179,945&1,471,903&3,708,042\\
%  \kw{Skitter}&1,696,415&11,095,298&7,535,667&3,559,631\\
 \kw{Douban} &1,588,565&13,918,375&9,034,537&4,883,838\\
%\hline
%\end{tabular}
%}
%\vspace{-0.4cm}
%\end{table}
%
%\begin{table}
%\topcaption{Statistic for large datasets}
%\label{tab:dataset2}
%\centering
%\def\arraystretch{1.1}
%\setlength{\tabcolsep}{0.50em}
%{\small
%\begin{tabular}{c|c|c|c|c}
%\hline
% Dataset&  $n$ &  $m$&  $|E^+|$ &  $|E^-|$ \\%&$d^+_{max}$&$d^-_{max}$\\
%\hline
 \kw{Pokec}&1,632,803&30,622,564&15,179,203&7,122,761\\
  \kw{Livejournal}&4,847,571&42,851,237&29,105,031&13,746,206\\
 \kw{Orkut}&3,072,441&117,184,899&79,664,169&37,520,730\\
  \kw{Dbpedia}&18,268,992&126,890,209&86,002,736&40,887,473\\
 \hline
\end{tabular}
}
\vspace{-0.4cm}
\end{table}

\subsection{The Performance of MBCE Algorithms}

\begin{figure}[t]
%\setlength{\abovecaptionskip}{20pt}
%\captionsetup[subfigure]{aboveskip=-1pt,belowskip=-1pt}
\begin{center}
\subfigure[\kw{Slashdot} (vary $k$)]{
\label{fig:time1}
\centering
\includegraphics[width=0.49\columnwidth]{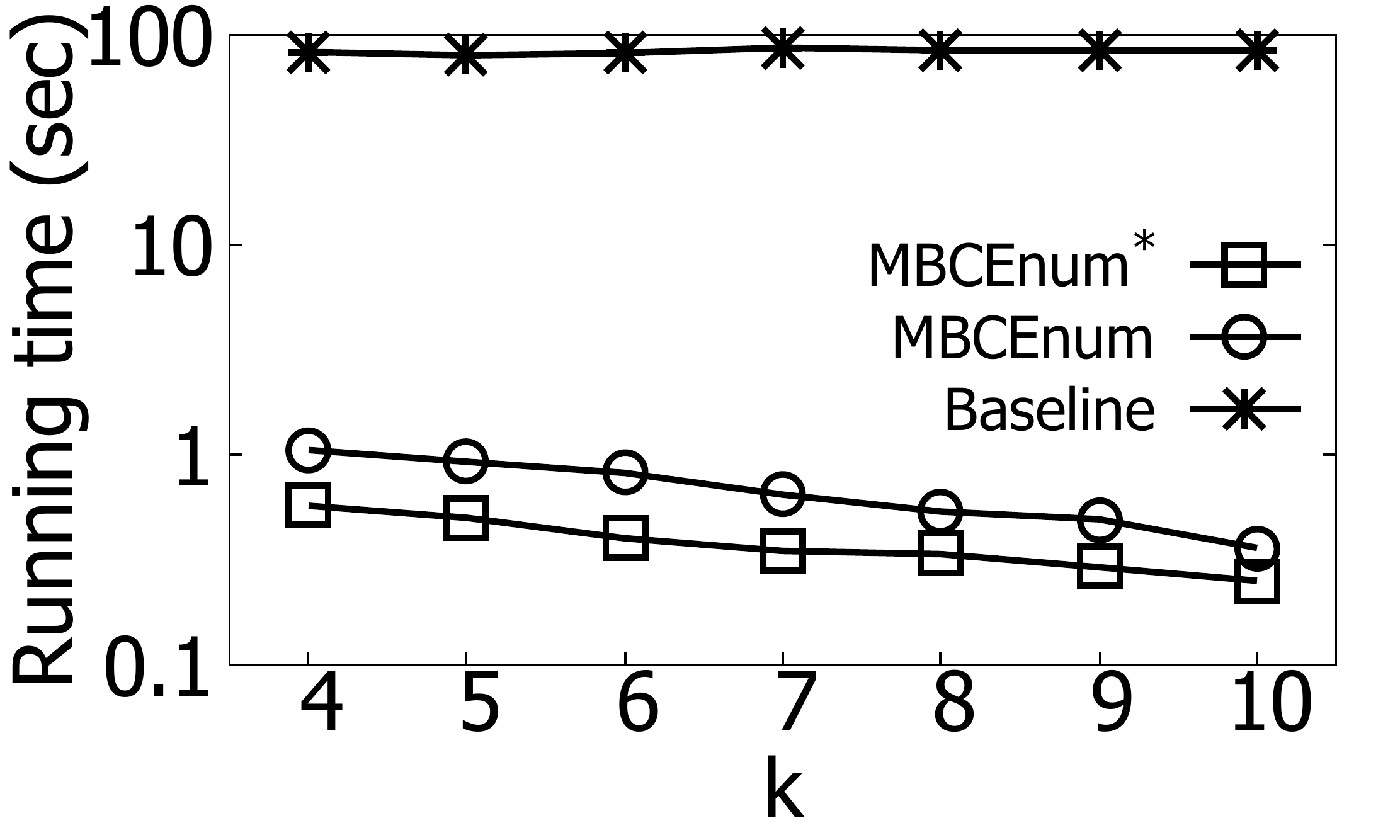}
}\vspace{-0.2cm}\subfigure[\kw{Epinions} (vary $k$)]{
\label{fig:time2}
\centering
\includegraphics[width=0.49\columnwidth]{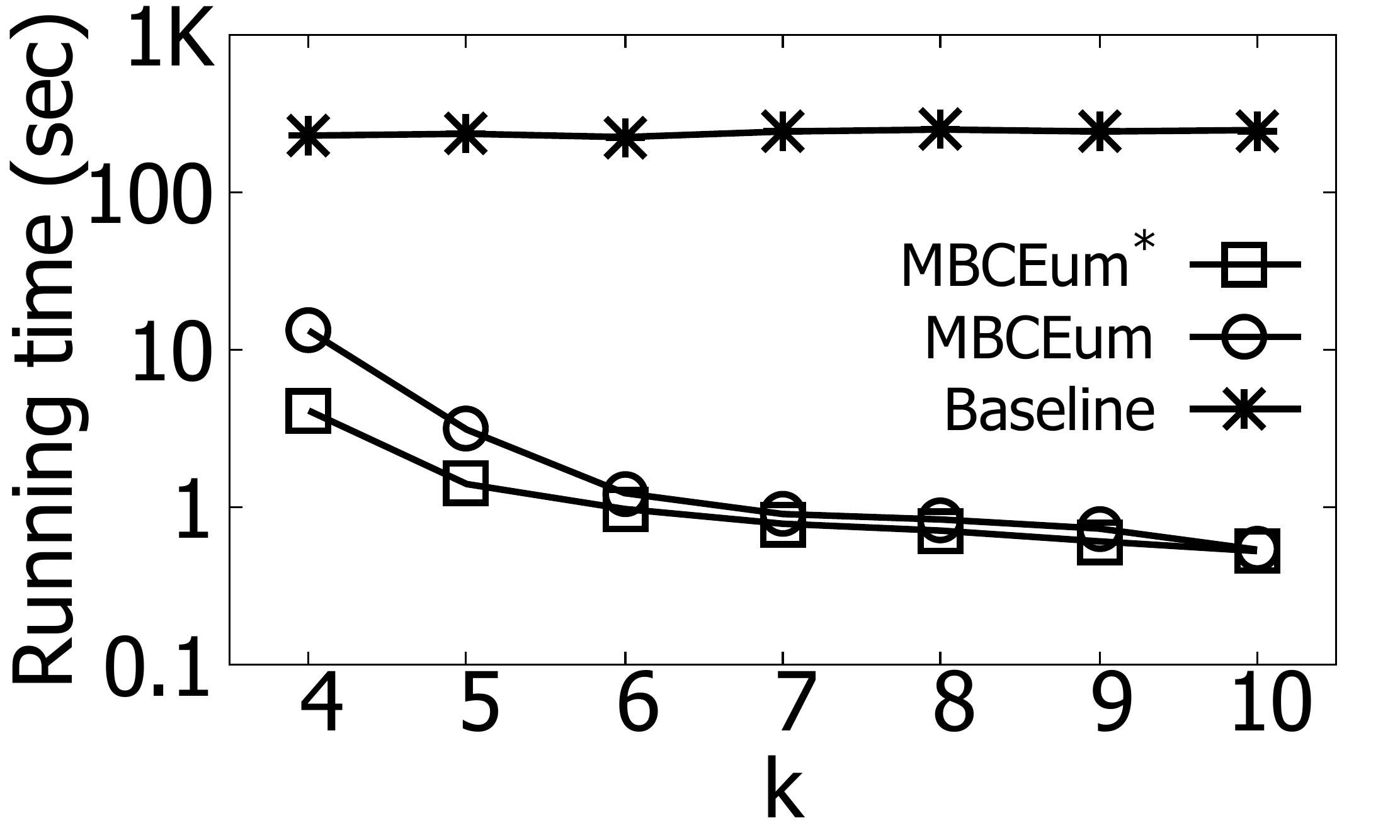}
}\\
% \subfigure[Wiki (Vary $k$)]{
% \label{fig:time3}
% \centering
% \includegraphics[width=0.49\columnwidth]{fig/runtime_wiki.pdf}
% }\vspace{-0.2cm}
\subfigure[\kw{DBLP} (Vary $k$)]{
\label{fig:time5}
\centering
\includegraphics[width=0.49\columnwidth]{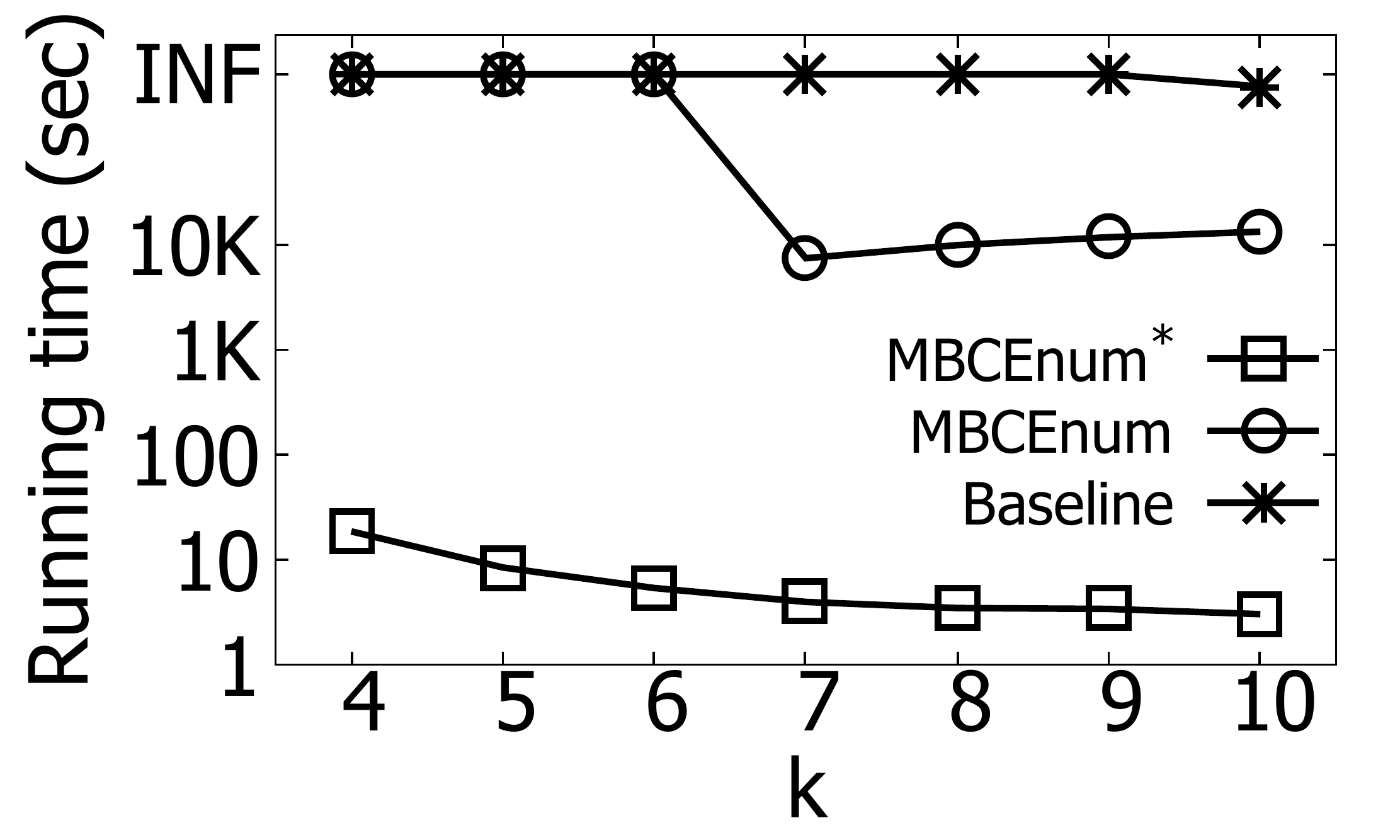}
}\vspace{-0.2cm}\subfigure[\kw{Douban} (Vary $k$)]{
\label{fig:time6}
\centering
\includegraphics[width=0.49\columnwidth]{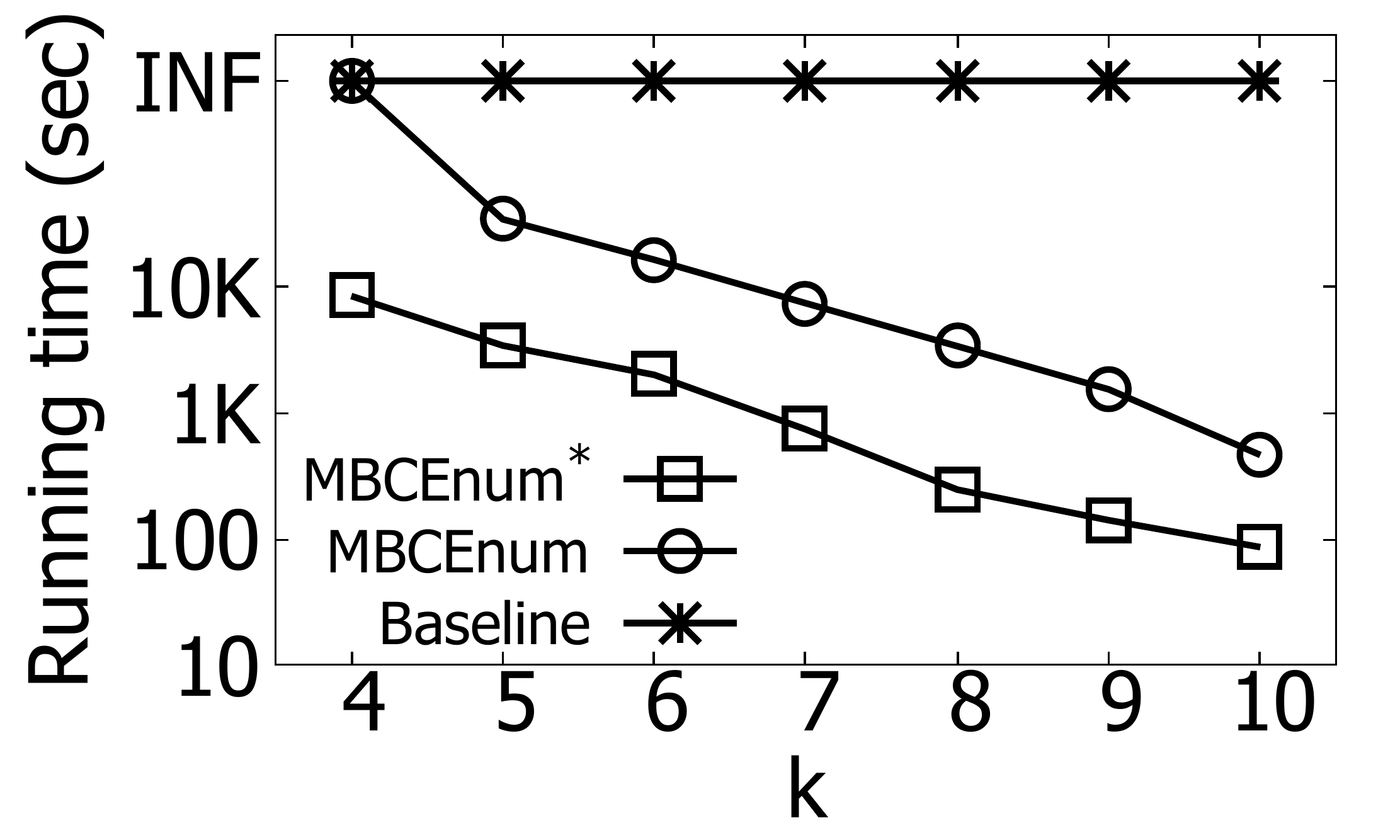}
}
% \vspace{-0.2cm}\subfigure[Youtube (Vary $k$)]{
% \label{fig:time4}
% \centering
% \includegraphics[width=0.49\columnwidth]{fig/runtime_youtube_label3.pdf}
% }\\
% \subfigure[Pokec (Vary $k$)]{
% \label{fig:time7}
% \centering
% \includegraphics[width=0.49\columnwidth]{fig/runtime_pokec_label3.pdf}
% }
% \vspace{-0.2cm}
% \subfigure[Orkut (Vary $k$)]{
% \label{fig:time8}
% \centering
% \includegraphics[width=0.49\columnwidth]{fig/runtime_orkut.pdf}
% }
\end{center}
\topcaption{Running time of  \kw{MBCE} algorithms varying $k$}
\label{fig:totaltime}
\vspace{-0.2cm}
\end{figure}

\stitle{Exp-1: Efficiency  of  \kw{MBCE} algorithms when varying $k$.} In this experiment, we evaluate the efficiency of three algorithms when varying $k$ from 4 to 10 and the results are shown in \reffig{totaltime}.

As shown in \reffig{totaltime}, \baseline consumes the most  time among  three algorithms on all datasets when we vary $k$ and it can only handle the small datasets. 
%For example, on \kw{Slashdot} (\reffig{totaltime} (a)), \kw{MBCEnum} and \kw{MBCEnum^*} are at least two orders of magnitude faster than \baseline. On \kw{Douban} (\reffig{totaltime} (d)), \baseline cannot finish the enumeration in 12 hours. This is because \baseline does not consider the uniqueness of the signed networks and lots of unnecessary computations are involved in the enumeration of \baseline. 
\kw{MBCEnum} is faster than \baseline on most of the test cases as \kw{MBCEnum} takes the uniqueness of the signed networks  into consideration and enumerates the maximal balanced cliques based on the signed network directly. \kw{MBCEnum^*} is the most efficient algorithm on all datasets when varying $k$ due to the utilization of in-enumeration optimization strategies, which reveals the effectiveness of in-enumeration optimization strategies. Another phenomena shown in \reffig{totaltime} is that the running time of all algorithms decreases as $k$ increases. This is because as $k$ increases, the pruning power of the optimization strategies proposed in \refsec{optimization} strengthens. %Consequently, more unnecessary computations are avoided.

\begin{figure}[t]
\begin{center}
\subfigure[\kw{DBLP} (Vary $n$)]{
\includegraphics[width=0.49\columnwidth]{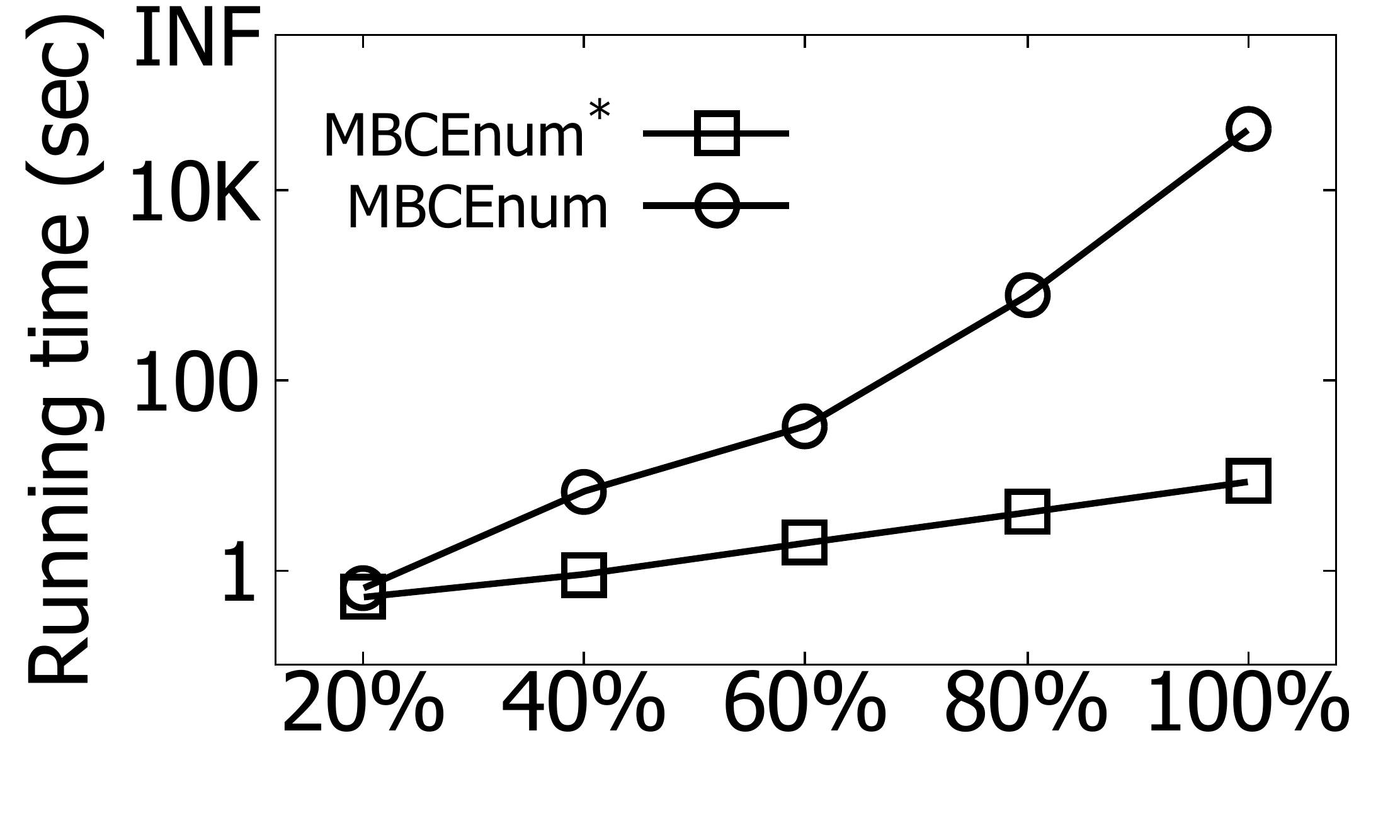}
}\subfigure[\kw{Douban} (Vary $n$)]{
\includegraphics[width=0.49\columnwidth]{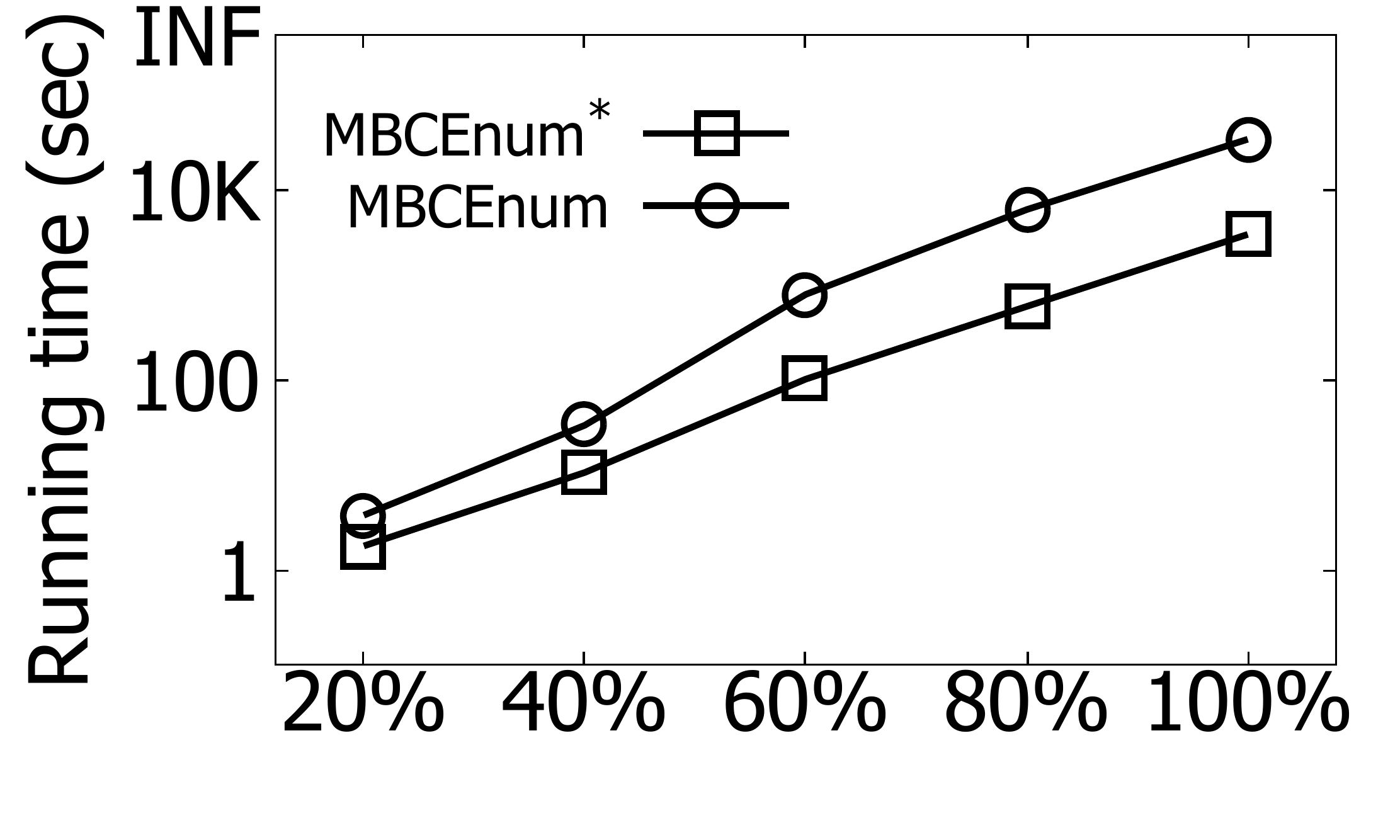}
}
\end{center}
%\vspace{-0.2cm}
\topcaption{Scalability of \kw{MBCEnum} and \kw{MBCEnum^*}, $k$=4}
\label{fig:scala}
%\vspace{-0.1cm}
\end{figure}

%\begin{figure}[t]
%\subfigure[Pokec]{
%\includegraphics[width=0.49\columnwidth]{fig/scal_pokec.pdf}
%}\subfigure[Orkut]{
%\includegraphics[width=0.49\columnwidth]{fig/scal_orkut.pdf}
%}
%\topcaption{Scalability of \kw{MBCEnum} and \mbcenump ($k=5$)}
%\label{fig:scala}
%\end{figure}

\stitle{Exp-2: Scalability of \kw{MBCE} algorithms.} In this experiment, we test the scalability of \kw{MBCEnum} and \kw{MBCEnum^*} on two large datasets \kw{DBLP} and \kw{Douban}  by varying their vertices from 20\% to 100\%. \reffig{scala} shows the results.

As shown in \reffig{scala}, when $n$ increases, the running time of both algorithms increases as well, but \kw{MBCEnum^*} outperforms \kw{MBCEnum} for all cases on both datasets. For example, on \kw{DBLP}, when we sample $20\%$ vertices, the running time of \kw{MBCEnum} and \kw{MBCEnum^*} is 0.6 seconds and 0.5 seconds, respectively, while when sampling $80\%$ vertices,  their running times are 770.6 seconds and 4.0 seconds, respectively. It shows that \kw{MBCEnum^*} has a good scalability in practice.

\begin{table}[t]
\topcaption{Case study on \kw{AdjWordNet}}
\label{tab:case}
\centering
\def\arraystretch{1.1}
\setlength{\tabcolsep}{0.66em}
\begin{tabular*}{\linewidth}{ p{0.45\linewidth}| p{0.45\linewidth}}
\hline
    $C_L$&    $C_R$\\%&$d^+_{max}$&$d^-_{max}$\\
\hline
raw, rough, rude&refined, smooth, suave\\
\hline
relaxing, reposeful, restful&restless,uneasy, ungratified,

unsatisfied\\
\hline
interior, internal, intimate&away, foreign, outer, outside,

 remote\\
\hline
assumed, false, fictitious, fictive, mistaken, off-key, pretended, put-on, sham, sour, untrue&actual, existent, existing, factual, genuine, literal, real, tangible, touchable, true, truthful, unfeigned, veridical\\
\hline
active, animated, combat-ready, dynamic, dynamical, fighting, participating, alive, live
 &adynamic, asthenic, debilitated, enervated, undynamic, stagnant, light\\
\hline
following, undermentioned, next&ahead, in-the-lead, leading, preeminent, prima, star, starring, stellar\\
\hline
undesirable, unsuitable, unwanted&cherished, treasured, wanted, precious\\
\hline
\end{tabular*}
\vspace{-0.4cm}
\end{table}

\stitle{Exp-3: Case study on AdjWordNet.} In this experiment, we perform a case study on the real dateset \kw{AdjWordNet}. In this dataset, two synonyms have a positive edge and two antonyms have a negative edge, and \reftable{case} shows some results obtained by our algorithm. As shown in \reftable{case}, words in $C_L$ or $C_R$ have similar meaning while each word from $C_L$ is an antonym to all words in $C_R$. This case study verifies that  maximal balanced clique enumeration can be applied in the applications to find synonym and antonym groups on dictionary data.

\subsection{The Performance of MBCS Algorithms}
%\label{sec:performance2}
% \newskip\subfigcapskip \subfigcapskip = -4pt
\stitle{Exp-4: Efficiency of \kw{MBCS} algorithms when varying $k$.}
To evaluate the efficiency of \kw{MBCS} algorithms, we record the running time of them on eight datasets, $k=[2-10]$, the results are shown at \reffig{totaltime2}.

As shown at \reffig{totaltime2}, with the value of $k$ increasing, the running time of three algorithms decreases on most datasets. For each value of $k$, \mbcssp is the fastest algorithm of the three algorithms, while \mbcs is the most time-consuming. 
%For instance, on Douban, when $k$=6, the running time of \mbcs, \mbcss, \mbcssp is 1285.2s, 632.6s and 191.7s, respectively. 
Moreover, on large graphs as Livejournal, Orkut and Dbpedia, when $k$=2, \mbcs and \mbcss can not get the result within a reasonable time, only \mbcssp can get the result. The running time of \mbcssp on the three graphs are 413.4s, 3626.0s and 1833.6s, respectively. It's because that \mbcs has to search on the whole graph, while \mbcss uses search space partition to search the maximum balanced clique within a partial subgraph. Based on \mbcss, benefited from multiple optimization strategies for further reducing the search space, \mbcssp is the most efficient algorithm of them, it can efficiently search result on all datasets.

%on Youtube, as shown at \reffig{time33}, when $k$=2，the running time of three algorithms is 220.9s,126.3s and 14.4s, respectively. \mbcssp can obtain speedup of an order of magnitude to \mbcs. 

%As shown at \reffig{totaltime2}, Another trend of running time is, with the value of $k$ increasing, the running time of three algorithms decreases on most datasets. Because with larger value of $k$, more useless edges and vertices can be pruned by the optimization strategies related to $k$ during the search process, which lead to smaller search space and shorter computing time.

\begin{figure}[htp]
%\setlength{\abovecaptionskip}{20pt}
%\captionsetup[subfigure]{aboveskip=-1pt,belowskip=-1pt}
\begin{center}

\subfigure[ \kw{Slashdot} (Vary $k$)]{
\label{fig:time11}
\centering
\includegraphics[width=0.49\columnwidth]{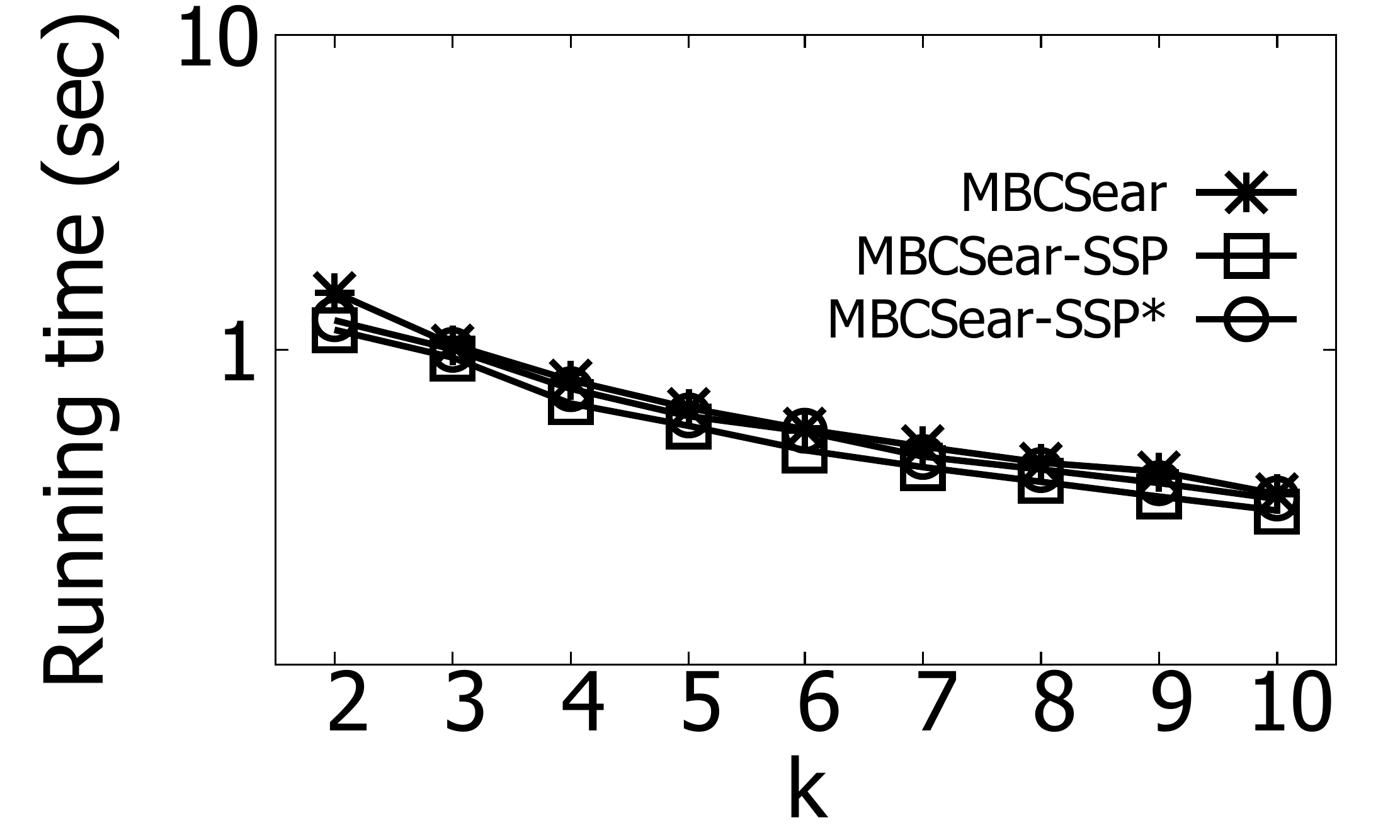}
}\subfigure[  \kw{Epinions} (Vary $k$)]{
\label{fig:time22}
\centering
\includegraphics[width=0.49\columnwidth]{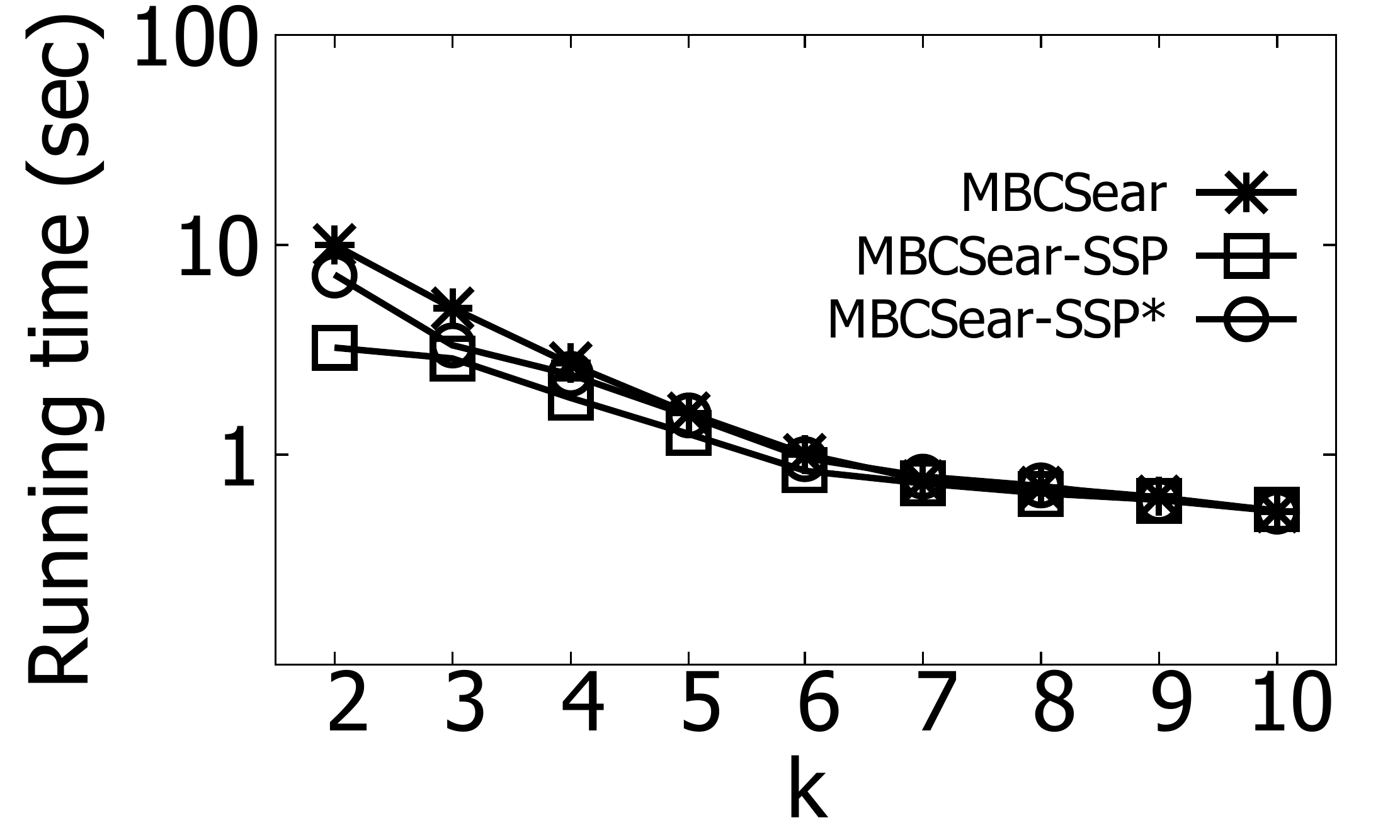}
}\\
%\subfigure[  Youtube (Vary $k$)]{
%\label{fig:time33}
%\centering
%\includegraphics[width=0.49\columnwidth]{MaximumBC/fig/runtime_youtube.pdf}
%}
\vspace{-0.3cm}\subfigure[  \kw{DBLP} (Vary $k$)]{
\label{fig:time44}
\centering
\includegraphics[width=0.49\columnwidth]{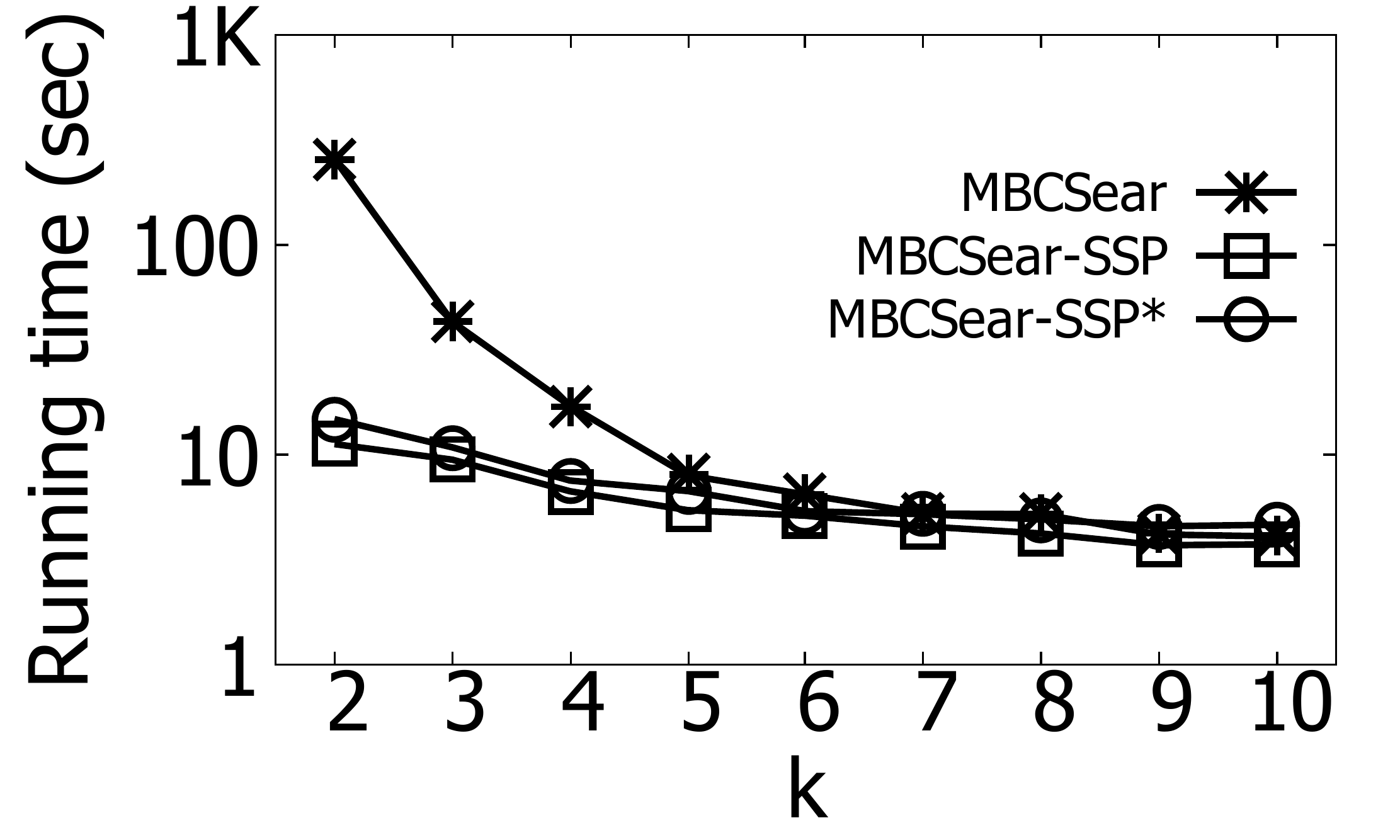}
}\subfigure[   \kw{Douban} (Vary $k$)]{
\label{fig:time55}
\centering
\includegraphics[width=0.49\columnwidth]{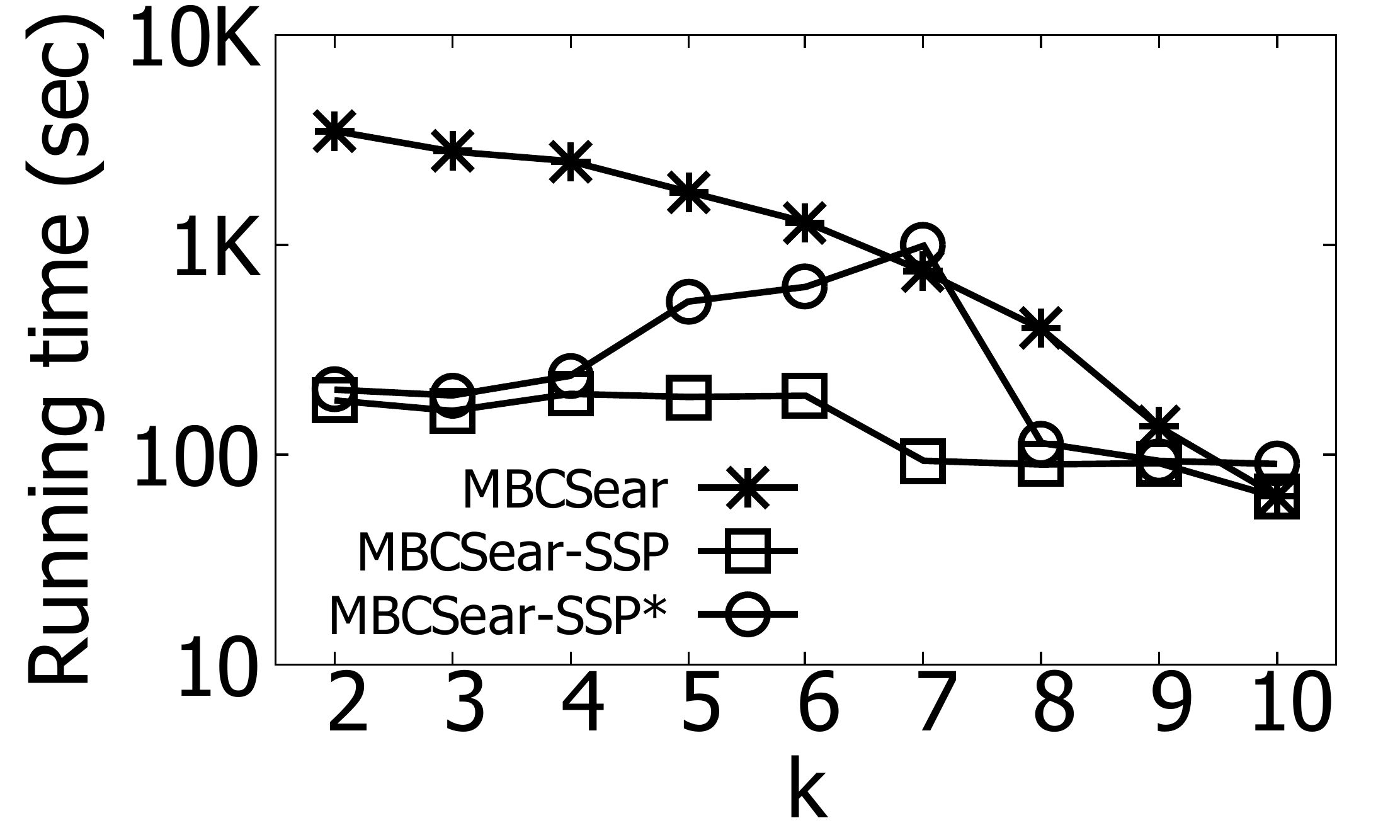}
%\subfigure[   \kw{Skitter} (Vary $k$)]{
%\label{fig:time55}
%\centering
%\includegraphics[width=0.49\columnwidth]{MaximumBC/fig/runtime_skitter.pdf}
}
\\\vspace{-0.3cm}
\subfigure[  \kw{Pokec} (Vary $k$)]{
\label{fig:time66}
\centering
\includegraphics[width=0.49\columnwidth]{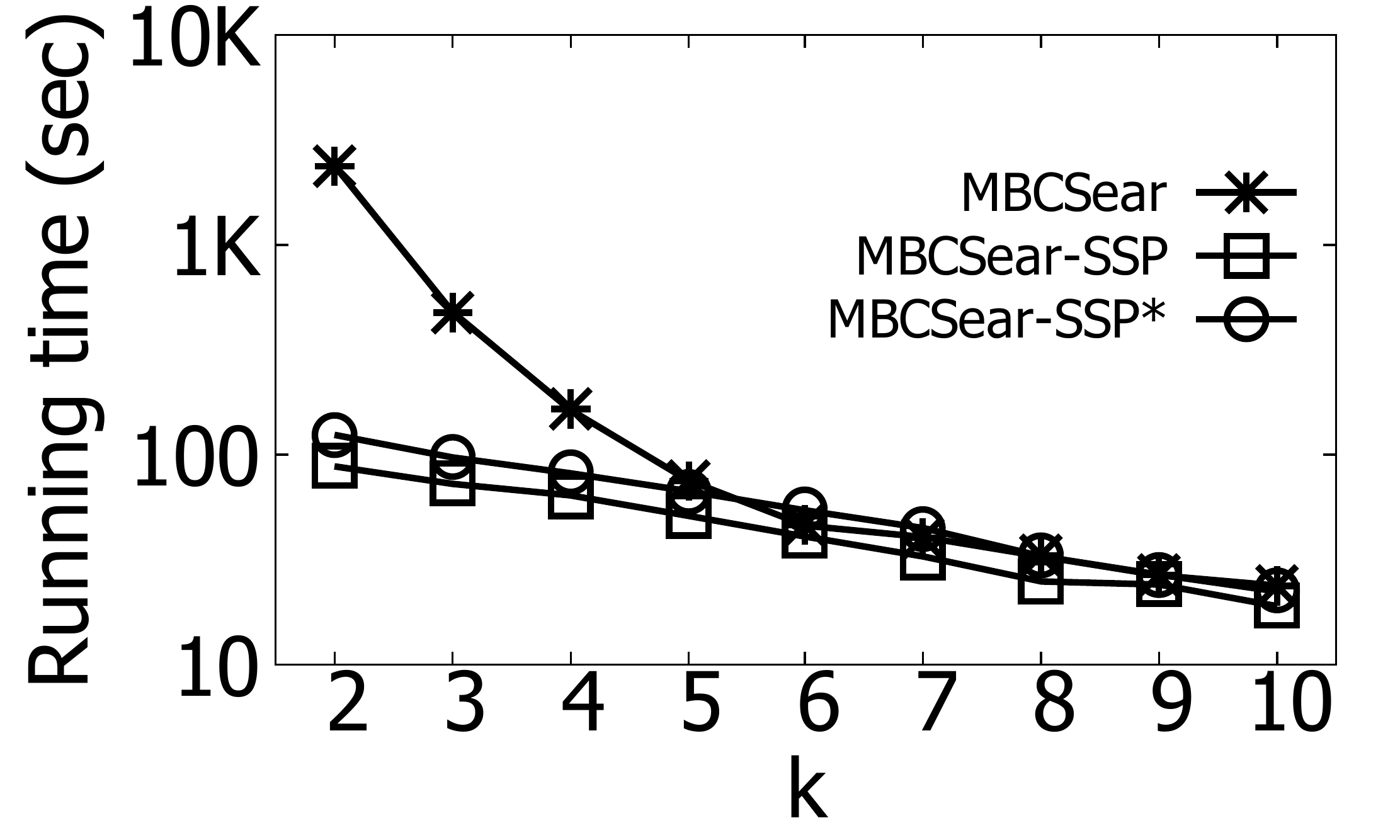}
}\subfigure[  Livejournal (Vary $k$)]{
\label{fig:time77}
\centering
\includegraphics[width=0.49\columnwidth]{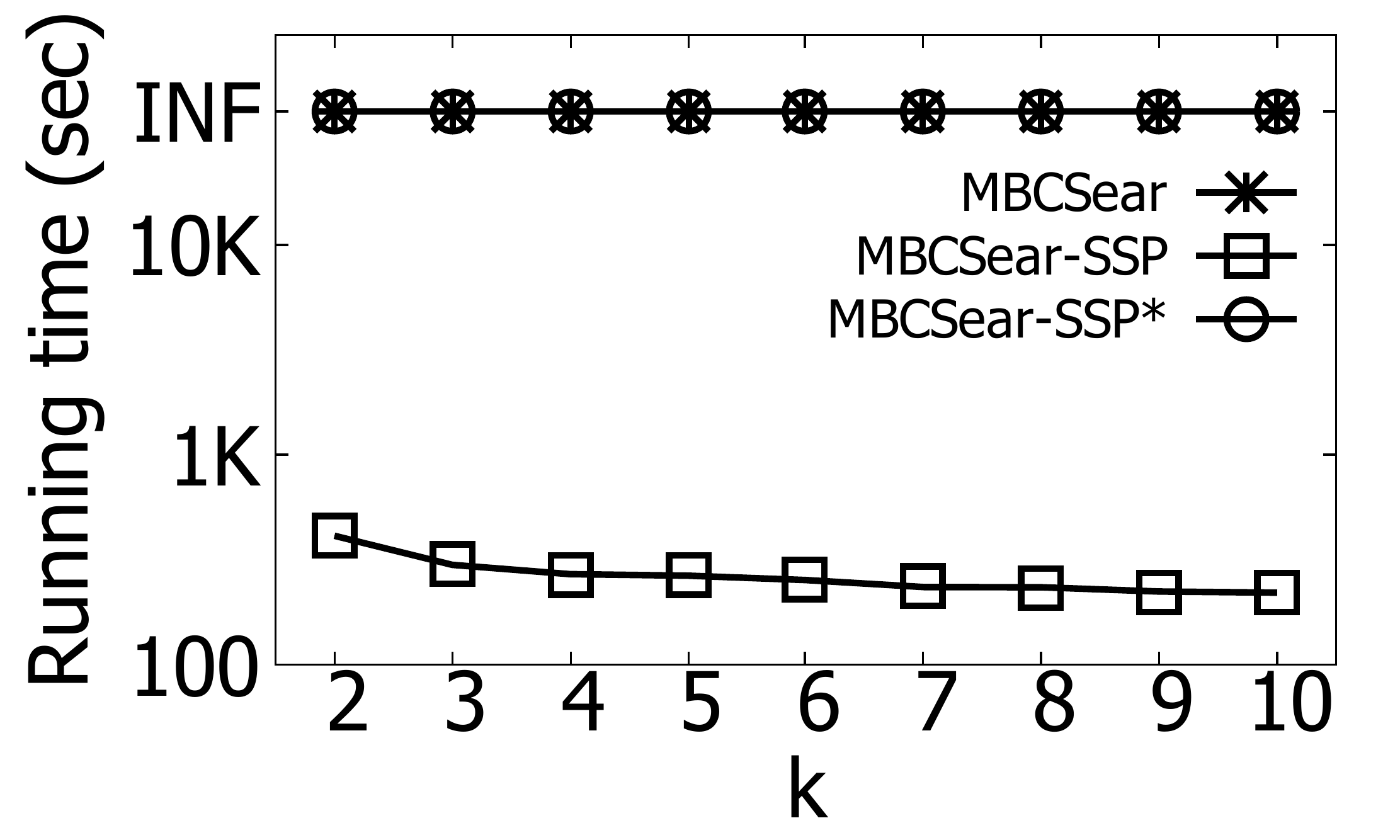}
}
\\\vspace{-0.3cm}
\subfigure[ \kw{Orkut} (Vary $k$)]{
\label{fig:time88}
\centering
\includegraphics[width=0.49\columnwidth]{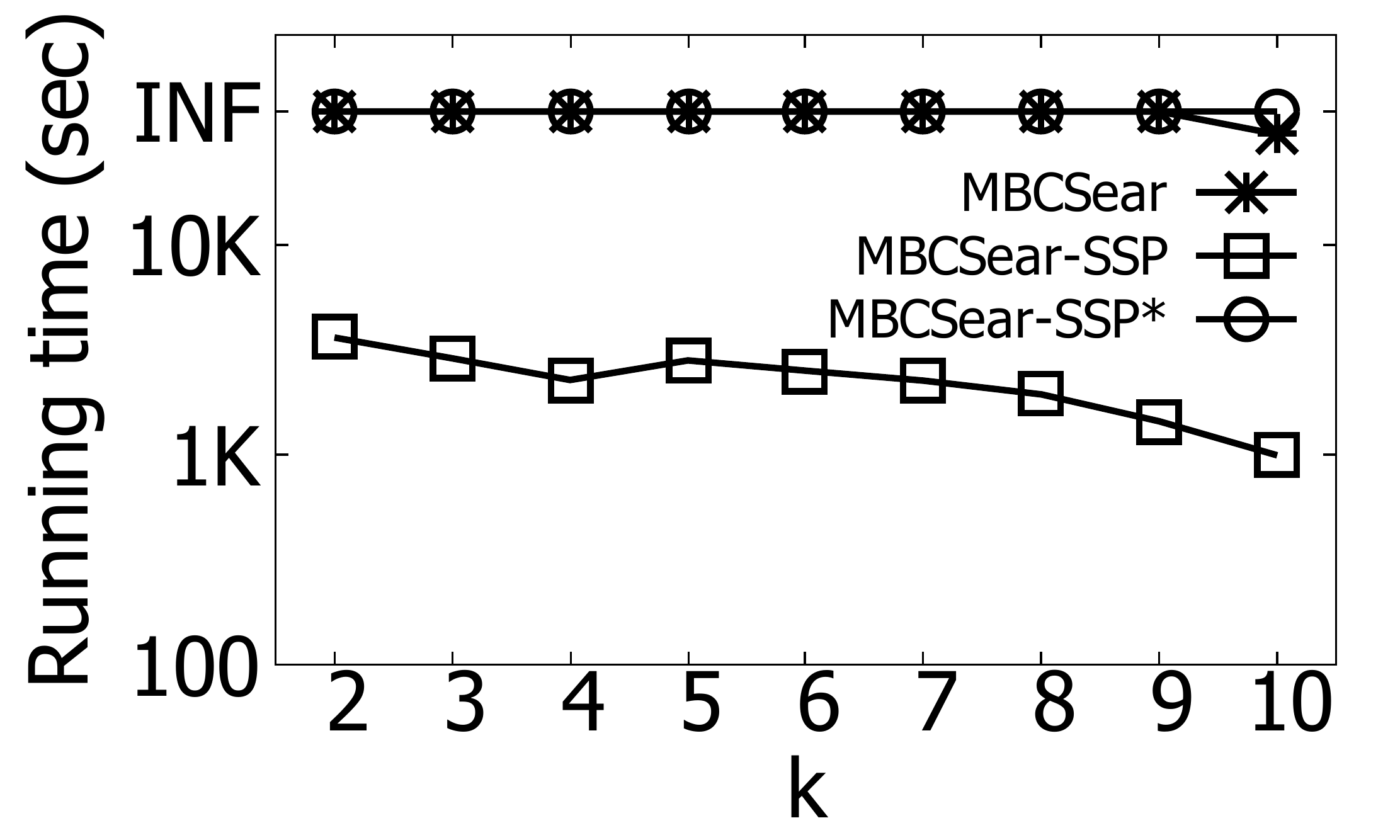}
}\subfigure[Dbpedia (Vary $k$)]{
\label{fig:time99}
\centering
\includegraphics[width=0.49\columnwidth]{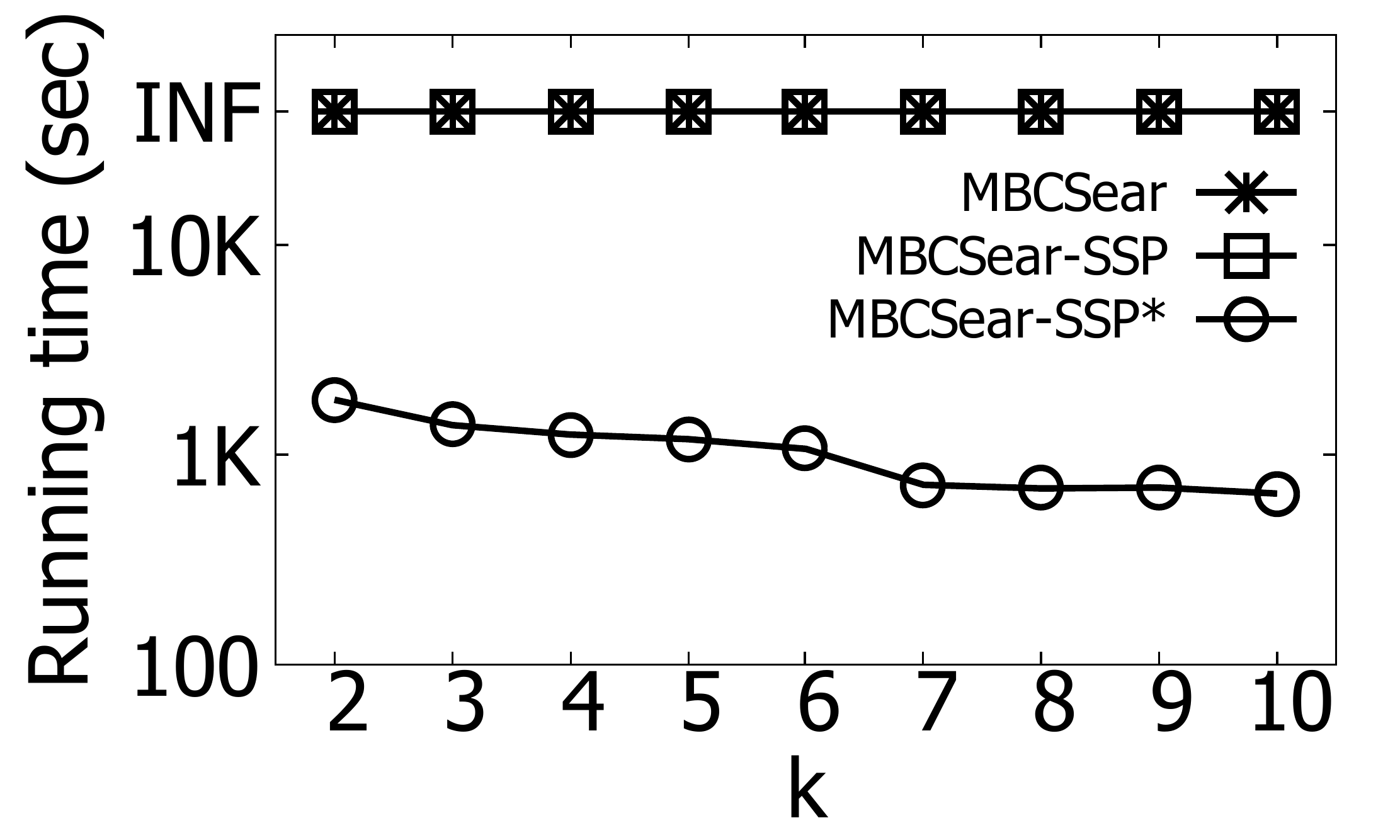}
}
\end{center}
\vspace{-0.4cm}
\topcaption{Running time of \kw{MBCS} algorithms varying $k$}
\label{fig:totaltime2}
\vspace{-0.4cm}
\end{figure}

\stitle{Exp-5: Effectiveness of \kw{MBCS} algorithms when varying $k$.} To intuitively compare the effectiveness of three \kw{MBCS} algorithms,  in this experiment, we record the amount of calculation of three algorithms on eight datasets, $k=[2-5]$. The calculation quantity is the time of invoking \mbcsu and \mbcsup, which can intuitively represent the search space of different algorithms. The experimental results are shown at \reffig{growcnt}.

As shown at \reffig{cnt1}, when $k=2$, on all datasets, the calculation quantity of \mbcssp is much less than that of other algorithms. Meanwhile, \mbcss's calculation quantity is less than \mbcs's. For instance, on DBLP(DB),  the calculation quantity of \mbcs, \mbcss and \mbcssp are 155,621, 328 and 183, respectively. When $k>2$, the trend  is similar. It's because \mbcssp and \mbcss are based on search space partitions, which can reduce the total search space effectively. The experimental results also confirm the reason for the efficiency of \mbcssp at \kw{Exp}-4.

\begin{figure}[t]
%\setlength{\abovecaptionskip}{20pt}
%\captionsetup[subfigure]{aboveskip=-1pt,belowskip=-1pt}
\begin{center}

\subfigure[ $k$=2]{
\label{fig:cnt1}
\centering
\includegraphics[width=0.49\columnwidth]{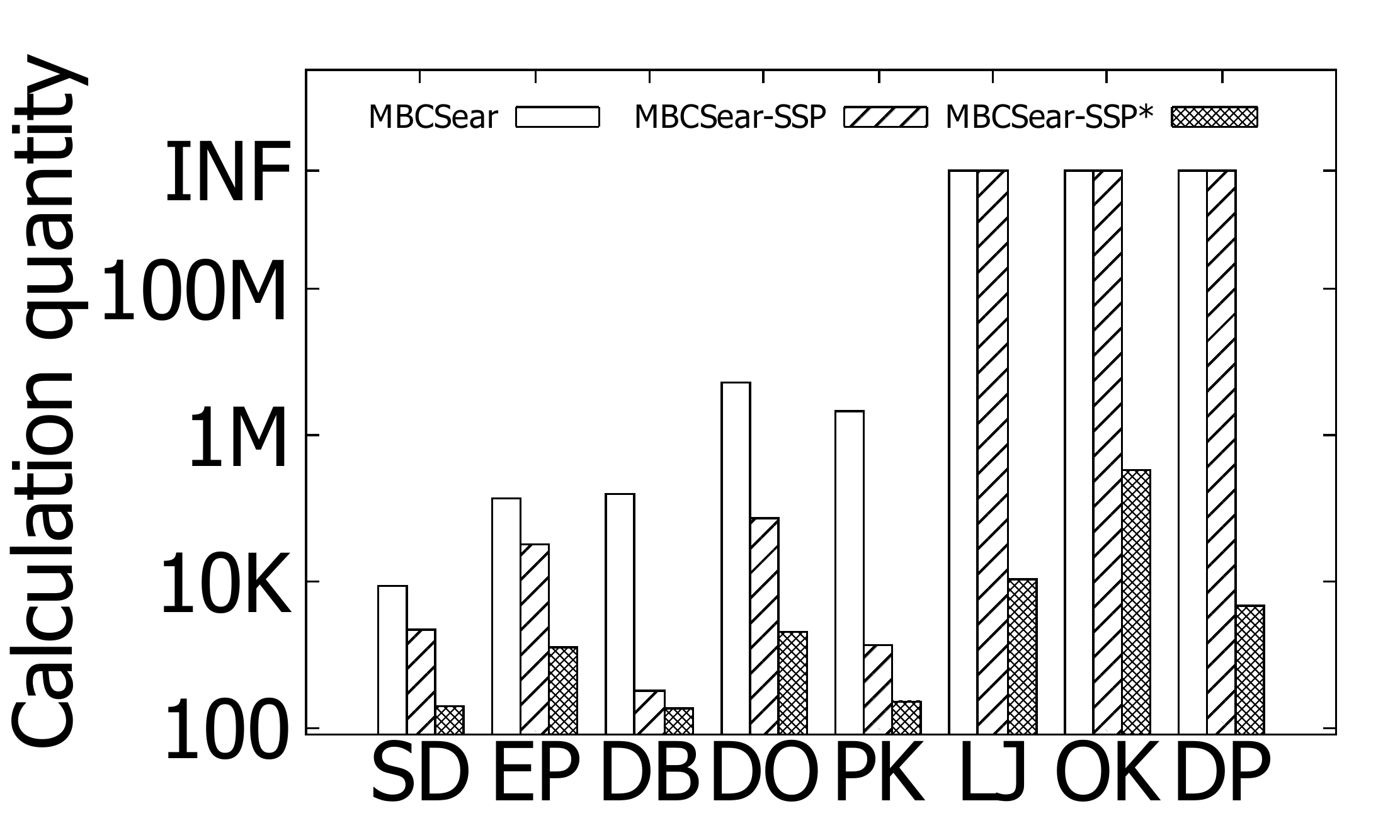}
}\subfigure[ $k$=3]{
\label{fig:cnt2}
\centering
\includegraphics[width=0.49\columnwidth]{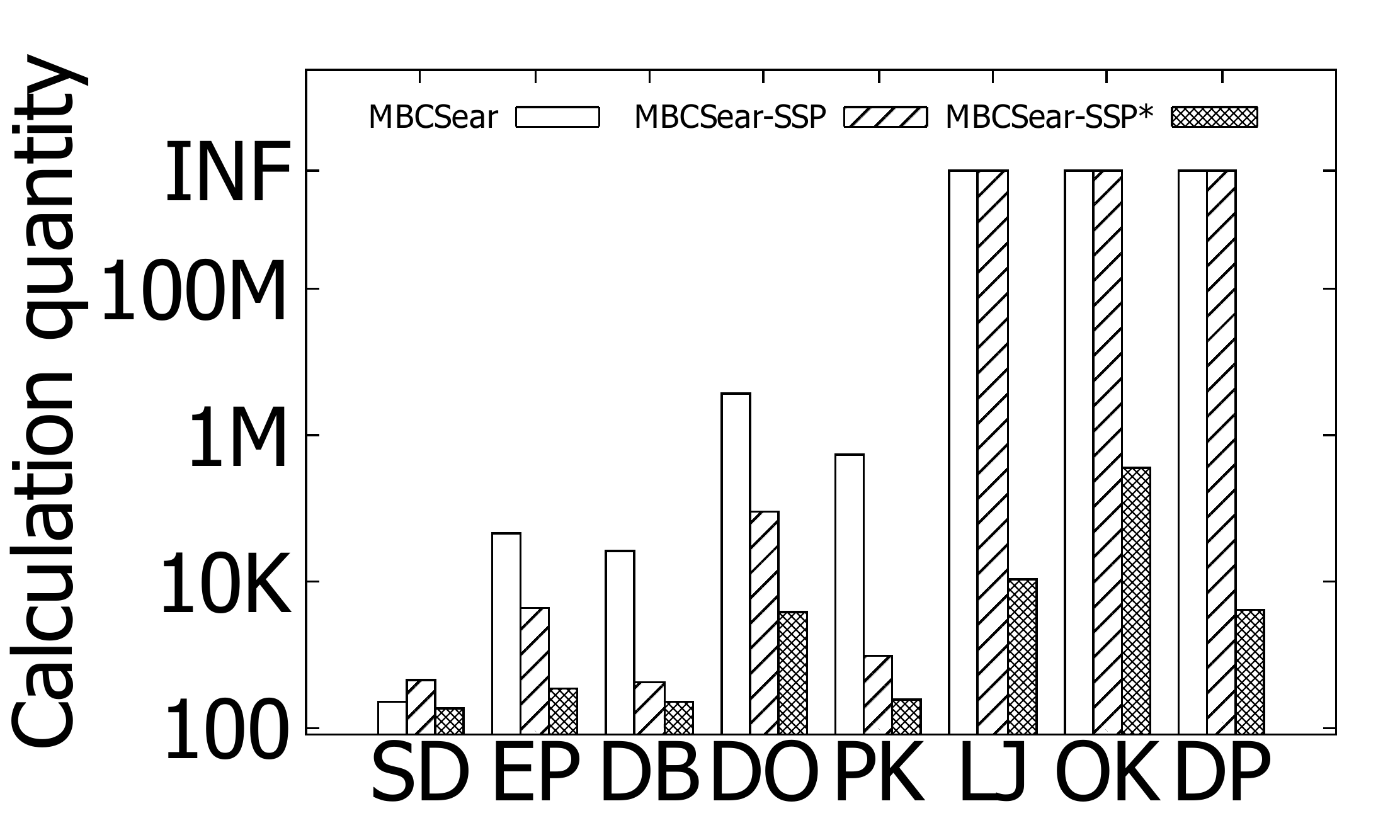}
}\\\vspace{-0.4cm}

\subfigure[ $k$=4]{
\label{fig:time3}
\centering
\includegraphics[width=0.49\columnwidth]{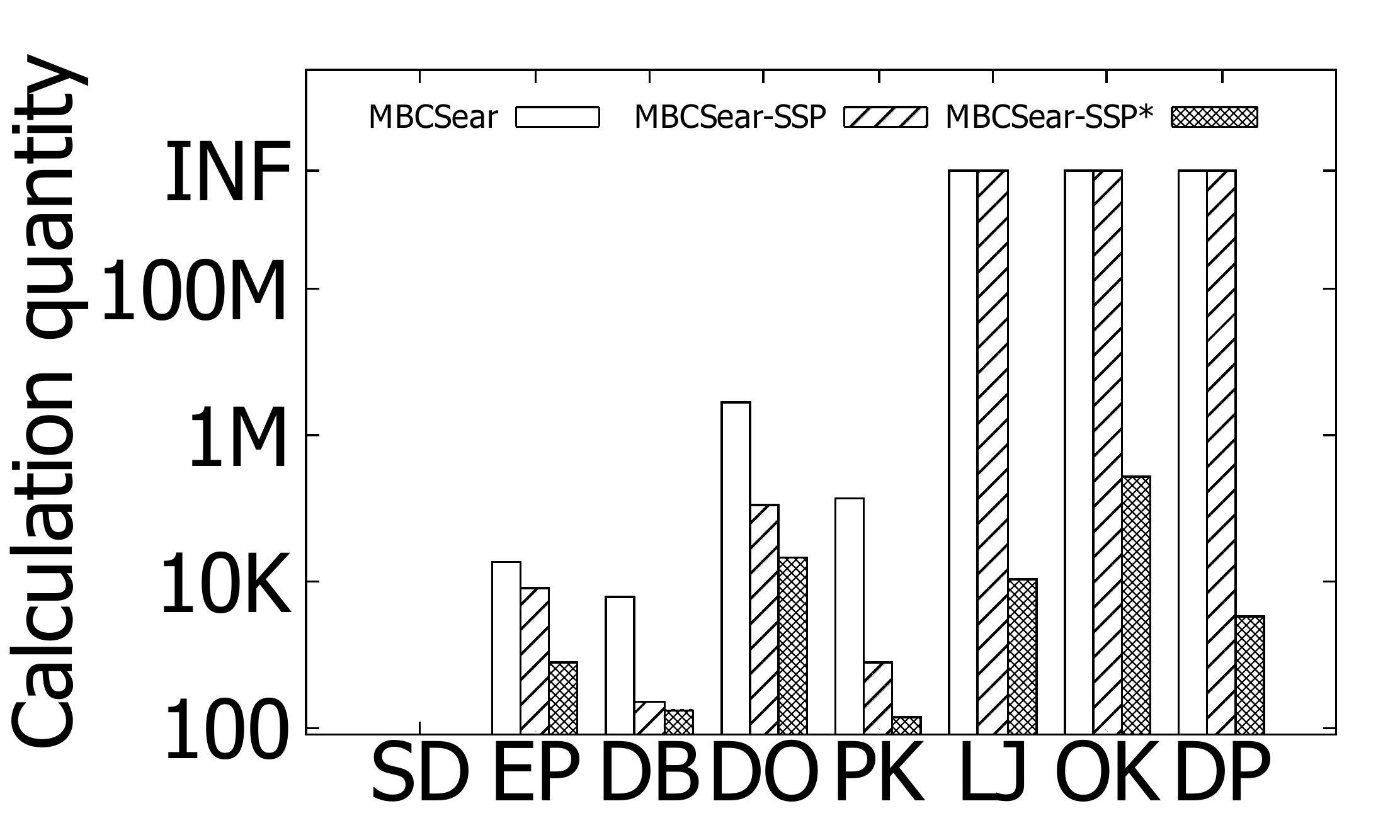}
}\subfigure[ $k$=5]{
\label{fig:time4}
\centering
\includegraphics[width=0.49\columnwidth]{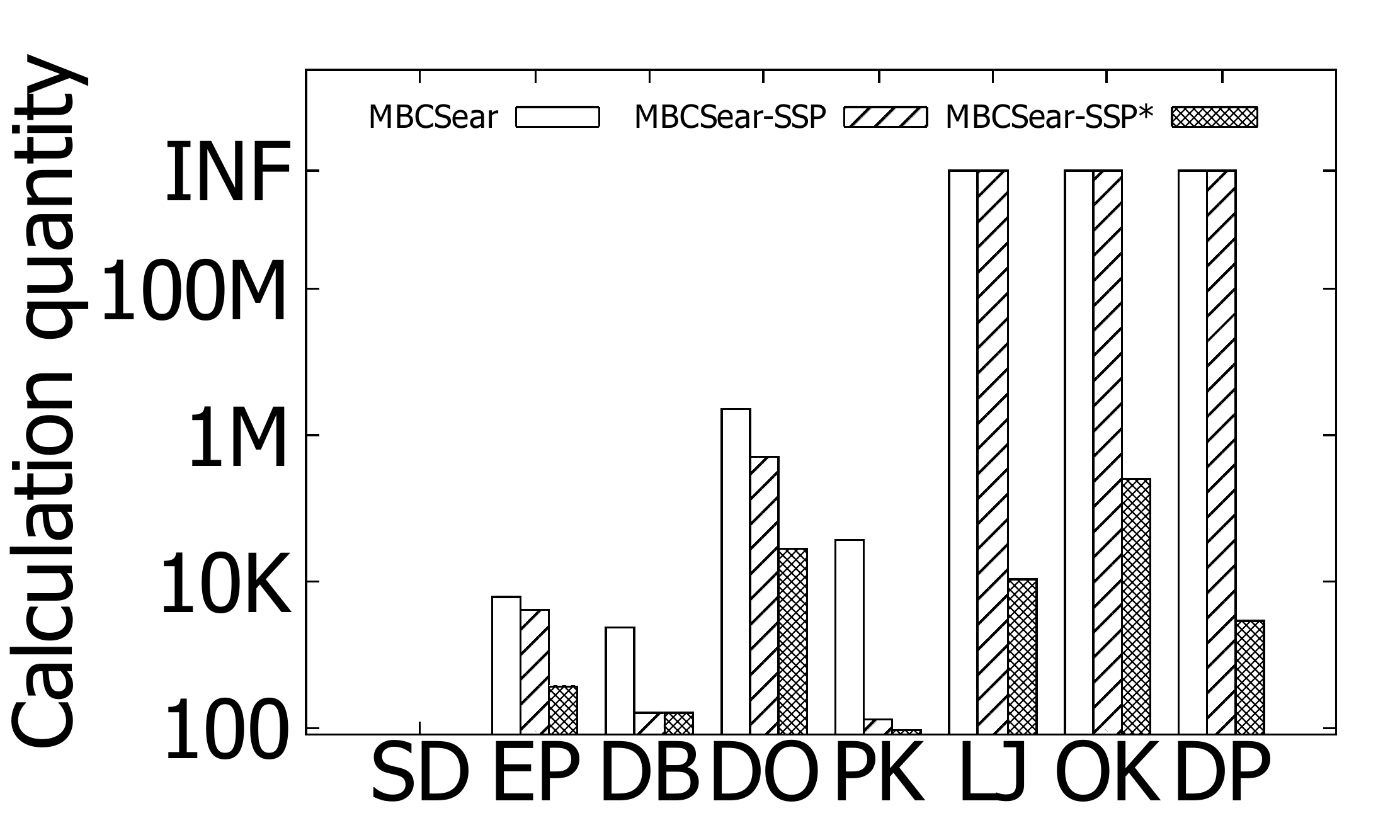}
}\vspace{-0.4cm}

\end{center}
\topcaption{Calculation quantity of different algorithms}
\label{fig:growcnt}
%\vspace{-0.4cm}
\end{figure}

\begin{table}
\caption{The search process of \mbcssp}
\vspace{-0.3cm}
\subtable[\mbcssp on Douban, $k$=2]{
%\topcaption{\mbcssp on Douban, $k$=2}
\label{tab:ss1}
\centering
\def\arraystretch{1.1}
\setlength{\tabcolsep}{0.65em}
{\small
\begin{tabular}{c|c|c|c|c|c}
\hline
 Index&Search Region&$|E^+_{G'}|$ & $|E^-_{G'}|$&$G'/G(\%)$&$\epsilon$ \\%&$d^+_{max}$&$d^-_{max}$\\
\hline
0&(2,78)&0&0&0&4\\
1&(2,39)&0&0&0&4\\
2&(2,20)&21,695&12,190&0.24&33\\
3&(13,13)&0&0&0&33\\
\hline
\end{tabular}
}
\vspace{-0.4cm}
%\caption{\mbcssp on Douban, $k$=2}
%\end{table}
%\end{subtable}
}
%\begin{table}

%\vspace{0.2cm}
\subtable[\mbcssp on Pokec, $k$=2]{
%\topcaption{\mbcssp on Pokec, $k$=2}
\label{tab:ss2}
\centering
\def\arraystretch{1.1}
\setlength{\tabcolsep}{0.7em}
\small
{
\begin{tabular}{c|c|c|c|c|c}
\hline
 Index&Search Region&$|E^+_{G'}|$ & $|E^-_{G'}|$&$G'/G(\%)$&$\epsilon$ \\%&$d^+_{max}$&$d^-_{max}$\\
\hline
0&(2,37)&0&0&0&4\\
1&(2,19)&6,605&2,521&0.30&29\\
2&(10,10)&0&0&0&29\\
\hline
\end{tabular}
}
\vspace{-0.4cm}
%\caption{\mbcssp on Pokec, $k$=2}
%\end{table}
%\end{subtable}
}
\label{tab:ss}
\vspace{-0.4cm}
\end{table}

\stitle{Exp-6: Search process of \mbcssp on real datasets.}
In this experiment, we show the search process of \mbcssp algorithm on Douban and Pokec datasets, $k$=2. Table \ref{tab:ss} shows every search region $(\overline{\kappa},\underline{\kappa})$, the size of its input graph $G'$ including positive edges number $|E^+_{G'}|$ and  negative edges number $|E^-_{G'}|$, the ratio of $G'$ in the original graph $G$, and the maximum balanced clique size $\epsilon$ found so far. Since Douban and Pokec are big datasets, the search process is time consuming, hence, \mbcssp adopts $Dec(\overline{k})=\lceil\frac{\overline{\kappa}}{2}\rceil$. 
%limit the number of search regions.

As Table \ref{tab:ss} shows, on Douban, the first search region $(\underline{\kappa}_0,\overline{\kappa}_0)$ is $(2,78)$, and $\epsilon$ is initialized as 4. \mbcssp first invokes \edgereductionv to pre-reduce useless edges in $G$. Due to the large value of $\overline{\kappa}_0$, all edges are pruned from $G$. Hence, $G'$ is empty here. For the second search region  $(2,39)$, $G'$ is still empty. For the third search region $(2,20)$, $G'$ has 21,695 positive edges and 12,190 negative edges, it only holds 0.24\% edges of the original graph which is much less than $G$. Then, \mbcssp finds the maximum balanced clique $C^*$ on $G'$. The size of  $C^*$ is 33. For the last search region $(13,13)$, $G'$ is empty, the search process is finished. The search process  on Pokec is similar. 
% On Pokec, as shown at Table \ref{tab:ss2}, the first search region is $(2,37)$, $G'$ is empty. For the second search region $(2,19)$, the pruned graph $G'$ has 6,605 positive edges and 2,521 negative edges, which only holds 0.3\% edges of $G$. \mbcssp gets the result on $G'$. The final search region is $(10,10)$, $G'$ is empty, the search process on Pokec is finished.

By observing the search process of \mbcssp on the two datasets, We find two significant phenomena. First, benefited from the search space partition paradigm, the number of search regions is limited. Second, the input graph $G'$ for each search region is much smaller than the original graph, because the edge reduction strategy can remove most of the invalid edges before the search starting. 
%In a result, the search space partition paradigm and the edge reduction strategy can effectively reduce the search space and improve the efficiency of \mbcssp.

%as we reduce the value of $\overline{\kappa}$ by dichotomy, and increase the value of $\underline{\kappa}$ according to the size of the current maximum balanced clique $\epsilon$, so that $\overline{\kappa}$ and $\underline{\kappa}$ can be equal after a limited number of searches.

\stitle{Exp-7: Scalability of \kw{MBCS} algorithms.} In this experiment, we evaluate the scalability of \kw{MBCS} algorithms on two biggest datasets Orkut and Dbpedia as Exp-2. The results are shown at \reffig{scal-mbcs}.

 As \reffig{scal-mbcs} shows, with the number of vertices increases, the running time of three algorithms increases as well. Among them, the growth rate of \mbcssp is the most stable. For instance, on Dbpedia with 40\% vertices, \mbcs cannot get result within a reasonable time, the running time of \mbcss and \mbcssp are 11953.0s and 319.9s, respectively. On Dbpedia with more than 40\% vertices, only \mbcssp can get result within a reasonable time. The trend of running time on Orkut is similar. Therefore, \mbcssp can scale to large-scale graphs.

%We randomly sample 20\%, 40\%, 60\% and 80\% vertices from the original graph to produce 4 sample graphs, then run \kw{MBCS} algorithms on the corresponding produced graphs and the original graph, the results are shown at \reffig{scal-mbcs}.

%As \reffig{scal-mbcs} shows, with the number of vertices increases, the running time of three algorithms increases as well. Among them, the growth rate of \mbcssp is the most stable. 
%For instance, on Dbpedia with 20\% vertices, the running time of three algorithms is almost the same. However, on Dbpedia with 40\% vertices, \mbcs cannot get result within a reasonable time, the running time of \mbcss and \mbcssp is 11953.0s and 319.9s, respectively. On other graphs with more than 40\% vertices, \mbcs and \mbcss cannot get result within a reasonable time. On Orkut, the trend of running time is similar. 

\begin{figure}[t]
%\setlength{\abovecaptionskip}{20pt}
%\captionsetup[subfigure]{aboveskip=-1pt,belowskip=-1pt}
\begin{center}
%\subfigure[\normalsize Livejournal]{
%
%\label{fig:scal11}
%\centering
%\includegraphics[width=0.49\columnwidth]{MaximumBC/fig/scal_livejournal.pdf}
%}
\subfigure[\small Orkut (Vary $n$)]{
\label{fig:scal22}
\centering
\includegraphics[width=0.49\columnwidth]{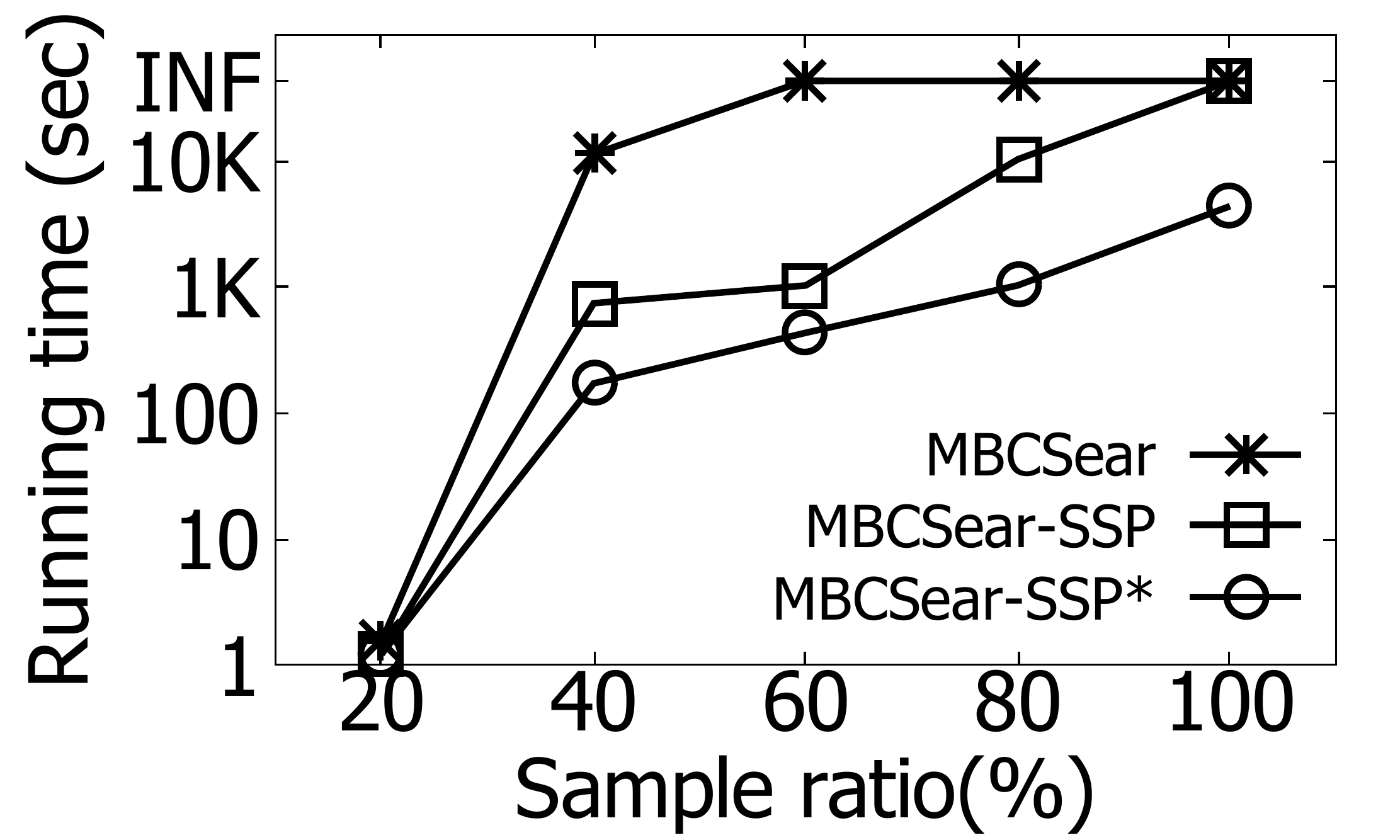}
}\vspace{-0.4cm}\subfigure[\small Dbpedia (Vary $n$)]{
\label{fig:scal33}
\centering
\includegraphics[width=0.49\columnwidth]{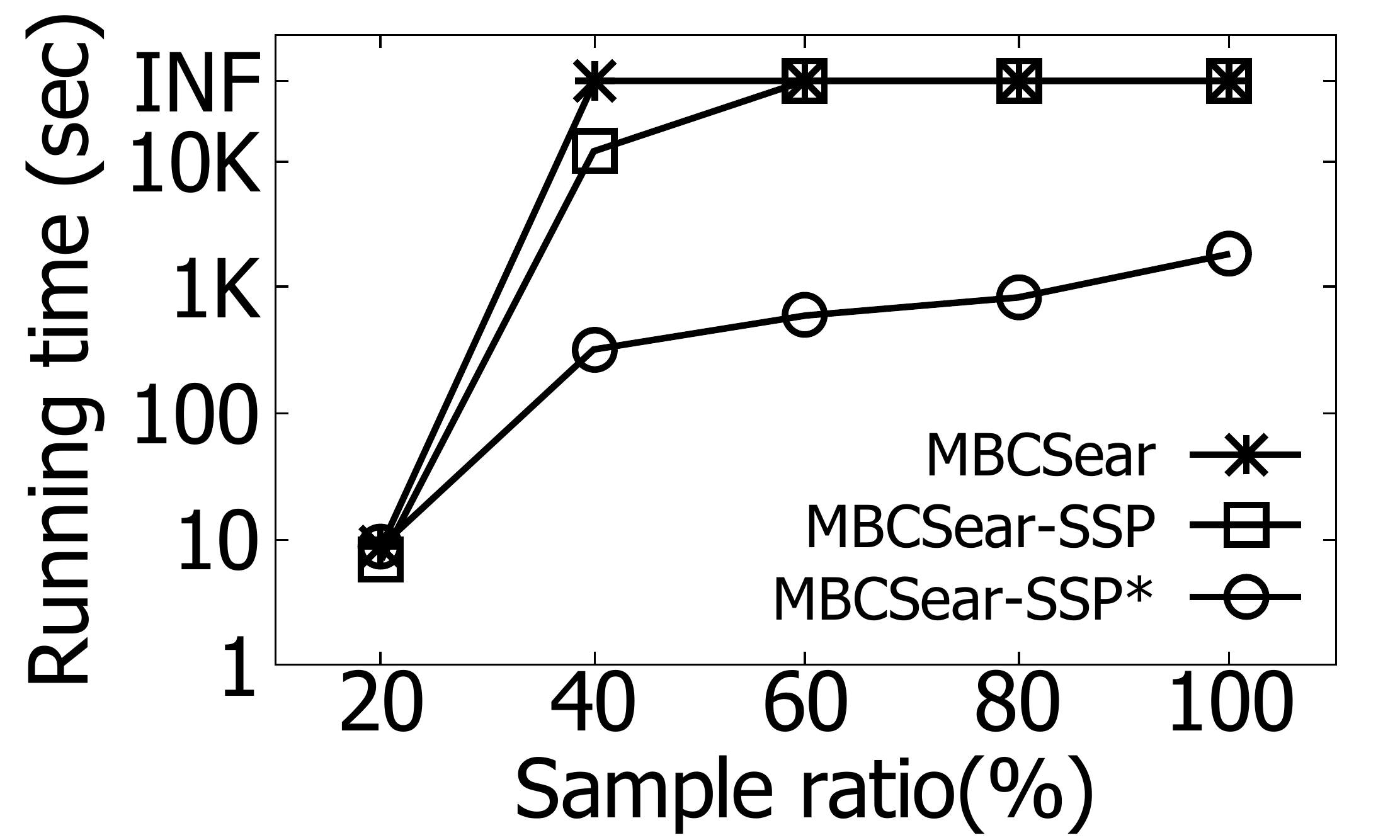}
}\\

\end{center}
\topcaption{Scalability of \kw{MBCS} algorithms, $k$=2}
\label{fig:scal-mbcs}
\vspace{-0.4cm}
\end{figure}

%\vspace{-0.4cm}

\section{Conclusions}
\label{sec:conclusion}
 In this paper, we study the maximal balanced clique enumeration problem in signed networks.  We propose a new enumeration algorithm tailored for signed networks. Based on the new enumeration algorithm, we explore two optimization strategies to further improve the efficiency of the enumeration algorithm. Besides, we study the maximum balanced clique search problem, and propose a novel search space partition-based search framework. Moreover, we explore multiple optimization strategies to further reduce the search space during search process. The experimental results on real datasets demonstrate the efficiency, effectiveness and scalability of our solutions.

% \stitle{Acknowledge.} Long Yuan is supported by  NSFC61902184 and NSF of Jiangsu Province BK20190453.   Xuemin Lin  is supported by 2018YFB1003504, NSFC61232006, ARC DP180103096 and DP170101628. Lu Qin is supported by ARC DP160101513.

%\balance
%% use section* for acknowledgment
%\ifCLASSOPTIONcompsoc
%  % The Computer Society usually uses the plural form
%\else
%  % regular IEEE prefers the singular form
%\fi

%\ifCLASSOPTIONcaptionsoff
%  \newpage
%\fi

%%\begin{acks}
% Zi Chen is supported by China Postdoctoral Science Foundation 2021M701214.
% Long Yuan is supported by  NSFC61902184, NSF of Jiangsu Province BK20190453, and  Science and Technology on Information Systems Engineering Laboratory WDZC20205250411. Li Han is supported by Shanghai Sailing Program 21YF1411100.
%%\end{acks}

%\newpage
\bibliographystyle{IEEEtran}
\bibliography{reference}
\vspace{-1.8cm}

\end{document}